\pdfoutput=1

\documentclass[preprint,3p]{elsarticle}

\usepackage{lineno,hyperref}
\usepackage{bm}
\usepackage{amsmath}
\usepackage{color}
\usepackage{algorithm}
\usepackage{algorithmic}
\usepackage{multirow}
\usepackage{mathtools}
\modulolinenumbers[5]
\usepackage{booktabs}

\usepackage{amssymb}

\usepackage{enumitem}

\usepackage{subfigure}
\usepackage{array}
\newcolumntype{C}[1]{>{\centering\arraybackslash}p{#1}}

\usepackage{xcolor}

\usepackage{bbding}
\usepackage{graphicx,import}

\newtheorem{theorem}{Theorem}[section]

\newenvironment{proof}[1][Proof]{\begin{trivlist}
\item[\hskip \labelsep {\bfseries #1}]}{\end{trivlist}}

\newenvironment{remark}[1][Remark]{\begin{trivlist}
\item[\hskip \labelsep {\bfseries #1}]}{\end{trivlist}}
\newenvironment{notation}[1][Notation]{\begin{trivlist}
\item[\hskip \labelsep {\bfseries #1}]}{\end{trivlist}}

\journal{Journal}

\definecolor{darkgreen}{rgb}{0,0.5,0}

\newcommand{\Div}[1]{\nabla \cdot {#1}}
\newcommand{\Grad}[1]{\nabla {#1}}

\newcommand{\boundary}[1]{\Gamma^{\mathrm{#1}}}
\newcommand{\hboundary}[1]{\Gamma_{h}^{\mathrm{#1}}}

\usepackage{stmaryrd} 

\newcommand{\avg}[1]{\{\!\{#1\}\!\}}

\newcommand{\jump}[1]{\llbracket {#1} \rrbracket }
\newcommand{\jumporiented}[1]{\left[ {#1} \right] }

\newcommand{\intdomain}[2]{ \left( {#1},{#2} \right)_{\Omega_{h}} }
\newcommand{\intinteriorfaces}[2]{ \left( {#1},{#2} \right)_{\Gamma_{h}^{\mathrm{int}} }}

\newcommand{\intele}[2]{ \left( {#1},{#2} \right)_{\Omega_{e}} }
\newcommand{\intelenp}[2]{ \left( {#1},{#2} \right)_{\Omega_{e}^{n+1}} }
\newcommand{\inteleface}[2]{ \left( {#1},{#2} \right)_{\partial\Omega_{e}} }
\newcommand{\intelefaceInterior}[2]{ \left( {#1},{#2} \right)_{\partial\Omega_{e}\setminus\Gamma_h }}

\newcommand{\intelefaceDirichlet}[2]
{ {\left( {#1},{#2} \right)}_{\partial\Omega_{e}\cap \Gamma^{\rm{D}}_{h}}}
\newcommand{\intelefaceNeumann}[2]{ {\left( {#1},{#2} \right)}_{\partial\Omega_{e}\cap \Gamma^{\rm{N}}_{h}}}

\usepackage{tikz, tkz-euclide}%
\usetikzlibrary{spy,calc}%

\newcommand{\plotgridsine}[1]{
	\begin{tikzpicture}[scale=0.30\textwidth/1.25cm]
	\def\width{1.0}
	\def\left{-0.5}
	\def\right{0.5}
	\def\A{#1}
	\foreach \y in {-0.5,-0.375,-0.25,-0.125,0,0.125,0.25,0.375,0.5}%
	{
		\draw[variable=\x,domain=\left:\right,samples=40] plot ({\x+sin(deg(2*pi*(\y-\left)/\width))*\A},{\y+sin(deg(2*pi*(\x-\left)/\width))*\A},0);
		\draw[variable=\x,domain=\left:\right,samples=40] plot ({\y+sin(deg(2*pi*(\x-\left)/\width))*\A},{\x+sin(deg(2*pi*(\y-\left)/\width))*\A},0);
	}
	\draw[draw=white,ultra thin] (-0.625,-0.625) rectangle (0.625,0.625);
	\end{tikzpicture}
}

\begin{document}

\begin{frontmatter}

\title{High-order arbitrary Lagrangian--Eulerian discontinuous Galerkin\\ methods for the incompressible Navier--Stokes equations}

\author{Niklas Fehn\corref{correspondingauthor1}}
\cortext[correspondingauthor1]{Corresponding author at: Institute for Computational Mechanics, Technical University of Munich, Boltzmannstr. 15, 85748 Garching, Germany. Tel.: +49 89 28915300; fax: +49 89 28915301}
\ead{fehn@lnm.mw.tum.de}
\author{Johannes Heinz\corref{correspondingauthor2}}
\ead{johannes.heinz@tum.de}
\author{Wolfgang A. Wall}
\ead{wall@lnm.mw.tum.de}
\author{Martin Kronbichler}
\ead{kronbichler@lnm.mw.tum.de}
\address{Institute for Computational Mechanics, Technical University of Munich,\\ Boltzmannstr. 15, 85748 Garching, Germany}

\begin{abstract}
This paper presents robust discontinuous Galerkin methods for the incompressible Navier--Stokes equations on moving meshes. High-order accurate arbitrary Lagrangian--Eulerian formulations are proposed in a unified framework for both monolithic as well as projection or splitting-type Navier--Stokes solvers. The framework is flexible, allows implicit and explicit formulations of the convective term, and adaptive time-stepping. The Navier--Stokes equations with ALE transport term are solved on the deformed geometry storing one instance of the mesh that is updated from one time step to the next. Discretization in space is applied to the time discrete equations so that all weak forms and mass matrices are evaluated at the end of the current time step. This design ensures that the proposed formulations fulfill the geometric conservation law automatically, as is shown theoretically and demonstrated numerically by the example of the free-stream preservation test. We discuss the peculiarities related to the imposition of boundary conditions in intermediate steps of projection-type methods and the ingredients needed to preserve high-order accuracy. We show numerically that the formulations proposed in this work maintain the formal order of accuracy of the Navier--Stokes solvers. Moreover, we demonstrate robustness and accuracy for under-resolved turbulent flows.
\end{abstract}


\end{frontmatter}

\section{Introduction}\label{Intro}

The arbitrary Lagrangian--Eulerian (ALE) continuum mechanics description is the basis of many methods to capture flow problems on deforming domains. A very prominent class of applications are fluid--structure interaction problems with moderate deformations of the structure, where moderate means that a mesh moving or mesh smoothing algorithm is able to handle the mesh deformation of the fluid mesh following the deformations imposed at the fluid--structure interface, as opposed to very large deformations and topological changes that require other, geometrically more flexible techniques. ALE methods have a long tradition and have first been developed for finite difference methods~\cite{Hirt1974}, see also~\cite{Donea1982} for a review on the early development of this methodology, as well as~\cite{Donea2017} for a survey of ALE methods. It has later been developed for finite element discretizations of the compressible Navier--Stokes equations in~\cite{Donea1977,Donea1982}, of the incompressible Navier--Stokes equations using linear elements in~\cite{Hughes1981} and spectral element discretizations in~\cite{Beskok2001}, and also for finite volume discretizations, see for example~\cite{Lesoinne1996}. In the context of discontinuous Galerkin (DG) discretizations, this technique has first been applied to the compressible Navier--Stokes equations being solved on the deforming domain~\cite{Lomtev1999,Nguyen2010,Mavriplis2011}, or solving transformed equations on a reference domain~\cite{Persson2009,Schnucke2018arxiv}. For the incompressible Navier--Stokes equations, the development of ALE-DG methods lagged somewhat behind as detailed below. An ALE method satisfying the geometric conservation law (GCL)~\cite{Thomas1979} is able to preserve a constant flow state on moving meshes. The GCL has extensively been discussed in the context of finite volume discretizations of the compressible Navier--Stokes equations, see for example the review article~\cite{Farhat2004} and references therein. Here, we especially refer to two works~\cite{Foerster2006,Etienne2009} addressing the solution of incompressible flow problems and being particularly relevant for the present work. Following the design described in~\cite{Foerster2006} one can easily construct incompressible flow solvers that automatically fulfill the geometric conservation law, and we therefore do not discuss this topic at length in the present work.

Before discussing ALE-DG methods for incompressible flows, let us first summarize the key issue of incompressible Navier--Stokes DG solvers in the Eulerian case with static meshes. The main problem originates from the nonlinearity of the convective term in case the numerical velocity field is not pointwise divergence-free like with standard~$L^2$-conforming spaces. In this circumstance, the energy stability derived e.g.~for an upwind flux with linear transport terms of constant speed is lost. Special DG discretizations that are exactly mass conserving and energy stable have been discussed in early mathematical literature on DG methods for the incompressible Navier--Stokes equations, see~\cite{Cockburn2005, Cockburn2007, Cockburn2009}. However, the importance of this aspect and its relevance for practical problems have long not been realized in application-oriented DG literature, where no attempts have been made to fulfill these properties~\cite{Bassi2006, Shahbazi07, Hesthaven07, Botti11, Ferrer11, Klein13}. It was later found in~\cite{Steinmoeller13,Joshi16,Krank17,Fehn18a} that compliance with (i)~the divergence-free constraint and (ii)~inter-element continuity of the normal velocity is crucial in obtaining robust solvers for engineering applications where the solution is under-resolved, such as turbulent flows. Methods explicitly addressing these two requirements can be categorized into two groups, those fulfilling them exactly for example by choosing appropriate function spaces such as~$H(\mathrm{div})$-conforming spaces where the divergence of the velocity space lies in the pressure space~\cite{Cockburn2005, Cockburn2007, Cockburn2009, Lehrenfeld2016, Piatkowski2018, Schroeder2018, Rhebergen2018}, and those fulfilling them weakly by appropriate stabilization terms~\cite{Joshi16, Krank17, Fehn18a, Schroeder2017, Akbas2018}. These constraints on the velocity solution can be imposed in form of a postprocessing step or in an inbuilt/monolithic way, independently of the two categories a method belongs to. A comparative study of the two categories has been shown recently in~\cite{Fehn19Hdiv} in the context of under-resolved turbulent flows.

ALE formulations for the incompressible Navier--Stokes equations using DG discretizations have been presented in~\cite{Ferrer12b,Wang2018}, with the mesh motion being restricted to rigid body rotations without mesh deformation in~\cite{Ferrer12b}. Both methods are based on the dual splitting projection scheme proposed by~\cite{Orszag1986,Karniadakis1991} and use a discontinuous Galerkin discretization of the velocity--pressure coupling terms without integration by parts, originating from the method proposed in~\cite{Hesthaven07}, that has been shown to be unstable for small time step sizes by independent studies, see~\cite{Fehn17,Xu2019}. Both works~\cite{Ferrer12b,Wang2018} use equal-order poylnomials for velocity and pressure, but it was shown in~\cite{Fehn17} that the dual splitting scheme per se is not inf--sup stable as is sometimes believed. Morever, the works~\cite{Ferrer12b,Wang2018} enforce the important aspect of~$H(\mathrm{div})$-conformity discussed above neither exactly nor weakly, so that the robustness of these methods remains questionable. These works also do not comment on the fulfillment of the geometric conservation law theoretically or by numerical experiments. Second-order convergence in time is shown in~\cite{Ferrer12b} for a rotating, non-deforming mesh with fixed boundaries, hence, not being representative of a deforming fluid domain or a fluid--structure interaction problem. A temporal convergence test is shown in~\cite{Wang2018} for a full FSI problem by comparing the error against the solution for the smallest time step size and second-order accuracy is shown. It remains unclear from these works whether third order accuracy can be achieved for the dual splitting scheme in the ALE case, which has for example been shown in~\cite{Krank17} for a DG method solving the Eulerian form of the equations. Open questions remain therefore from previous works~\cite{Ferrer12b,Wang2018} on ALE-DG methods for the incompressible Navier--Stokes equations regarding how these methods are implemented exactly and questions about their numerical properties in terms of stability (small time steps, inf--sup problem, and under-resolved turbulence or energy stability), fulfillment of the geometric conservation law, and temporal convergence rates. 

More sophisticated convergence tests are presented in~\cite{Rhebergen2012} for a space--time HDG method on deforming domains, where it is argued that the reasons for choosing the space--time approach are automatically satifying the geometric conservation law and achieving arbitrarily high order in space and time. An energy-stable version of such a space--time HDG approach has been proposed recently in~\cite{Horvath2019} achieved through a velocity field that is~$H(\mathrm{div})$-conforming and exactly divergence-free. As shown in the present work, fulfillment of the GCL and discretizations that formally exhibit arbitrarily high order of accuracy in space can also be achieved with a classical method-of-lines approach. Achieving arbitrarily high order of accuracy in time with projection-type Navier--Stokes solvers is non-trivial, irrespective of the grid motion. However, our experience is that second or third-order time integration schemes are sufficient in terms of accuracy for practical problems, especially if the time step size is restricted according to the CFL condition when treating the convective term explicitly in time, which is the state-of-the-art solution technique used by some of the most sophisticated and computationally efficient high-order CFD solvers, such as Nektar++\cite{Nektar++} and Nek5000~\cite{Nek5000}. We emphasize that formal orders of accuracy describe the optimal behavior observable only in the asymptotic regime for sufficiently smooth solution. In this context, it should be mentioned that high order of convergence in space is rarely observed for practical problems. The dissipation/dispersion properties and the associated improved resolution capabilities of high-order methods rather than theoretical rates of convergence motivate the use of high-order methods for application-relevant, turbulent flows~\cite{Gassner2013,Moura2017,Fehn18b}. When it comes to the aspect of computational costs, it has not yet been demonstrated that space--time approaches can keep up with the fast solution techniques for incompressible flows mentioned above. A similar argument holds for the computational efficiency of matrix-based HDG solvers, which significantly trail behind fast matrix-free DG implementations on modern CPU hardware, as shown by a recent study~\cite{Kronbichler2018}.

For these reasons, our goal is to develop ALE-DG methods based on the method-of-lines approach that combine computationally efficient matrix-free DG implementations~\cite{Kronbichler2019fast} and fast Navier--Stokes solution algorithms~\cite{Fehn17,Fehn18a} with desirable discretization properties in terms of optimal convergence rates in time and space, the geometric conservation law, and robustness for turbulent flows. To develop algorithms as simple as possible, we make the following design choices: We solve the ALE equations on the deformed geometry, storing one instance of the mesh that is updated from one time step to the next by updating the coordinates of all nodal points. We introduce the ALE equations on the level of differential equations, subsequently discretized in time and space. This way, it is straight-forward to satisfy the geometric conservation law automatically~\cite{Foerster2006}, i.e., independently of the mesh motion and how the mesh velocity is computed numerically. We present a unified framework for both monolithic solvers and widely used projection-type solvers. The formulation is flexible regarding implicit versus explicit formulations of the convective term, and the framework naturally includes the option for adaptive time-stepping. Although we consider analytical mesh motions in the present work, the methods are formulated with fluid--structure interaction problems in mind, i.e., the ALE formulations are implemented in a way that they only require knowledge about the coordinates of all grid nodes at discrete instances of time, and the grid velocity is computed from these grid coordinates in a way that the formal order of accuracy of the time integration schemes is maintained on moving meshes.

The outline of this article is as follows. We derive the ALE form of the incompressible Navier--Stokes equations in Section~\ref{sec:MathematicalModel}. Aspects related to the temporal discretization are discussed in Section~\ref{sec:TemporalDiscretization}, and the spatial discretization is subject of Section~\ref{sec:SpatialDiscretization}. Here, our focus lies on the aspects relevant to ALE, with the goal to provide a comprehensive and simple formulation that can easily be included in existing flow solvers. Numerical results are presented in Section~\ref{sec:NumericalResults}, and we summarize our results in Section~\ref{sec:Conclusion}.

\section{Incompressible Navier--Stokes equations in arbitrary Lagrangian--Eulerian formulation}\label{sec:MathematicalModel}
We consider the incompressible Navier--Stokes equations in a domain~$\Omega \subset \mathbb{R}^d$
\begin{align}
\left.\frac{\partial \bm{u}}{\partial t}\right\vert_{\bm{X}} - \nabla \cdot \bm{F}_{\mathrm{v}} (\bm{u}) + \nabla p &= \bm{f} \; ,\label{eq:MomentumEquationMaterial}\\
\nabla \cdot \bm{u} &= 0 \; ,\label{eq:ContinuityEquation}
\end{align}
where~$\bm{u}=(u_1,...,u_d)^{\mathsf{T}}$ is the velocity vector and~$p$ the kinematic pressure. The body force vector is denoted by~$\bm{f} = (f_1,...,f_d)^{\mathsf{T}}$. Spatial derivatives are defined w.r.t. the Eulerian coordinates~$\bm{x}=(x_1,...,x_d)^{\mathsf{T}}$,~$\nabla{\phi} = \frac{\partial \phi}{\partial \bm{x}}$, and~$\left.\right\vert_{\bm{X}}$ denotes the material time derivative, i.e., the total time derivative along the trajectory of a material point~$\bm{X}$ of the fluid. The viscous term is written in Laplace formulation~$\bm{F}_{\mathrm{v}}(\bm{u})=\nu \nabla \bm{u}$ with the constant kinematic viscosity~$\nu$.
\\
To obtain the ALE form of the above equations, the material time derivative (or the Eulerian time derivative when considering the incompressible Navier--Stokes equations written in Eulerian coordinates) is replaced by a time derivative with respect to a fixed point of the mesh (denoted as ALE time derivative in the following), which gives rise to an additional transport term with transport by the grid velocity~$\bm{u}_{\mathrm{G}}$. However, this transport term has the same structure as the convective term in Eulerian description, so that the same implementation with transport velocity~$\bm{w} = \bm{u}-\bm{u}_{\mathrm{G}}$ instead of the fluid velocity~$\bm{u}$ can be used. The motivation behind is to apply the time integration scheme with an update of the solution vectors just as in the Eulerian case, thereby automatically obtaining the solution coefficients on the new mesh. By this technique, expensive transformations of the solution vector (containing the degrees of freedom of the finite element expansion) from one mesh at a previous time instant onto another one at the current time instant is avoided. Apart from the Eulerian coordinates~$\bm{x}$ describing an (arbitrary) point in Euclidean space and the material coordinates~$\bm{X}$ describing a fixed fluid element, we introduce the mesh coordinates~$\bm{\chi}$ describing a fixed point of the mesh.
For transient problems,~$\bm{x}(\bm{X},t)$ and~$\bm{x}(\bm{\chi},t)$ describe the trajectories of a fixed fluid element or a fixed point of the mesh in the Eulerian coordinates~$\bm{x}$ as a function of time. The fluid and mesh velocity are therefore given as
\begin{align*}
\bm{u} = \left. \frac{\partial \bm{x}}{\partial t} \right\vert_{\bm{X}} , \;  \bm{u}_{\mathrm{G}} = \left. \frac{\partial \bm{x}}{\partial t} \right\vert_{\bm{\chi}} .
\end{align*}
The material time derivative of an arbitrary quantity~$\phi$ is given as
\begin{align}
\left. \frac{\partial \phi\left(\bm{x}(\bm{X},t),t\right)}{\partial t} \right\vert_{\bm{X}} = \left. \frac{\partial \phi}{\partial t} \right\vert_{\bm{x}} +  \frac{\partial \phi}{\partial \bm{x}} \cdot \underbrace{\left.\frac{\partial \bm{x}}{\partial t}\right\vert_{\bm{X}}}_{=u} = \left. \frac{\partial \phi}{\partial t} \right\vert_{\bm{x}} +   \left(u\cdot \nabla \right) \phi \; . \label{eq:MaterialToEulerian}
\end{align}
The same relation can be stated for the mesh reference frame~$\bm{\chi}$ to obtain the desired relation between the Eulerian and ALE time derivatives
\begin{align}
\left. \frac{\partial \phi\left(\bm{x}(\bm{\chi},t),t\right)}{\partial t} \right\vert_{\bm{\chi}} = \left. \frac{\partial \phi}{\partial t} \right\vert_{\bm{x}} +  \frac{\partial \phi}{\partial \bm{x}} \cdot \underbrace{\left.\frac{\partial \bm{x}}{\partial t}\right\vert_{\bm{\chi}}}_{=\bm{u}_{\mathrm{G}}} = \left. \frac{\partial \phi}{\partial t} \right\vert_{\bm{x}} +   \left(\bm{u}_{\mathrm{G}}\cdot \nabla \right) \phi \; .\label{eq:ALEToEulerian}
\end{align}
Inserting equations~\eqref{eq:MaterialToEulerian} and~\eqref{eq:ALEToEulerian} into equation~\eqref{eq:MomentumEquationMaterial}, we arrive at the incompressible Navier--Stokes equations in ALE formulation
\begin{align}
\left.\frac{\partial \bm{u}}{\partial t} \right\vert_{\bm{\chi}}  + ((\bm{u}-\bm{u}_{\mathrm{G}})\cdot \nabla )\bm{u} - \nabla \cdot \bm{F}_{\mathrm{v}} (\bm{u}) + \nabla p &= \bm{f} \;\; \text{in}\; \Omega(t)  \; ,\label{eq:MomentumEquationALE}\\
\nabla \cdot \bm{u} &= 0 \;\; \text{in}\; \Omega(t) \; .\label{eq:ContinuityEquationALE}
\end{align}

\begin{remark}
Note that the formulation chosen as a starting point for discretization in time and space has important implications regarding compliance with the geometric conservation law. Using the above differential formulation with the convective term written in non-conservative form to derive temporal and spatial discretization allows to satisfy the GCL automatically~\cite{Foerster2006}. Alternative conservative formulations with time derivative in front of the integral over a temporally changing domain contain an additional term in which the divergence of the mesh velocity occurs, and fulfilling the GCL is more complicated in this case~\cite{Foerster2006,Etienne2009}. We mention that the study~\cite{Etienne2009} is inconclusive in the sense that the non-conservative formulation is not restricted to first-order accuracy in time as implied in that work. As demonstrated in~\cite{Foerster2006} and in the present work, high-order accuracy in time can be achieved with the non-conservative formulation. 
\end{remark}

\begin{figure}[!ht]
\centering
  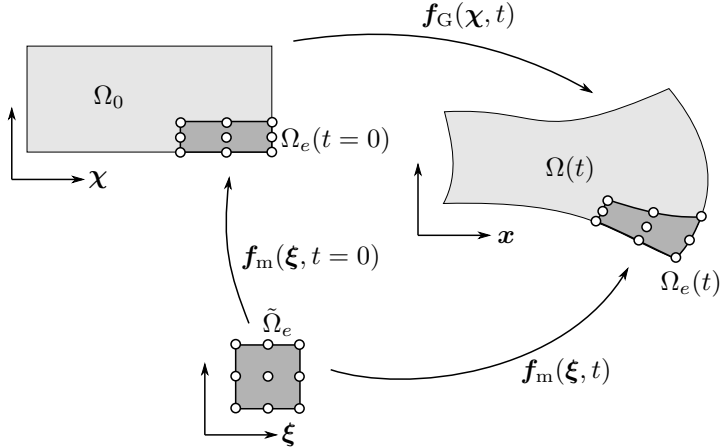
\caption{Illustration of coordinate systems~$\boldsymbol{\chi}$ and~$\boldsymbol{x}$ with mesh deformation~$\bm{f}_{\mathrm{G}}$, and reference coordinates~$\boldsymbol{\xi}$ with finite element mapping~$\bm{f}_{\mathrm{m}}$ of polynomial degree~$k_{\mathrm{m}}=2$.}
\label{fig:DeformationAndMapping}
\end{figure}

The motion of the domain~$\Omega(t)$ is described by a function~$\bm{f}_{\mathrm{G}}=\bm{f}_{\mathrm{G}}(\bm{\chi},t)$
\begin{align*}
\bm{f}_{\mathrm{G}}: \begin{cases}
\Omega_0 \times \left[0,T\right]\rightarrow \Omega(t), \, \Omega_0, \Omega(t) \subset \mathbb{R}^d\, ,\\
\left(\bm{\chi},t\right)\mapsto \bm{x}\left(\bm{\chi}, t\right)\, .
\end{cases}
\end{align*}
With respect to the argument~$\bm{\chi}$, the map~$\bm{f}_{\mathrm{G}}$ is a homeomorphism for all times, and the argument~$t$ describes a continuous deformation over time. In the context of this work and for the numerical results shown below,~$\bm{f}_{\mathrm{G}}$ will be an analytically defined, smooth function in space and time. An illustration is shown in Figure~\ref{fig:DeformationAndMapping}. Without loss of generality, we assume~$\bm{x}\left(\bm{\chi}, t=0\right)=\bm{\chi}$, and therefore~$\Omega(t=0) = \Omega_0$. To demonstrate high order of accuracy of the multistep BDF time integration schemes that are used in this work and that require a starting procedure to demonstrate the formal order of accuracy, it is essential that the mesh motion is continuously differentiable in time. In the context of fluid--structure interaction, the mesh motion is defined by the deformation of the fluid--structure interface according to the structural displacements and a mesh smoothing algorithm calculating the mesh deformation in the interior of the fluid domain.

The incompressible Navier--Stokes equations~\eqref{eq:MomentumEquationALE} and~\eqref{eq:ContinuityEquationALE} are subject to the initial condition
\begin{align*}
\bm{u}(\bm{x}, t=0) = \bm{u}_0(\bm{x}) \;\; \text{in} \; \Omega_0 \; ,
\end{align*}
where~$\bm{u}_0(\bm{x})$ is divergence-free and fulfills the velocity Dirichlet boundary condition shown below. On the boundary~$\Gamma = \partial \Omega = \boundary{D} \cup \boundary{N}$ with~$\boundary{D} \cap \boundary{N} = \emptyset$, Dirichlet and Neumann boundary conditions are prescribed
\begin{alignat}{2}
\bm{u} &= \bm{g}_{u}\;\; &&\text{on} \; \boundary{D}(t)\; ,\label{DirichletBC}\\
\left(\bm{F}_{\mathrm{v}} (\bm{u})  - p \bm{I} \right) \cdot \bm{n} &= \bm{h}\;\; &&\text{on} \; \boundary{N}(t)\; ,\label{NeumannBC_Coupled}
\end{alignat}
where~$\bm{n}$ is the outward pointing unit normal vector and~$\bm{I}$ the identity matrix. As explained in~\cite{Fehn17}, the Neumann boundary condition is split into a viscous part~$\bm{h}_u$ and a pressure part~$g_p$ in case of projection-type Navier--Stokes solvers, i.e.,~$\bm{F}_{\mathrm{v}} (\bm{u})\cdot \bm{n} = \bm{h}_u$ and~$p = g_p$ on~$\boundary{N}(t)$ with~$\bm{h} = \bm{h}_u - g_p \bm{n}$. Another characteristic of the incompressible Navier--Stokes equations is that the pressure is only defined up to an additive constant in case of pure Dirichlet boundary conditions,~$\Gamma = \boundary{D}$. A unique pressure solution is obtained from the constraint~$\int_{\Omega} p \;\mathrm{d}\Omega = 0$. In that special case the velocity Dirichlet boundary condition has to fulfill the constraint~$\int_{\boundary{D}} \bm{g}_{u} \cdot \bm{n}\; \mathrm{d}\Gamma = 0$, which is the integral version of the continuity equation~\eqref{eq:ContinuityEquationALE}, transformed into a surface integral via Gauss' divergence theorem.

\section{Temporal discretization}\label{sec:TemporalDiscretization}

A multitude of solution strategies have been proposed over the last decades to solve the incompressible Navier--Stokes equations. While monolithic approaches are straight-forward in terms of time integration and in achieving high-order accuracy, projection methods with a splitting of velocity and pressure unknowns are particularly interesting from the point of view of computational costs, as these techniques decompose the problem into easier-to-solve equations such as simple Poisson or Helmholtz-like problems. The literature on projection methods is vast, see~\cite{Guermond06,Karniadakis13} for an overview. Here, we focus on those methods that we believe are most widely used and that also cover different aspects related to implicit versus mixed explicit--implicit formulations and the availability of high-order formulations, e.g., through rotational formulations. As representatives of projection methods, we investigate (incremental) pressure-correction schemes in rotational form, see~\cite{Chorin68,hirt1972,Goda1979,VanKan1986,Timmermans1996,Guermond2004} for the development of this approach, as well as velocity-correction schemes, using the high-order formulation proposed in~\cite{Orszag1986, Karniadakis1991}. For these methods, the splitting of the incompressible Navier--Stokes equations is performed on the level of differential equations, as compared to algebraic splitting methods. In the following subsections~\ref{sec:CoupledSolution},~\ref{sec:DualSplitting}, and~\ref{sec:PressureCorrection} we briefly summarize the different solution strategies considered in this work with a focus on the ALE relevant aspects and especially boundary conditions, while we refer to previous works~\cite{Fehn17, Fehn18a} for a more detailed description in the context of Eulerian formulations and high-order DG methods. In Section~\ref{sec:NavierStokesSolversDiscussion}, the different Navier--Stokes solvers are discussed in terms of stability and achievable rates of convergence.

\begin{notation}
Backward differentiation formula (BDF) time integration is used in this work. The time interval~$[0,T]$ is divided into~$N$ time steps of variable size. With~$n=0,...,N-1$ denoting the time step number, the equations are advanced from time~$t_{n}$ to~$t_{n+1} = t_{n} + \Delta t_n$ in time step~$n$, leading to the time grid~${\lbrace t_{i}\rbrace}_{i=0}^{N} =  \lbrace{ t_0 + \sum_{j=0}^{i-1} \Delta t_{j}\rbrace}_{i=0}^{N}$. BDF schemes of order~$J=1,2,3$ are considered. Although A-stability is only achieved for time integration schemes of order~$J=1,2$, third-order accurate schemes have been found to be useful for practical problems as well and will be investigated in this work. Coefficients~$\gamma_0$ and~$\alpha_i$ of the BDF time integration scheme as well as coefficients~$\beta_i$ of the extrapolation scheme used to extrapolate explicit terms are listed in Table~\ref{tab:CoefficientsBDFTimeIntegration} for the case of a constant time step size~$\Delta t = T/N$, see also~\cite{Karniadakis1991}. An extension to variable time step sizes is straight-forward, where the coefficients~$\gamma_0^n$,~$\alpha_i^n$, and~$\beta_i^n$ vary from one time step to the next and can be expressed as simple rational functions of the time step sizes~$\Delta t_n, ...,\Delta t_{n-J+1}$, see also~\cite{Wang2008}. The time integration constant for adaptive time-stepping are summarized in~\ref{sec:AdaptiveTimeStepping}.
\end{notation}

\begin{table}[h]
\caption{Coefficients of BDF time integration scheme and extrapolation scheme for constant time step size, see~\cite{Karniadakis1991}.}\label{tab:CoefficientsBDFTimeIntegration}
\begin{center}
\begin{tabular}{cccccccc}
\toprule
Order & $\gamma_0$  & $\alpha_0$ & $\alpha_1$ & $\alpha_2$ & $\beta_0$ & $\beta_1$ & $\beta_2$\\ 
\midrule
1 & $1$     & $1$ & -      & -     & $1$ & -    & -  \\
2 & $3/2$   & $2$ & $-1/2$ & -     & $2$ & $-1$ & -  \\
3 & $11/6$  & $3$ & $-3/2$ & $1/3$ & $3$ & $-3$ & $1$\\
\bottomrule
\end{tabular} 
\end{center}
\end{table}

\subsection{Coupled solution approach}\label{sec:CoupledSolution}
Applying the BDF scheme to equations~\eqref{eq:MomentumEquationALE} and~\eqref{eq:ContinuityEquationALE} and using a fully implicit formulation (including the convective term), we obtain
\begin{align}
\left.\frac{\gamma_0^n \bm{u}^{n+1}-\sum_{i=0}^{J-1}\alpha_i^n\bm{u}^{n-i}}{\Delta t_n}\right\vert_{\boldsymbol{\chi}} + \left( \left(\bm{u}^{n+1}  - \bm{u}_{\mathrm{G}}^{n+1} \right) \cdot \nabla\right) \bm{u}^{n+1}
- \Div{\bm{F}_{\mathrm{v}} (\bm{u}^{n+1})} + \Grad{p^{n+1}} &= \bm{f}\left(t_{n+1}\right)\; ,\label{eq:TemporalDiscretization_Coupled_Momentum}\\
\Div{\bm{u}^{n+1}} &= 0 \; ,\label{eq:TemporalDiscretization_Coupled_Continuity}
\end{align}
where~$\left.\right\vert_{\boldsymbol{\chi}}$ means that all terms of the BDF sum are evaluated at constant~$\boldsymbol{\chi}$, i.e., an ALE-type time derivative has to be considered. The boundary conditions are
\begin{align*}
\bm{u}^{n+1} = \bm{g}_{u}^{n+1}\;\; \text{on} \; \boundary{D} \; ,\\
\left(\bm{F}_{\mathrm{v}} (\bm{u}^{n+1})-p^{n+1} \bm{I}\right)\cdot \bm{n} = \bm{h}^{n+1}\;\; \text{on} \; \boundary{N} \; .
\end{align*}
As an alternative formulation, we also study an explicit formulation of convective term, discretized in time via an extrapolation scheme of order~$J$
\begin{align}
\left.\frac{\gamma_0^n \bm{u}^{n+1}-\sum_{i=0}^{J-1}\alpha_i^n\bm{u}^{n-i}}{\Delta t_n}\right\vert_{\boldsymbol{\chi}} + \sum_{i=0}^{J-1}\beta_i^{n}\left(\left(\bm{u}^{n-i}  - \bm{u}_{\mathrm{G}}^{n+1} \right) \cdot \nabla\right) \bm{u}^{n-i}
- \Div{\bm{F}_{\mathrm{v}} (\bm{u}^{n+1})} + \Grad{p^{n+1}} &= \bm{f}\left(t_{n+1}\right)\; ,\label{TemporalDiscretization_Coupled_Momentum_Explicit}
\end{align}
where the following boundary condition will be used in the convective term
\begin{align}
\bm{u}^{n-i} = \bm{g}_{u}^{n+1}\;\; \text{on} \; \boundary{D} \; .\label{eq:Coupled_DBC_ConvectiveTerm}
\end{align}
\begin{remark}
While the boundary condition~$\bm{u}^{n-i} = \bm{g}_{u}^{n-i}$ might be considered a natural formulation as well, it is interesting to study whether equation~\eqref{eq:Coupled_DBC_ConvectiveTerm} also preserves optimal rates of convergence. The advantage of this formulation is that only one version of the boundary condition is required at any one time of the solution of the transient problem, therefore easing implementation without the need to store and keep track of previous versions of the boundary condition for fluid--structure interaction problems. For the projection methods discussed below, we will see that we can not fully maintain this goal as these methods are more involved regarding the formulation of boundary conditions.
\end{remark}

\subsection{High-order dual splitting scheme}\label{sec:DualSplitting}
The high-order dual splitting scheme~\cite{Karniadakis1991} consists of the following four sub-steps to be solved in each time step
\begin{align}
\begin{aligned}
\left.\frac{\gamma_0^n\hat{\bm{u}}-\sum_{i=0}^{J-1}\alpha_i^n\bm{u}^{n-i}}{\Delta t_n}\right\vert_{\boldsymbol{\chi}} =
\quad- \sum_{i=0}^{J-1}\beta_i^{n}\left(\left(\bm{u}^{n-i}  - \bm{u}_{\mathrm{G}}^{n+1} \right) \cdot \nabla\right) \bm{u}^{n-i}
+ \bm{f}\left(t_{n+1}\right),&\\ 
\bm{u}^{n-i} = \bm{g}_u^{n+1}  &\quad\text{on }  \boundary{D},
\end{aligned}&\label{eq:DualSplitting_ConvectiveStep}\\[10pt]
\begin{aligned}
-\nabla^2 p^{n+1} = -\frac{\gamma_0^n}{\Delta t_n}\Div{\hat{\bm{u}}},&\\
\Grad{p}^{n+1}\cdot\bm{n} = h_p^{n+1} &\quad\text{on }  \boundary{D},\\
p^{n+1} = g_p^{n+1} &\quad\text{on }  \boundary{N},\\
\hat{\bm{u}} = \bm{g}_{\hat{u}}^{n+1}  &\quad\text{on }  \boundary{D},
 \end{aligned}&\label{eq:DualSplitting_PressureStep}\\[10pt]
 \begin{aligned}
\hat{\hat{\bm{u}}} = \hat{\bm{u}} - \frac{\Delta t_n}{\gamma_0^n} \Grad{p^{n+1}} ,&\\  
p^{n+1} = g_p^{n+1}& \quad\text{on }  \boundary{N},
\end{aligned}&\label{eq:DualSplitting_ProjectionStep}\\[10pt]
\begin{aligned}
\frac{\gamma_0^n}{\Delta t_n} \bm{u}^{n+1}  -  \Div{\bm{F}_{\mathrm{v}}\left(\bm{u}^{n+1}\right)} =
\frac{\gamma_0^n}{\Delta t_n}\hat{\hat{\bm{u}}} ,&\\   
\bm{u}^{n+1} = \bm{g}_{u}^{n+1} &\quad\text{on } \boundary{D},\\  
\bm{F}_{\mathrm{v}} (\bm{u}^{n+1})\cdot \bm{n} = \bm{h}_u^{n+1} &\quad\text{on } \boundary{N}.
\end{aligned}&\label{eq:DualSplitting_ViscousStep}
\end{align}
The convective term, the pressure term, and the viscous term are taken into account in different sub-steps, and the pressure Poisson equation~\eqref{eq:DualSplitting_PressureStep} is obtained from equation~\eqref{eq:DualSplitting_ProjectionStep} by requiring~$\Div{\hat{\hat{\bm{u}}}}=0$. The viscous term is formulated implicitly in time, while the convective term is formulated explicitly for this scheme. The authors are not aware of variants of this splitting scheme that allow an implicit treatment of the convective term and -- at the same time -- achieve higher-order accuracy in time.

The pressure Neumann boundary condition~$h_p$ and the Dirichlet boundary condition~$\bm{g}_{\hat{u}}$ for the intermediate velocity need to be discussed in more detail. A consistent Neumann boundary condition for the pressure is derived by multiplying the momentum equation of the incompressible Navier--Stokes equations by the normal vector~$\bm{n}$~\cite{Orszag1986,Karniadakis1991,Karniadakis13}, which yields the following result in case of the ALE form of the incompressible Navier--Stokes equations
\begin{align}
\begin{split}
 h_p\left(t_{n+1}\right) =
&- \left[
\left.\frac{\gamma_0^n \bm{g}_u^{n+1}-\sum_{i=0}^{J-1}\alpha_i^n\bm{g}_u^{n-i}}{\Delta t_n}\right\vert_{\bm{\chi}} -\bm{f}\left(t_{n+1}\right) \right]\cdot \bm{n}^{n+1} \\
& - \left[ \sum_{i=0}^{J_p-1} \beta_i^n\left(\left(\left(\bm{u}^{n-i}  - \bm{u}_{\mathrm{G}}^{n+1} \right) \cdot \nabla\right)\bm{u}^{n-i}
+\nu\nabla\times\boldsymbol{\omega}^{n-i}\right)
 \right]\cdot \bm{n}^{n+1}\; .
 \end{split}\label{eq:DualSplitting_PressureBC_GammaD}
\end{align}
As already noted in our previous work~\cite{Fehn17} dealing with the Eulerian case, the time derivative term and body force term appear on the right-hand side as well in the general case of time dependent boundary conditions and right-hand side vectors~$\bm{f}\neq \bm{0}$, compared to the original formulation in~\cite{Karniadakis1991} and a later work~\cite{Beskok2001} where these terms are dropped. Compared to~\cite{Fehn17} where the exact derivative~$\partial \bm{g}_u/\partial t$ is used by exploiting that this term is known analytically in the Eulerian case on~$\boundary{D}$, the acceleration term has to be replaced by a discrete BDF time derivative in the ALE or fluid--structure interaction case where the boundary condition is only known at discrete times. Hence, one has to record the history of Dirichlet boundary values~$\bm{g}_u$ in case of the dual splitting scheme. Since the time derivative is of ALE-type at constant~$\boldsymbol{\chi}$, the convective term needs to be formulated in ALE form as well. The velocity~$\bm{u}_h^{n+1}$ at time~$t_{n+1}$ is unknown at this point of the algorithm, so that the convective term and the viscous term need to be formulated explicitly. For the viscous term we use the well-known rotational formulation~$\nabla \times \bm{\omega}$ (with vorticity~$\bm{\omega}=\nabla \times \bm{u}$) which is obtained from the identity~$\nabla^2 \bm{u} = \Grad{\left(\Div{\bm{u}}\right)}-\nabla \times \left( \nabla \times \bm{u} \right)=-\nabla \times \left( \nabla \times \bm{u} \right)= - \nabla \times \bm{\omega}$. This formulation proposed by~\cite{Orszag1986,Karniadakis1991} makes use of the incompressibility constraint~$\Div{\bm{u}}=0$, and it is well understood that the rotational formulation significantly improves accuracy as compared to the Laplace formulation~\cite{Karniadakis13,Guermond06}; nevertheless the latter formulation~$\nabla^2 \bm{u}$ is also used sometimes~\cite{EscobarVargas14,Gao2017,Wang2018}. 

The velocity Dirichlet boundary condition for the intermediate velocity~$\hat{\bm{u}}$ is derived in a similar fashion from equation~\eqref{eq:DualSplitting_ConvectiveStep}, see also~\cite{Fehn17} where this boundary condition has been proposed for the Eulerian case
\begin{align}
\bm{g}_{\hat{u}}\left(\bm{\chi},t_{n+1}\right) = \sum_{i=0}^{J-1}\frac{\alpha_i^n}{\gamma_0^n}\bm{g}_{u}(\bm{\chi},t_{n-i})-\frac{\Delta t_n}{\gamma_0^n} \sum_{i=0}^{J-1}\beta_i^n \left(\left(\bm{u}^{n-i}  - \bm{u}_{\mathrm{G}}^{n+1} \right) \cdot \nabla\right)\bm{u}^{n-i}
+ \frac{\Delta t_n}{\gamma_0^n} \bm{f}\left(t_{n+1}\right)\; ,\label{eq:DualSplitting_DBC_IntermediateVelocity}
\end{align}
As indicated in the above equation, the history of the boundary condition~$\bm{g}_{u}$ is evaluated in grid coordinates~$\boldsymbol{\chi}$ following the moving mesh from one time instant to the next, so that the ALE form of the convective term is required in this boundary condition similar to the pressure Neumann boundary condition~\eqref{eq:DualSplitting_PressureBC_GammaD}. The convective term again needs to be extrapolated.

\begin{remark}
One might raise the question why the terms in equation~\eqref{eq:DualSplitting_PressureBC_GammaD} coming from the acceleration term evaluate the prescribed boundary data~$\bm{g}_u$ instead of simply evaluating the numerical solution~$\bm{u}$ coming from the interior of the domain, since this is also done for the convective and viscous terms in the boundary conditions. From the formulations in~\cite{Hesthaven07,Ferrer11,Ferrer14,Gao2017,Wang2018} it is unclear whether the acceleration term in equation~\eqref{eq:DualSplitting_PressureBC_GammaD} should directly evaluate the numerical solution~$\bm{u}$ if the time derivative is not known analytically. The same holds for the BDF time derivative terms in equation~\eqref{eq:DualSplitting_DBC_IntermediateVelocity}. We found that using~$\bm{u}$ instead of~$\bm{g}_u$ for the time derivative terms in equations~\eqref{eq:DualSplitting_PressureBC_GammaD} and~\eqref{eq:DualSplitting_DBC_IntermediateVelocity} leads to instabilities that occur for small time step sizes. We therefore recommend to use the available boundary data~$\bm{g}_u$ in the above boundary conditions whenever possible. We further note that the works~\cite{Hesthaven07,Ferrer11} do not use a boundary condition like equation~\eqref{eq:DualSplitting_DBC_IntermediateVelocity} since these works use a different DG discretization of the velocity divergence term in equation~\eqref{eq:DualSplitting_PressureStep} compared to the present work, namely a formulation without integration by parts and, hence, without imposition of boundary conditions, see Section~\ref{sec:SpatialDiscretization} for details. However, it was shown in~\cite{Fehn17} that integration by parts of this term in the context of discontinuous Galerkin methods is essential to obtain stability for small time steps and that a consistent Dirichlet boundary condition according to equation~\eqref{eq:DualSplitting_DBC_IntermediateVelocity} is necessary to obtain high-order accuracy in time.
\end{remark}

\subsection{Pressure-correction scheme}\label{sec:PressureCorrection}
Pressure-correction schemes are another class of projection methods, treating both convective and viscous terms in the same sub-step, see~\cite{Guermond06} for detailed information. Extending the formulation of pressure-correction schemes for non-moving meshes shown in~\cite{Fehn17,Fehn18a} to the ALE formulation of the incompressible Navier--Stokes equations, the pressure-correction schemes can be summarized as follows
\begin{align}
\begin{aligned}
\left.\frac{\gamma_0^n \hat{\bm{u}}-\sum_{i=0}^{J-1}\alpha_i^n\bm{u}^{n-i}}{\Delta t_n}\right\vert_{\boldsymbol{\chi}} + \left( \left(\hat{\bm{u}}  - \bm{u}_{\mathrm{G}}^{n+1} \right) \cdot \nabla\right) \hat{\bm{u}}
-\Div{\bm{F}_{\mathrm{v}}\left(\hat{\bm{u}}\right)} = &\\
- \sum_{i=0}^{J_p-1}\beta_i^n \Grad{p^{n-i}} + \bm{f}\left(t_{n+1}\right) , &\\ 
\hat{\bm{u}} = \bm{g}_{u}^{n+1} &\quad\text{on }   \boundary{D} , \\
\bm{F}_{\mathrm{v}} (\hat{\bm{u}})\cdot \bm{n} = \bm{h}_u^{n+1}, &\quad\text{on }  \boundary{N} , \\
p^{n-i}=g_p^{n-i} &\quad\text{on }  \boundary{N} ,
\end{aligned}&\label{eq:PressureConvection_MomentumStep_Impl}\\[10pt]
\begin{aligned}
-\nabla^2 \phi^{n+1}= -\frac{\gamma_0^n }{\Delta t_n}\Div{\hat{\bm{u}}} , &\\
\Grad{\phi^{n+1}}\cdot\bm{n} = h_{\phi}^{n+1} = 0 &\quad\text{on }  \boundary{D}, \\
\hat{\bm{u}} = \bm{g}_{u}^{n+1} &\quad\text{on }  \boundary{D} , \\
\phi^{n+1} = g_{\phi}^{n+1}   &\quad\text{on }  \boundary{N}  ,
\end{aligned}&\label{eq:PressureCorrection_PressurePoissonEquation}\\[10pt]
\begin{aligned}
p^{n+1} = \phi^{n+1} + \sum_{i=0}^{J_p-1}\beta_i^n p^{n-i} - \zeta \nu \Div{\hat{\bm{u}}} , &\\
\hat{\bm{u}} = \bm{g}_{u}^{n+1} &\quad\text{on }  \boundary{D} ,
\end{aligned}&\label{eq:PressureCorrection_PressureUpdate}\\[10pt]
\begin{aligned}
\bm{u}^{n+1} = \hat{\bm{u}} - \frac{\Delta t_n}{\gamma_0^n} \Grad{\phi^{n+1}} , &\\
\phi^{n+1} = g_{\phi}^{n+1}  &\quad\text{on }  \boundary{N} .
\end{aligned}&\label{eq:PressureCorrection_Projection}
\end{align}
Here,~$J_p$ is the order of extrapolation of the pressure gradient term. Schemes with~$J_p=0$ are called non-incremental, and schemes with~$J_p \geq 1$ incremental pressure-correction schemes. The standard formulation is obtained for~$\zeta=0$, while the rotational formulation corresponds to~$\zeta=1$. We exclusively study the rotational version in the present work due to improved convergence rates and accuracy~\cite{Guermond06}. In the above equations, the convective term is formulated implicitly. As for the monolithic solver, we consider an alternative formulation with explicit treatment of the convective term, resulting in the following momentum equation in the first sub-step
\begin{align}
\left.\frac{\gamma_0^n \hat{\bm{u}}-\sum_{i=0}^{J-1}\alpha_i^n\bm{u}^{n-i}}{\Delta t_n}\right\vert_{\boldsymbol{\chi}}
-\Div{\bm{F}_{\mathrm{v}}\left(\hat{\bm{u}}\right)} = - \sum_{i=0}^{J-1}\beta_i^{n}\left(\left(\bm{u}^{n-i}  - \bm{u}_{\mathrm{G}}^{n+1} \right) \cdot \nabla\right) \bm{u}^{n-i}
- \sum_{i=0}^{J_p-1}\beta_i^n \Grad{p^{n-i}}
+ \bm{f}\left(t_{n+1}\right)\; ,\label{PressureConvection_MomentumStep_Expl}
\end{align}
where the following boundary condition is imposed for the convective term
\begin{align}
\bm{u}^{n-i} = \bm{g}_{u}^{n+1}\;\; \text{on} \; \boundary{D} \; .
\end{align}
The pressure Poisson equation~\eqref{eq:PressureCorrection_PressurePoissonEquation} is subject to the pressure Dirichlet boundary condition
\begin{align}
g_{\phi}(\bm{\chi},t_{n+1}) = g_p \left(\bm{\chi},t_{n+1}\right) -\sum_{i=0}^{J_p-1}\beta_i g_{p}\left(\bm{\chi},t_{n-i}\right) \; .\label{eq:PressureCorrection_PressureBC_GammaN}
\end{align}

\begin{remark}
As mentioned in~\cite{Fehn17} for the Eulerian version of the present solver, the boundary condition~\eqref{eq:PressureCorrection_PressureBC_GammaN} is an extension of the boundary condition (10.3) in~\cite{Guermond06} towards time-dependent pressure boundary conditions on~$\boundary{N}$. Note that we prescribe~$p^{n-i}=g_p^{n-i}$ in equation~\eqref{eq:PressureConvection_MomentumStep_Impl} to be consistent with boundary condition~\eqref{eq:PressureCorrection_PressureBC_GammaN}. We otherwise observe suboptimal rates of convergence when prescribing~$p^{n-i}=g_p^{n+1}$. Hence, a history of pressure Dirichlet boundary values on~$\boundary{N}$ has to be stored for the pressure-correction scheme for higher-order schemes with~$J_p\geq 1$, which can be seen as a consequence of the operator splitting as compared to the monolithic solver described in Section~\ref{sec:CoupledSolution}.
\end{remark}

\subsection{Discussion of incompressible Navier--Stokes solvers}\label{sec:NavierStokesSolversDiscussion}

Let us briefly summarize and discuss the different Navier--Stokes solver strategies. Already when considering the unsteady Stokes equations without convection, projection methods are in general only conditionally stable for higher extrapolation order~$J_p$. According to~\cite{Leriche2006}, the dual splitting scheme is unconditionally stable for~$J_p\leq 2$ independent of the order~$1\leq J \leq 4$ of the BDF scheme, but only conditionally stable for~$J_p>2$. Hence, the scheme with parameters~$J=3$ and~$J_p=2$ provides the highest order of accuracy, namely~$\Delta t^3$, among the schemes that are unconditionally stable. Pressure-correction schemes are only unconditionally stable for~$J_p\leq 1$, while they are conditionally stable for~$J_p \geq 2$, see~\cite{Guermond06}. This is also in agreement with our numerical results where we observe instabilities for~$J_p=2$. Hence, the highest accuracy combined with unconditional stability is achieved for~$J=2, J_p=1$, resulting in a second-order accurate scheme,~$\Delta t^2$. Let us mention that the scheme~$J=3, J_p=2$ shows indeed third-order accuracy, but is only conditionally stable and not suited for practical problems. From these considerations, we derive the parameters~$J_p=\min(2, J), J\leq 3$ for the dual splitting scheme, and~$J_p=\min(2, J)-1, J\leq 2$ for the pressure-correction scheme as an optimal choice. The monolithic solution approach and the dual splitting scheme have the advantage that the third-order schemes are stable. The fact that the pressure-correction scheme allows both implicit and explicit formulations of the convective term can be seen as an advantage over the dual splitting scheme, especially when one wants to avoid restrictions of the time step size according to the CFL condition and when high-order of accuracy is not the primary target. Finally, the simpler structure of algebraic equations is often considered an advantage of projection methods over the monolithic approach. We therefore analyze these different methods in the present work since we believe they cover different facets of incompressible Navier--Stokes solvers.

\section{Spatial discretization}\label{sec:SpatialDiscretization}

\begin{notation}
The computational domain~$\Omega_h = \bigcup_{e=1}^{N_{\text{el}}} \Omega_{e} \in \mathbb{R}^d$ is composed of~$N_{\text{el}}$ finite elements, which are non-overlapping and of quadrilateral/hexahedral shape in the context of this work. The boundary~$\Gamma_h = \partial \Omega_h$ approximates~$\Gamma$, and it holds~$\Gamma_h = \hboundary{D} \cup \hboundary{N}$ with~$\hboundary{D} \cap \hboundary{N} =\emptyset$ as in the spatially continuous case. Approximations to velocity~$\bm{u}(\bm{x},t)$ and pressure~$p(\bm{x},t)$ are denoted by~${\bm{u}_h(\bm{x},t)\in\mathcal{V}^{u}_h}$ and~$p_h(\bm{x},t)\in \mathcal{V}^{p}_h$, where the discontinuous Galerkin finite element spaces of test and trial functions are defined as
\begin{align*}
\mathcal{V}^{u}_{h} &= \left\lbrace\bm{u}_h\in \left[L_2(\Omega_h)\right]^d\; : \; \bm{u}_h\left(\bm{x}^e(\boldsymbol{\xi},t)\right)\vert_{\Omega_{e}}= \tilde{\bm{u}}_h^e(\boldsymbol{\xi})\vert_{\tilde{\Omega}_{e}}\in \mathcal{V}^{u}_{h,e}=[\mathcal{Q}_{k_u}(\tilde{\Omega}_{e})]^d\; ,\;\; \forall e=1,\ldots,N_{\text{el}} \right\rbrace\;\; ,\\
\mathcal{V}^{p}_{h} &= \left\lbrace p_h\in L_2(\Omega_h)\; : \; p_h\left(\bm{x}^e(\boldsymbol{\xi},t)\right)\vert_{\Omega_{e}} = \tilde{p}_h^e(\boldsymbol{\xi})\vert_{\tilde{\Omega}_{e}}\in \mathcal{V}^{p}_{h,e}=\mathcal{Q}_{k_p}(\tilde{\Omega}_{e})\; ,\;\; \forall e=1,\ldots,N_{\text{el}} \right\rbrace\; .
\end{align*}
The polynomial space~$\mathcal{Q}_{k}(\tilde{\Omega}_{e})$ of tensor degree~$\leq k$ is defined on the reference element~$\tilde{\Omega}_e=[0,1]^d$ with reference coordinates~$\boldsymbol{\xi}=(\xi_1,...,\xi_d)^{\mathsf{T}}$. We approximate velocity and pressure on element~$e$ by nodal Lagrange polynomials
\begin{align}
\tilde{\bm{u}}_h^e(\boldsymbol{\xi},t) = \sum_{i_1,...,i_d=0}^{k_u} l_{i_1...i_d}^{k_u}(\boldsymbol{\xi})\bm{u}_{i_1...i_d}^e(t)\;\; , \;\; \tilde{p}_h^e(\boldsymbol{\xi},t) = \sum_{i_1,...,i_d=0}^{k_p} l_{i_1...i_d}^{k_p}(\boldsymbol{\xi})p_{i_1...i_d}^e(t)\; ,\label{eq:finite_element_expansion}
\end{align}
where~$ \bm{u}_{i_1...i_d}^e$ and~$p_{i_1...i_d}^e$ denote the nodal degrees of freedom of the velocity and pressure solution on element~$e$, respectively. The multidimensional shape functions~$l_{i_1...i_d}^{k}$ are given as the tensor product of one-dimensional shape functions,~$l_{i_1...i_d}^{k}(\boldsymbol{\xi})=\prod_{n=1}^{d} l_{i_n}^{k,\mathrm{1D}}(\xi_n)$, where~$l_i^{k,\mathrm{1D}}(\xi)$ are the Lagrange polynomials of degree~$k$ based on the Legendre--Gauss--Lobatto nodes. 

For reasons of inf--sup stability, the polynomial degree for the pressure is~$k_p = k_u-1$, see~\cite{Fehn17}. For ease of notation, we simply write~$k_u=k$ in the following. In the above equations,~$\bm{x}^e(\boldsymbol{\xi},t) : \tilde{\Omega}_e \rightarrow \Omega_e(t)$ denotes the mapping from reference space to physical space
\begin{align*}
\bm{f}_{\mathrm{m}}^e: \begin{cases}
\tilde{\Omega}_{e} \times \left[0,T\right]\rightarrow \Omega_e(t), \, \tilde{\Omega}_{e} = \left[0,1\right]^d, \Omega_e(t) \subset \mathbb{R}^d\, ,\\
\left(\bm{\xi},t\right)\mapsto \bm{x}^e\left(\bm{\xi}, t\right)\, .
\end{cases}
\end{align*}
For the mapping, the same ansatz is used as for approximating the solution, but with polynomial degree~$k_{\mathrm{m}}$
\begin{align*}
\bm{x}^e(\boldsymbol{\xi},t) = \sum_{i_1,...,i_d=0}^{k_\mathrm{m}} l_{i_1...i_d}^{k_{\mathrm{m}}}(\boldsymbol{\xi})\bm{x}_{i_1...i_d}^e(t)\; .
\end{align*}
The mapping can be seen in analogy to the function~$\bm{f}_{\mathrm{G}}$ describing the topological changes of the domain~$\Omega$. While~$\bm{f}_{\mathrm{G}}$ is defined globally for the whole domain and in a spatially continuous way, the finite element mapping describes the mesh motion for each element of the mesh in the discrete setting and is of finite dimension. For the following derivations, it is important to realize that a point with constant~$\boldsymbol{\chi}$ can be thought of as a point with fixed~$\boldsymbol{\xi}$ coordinates within one element, i.e., there exists a bijective map between~$\boldsymbol{\chi}$ and~$\boldsymbol{\xi}$ for each element. An illustration is given in Figure~\ref{fig:DeformationAndMapping}. In the context of this work, we restrict ourselves to problems for which the topology of the mesh does not change, i.e., no remeshing. Hence, the data structures do not have to be adjusted dynamically when moving the mesh. Updating the mesh means updating the~$d (k_{\mathrm{m}}+1)^d$ mapping degrees of freedom per element, i.e.,
\begin{align*}
 \bm{x}_{i_1...i_d}^e(t^{n+1}) = \bm{f}_{\mathrm{G}}( \bm{x}_{i_1...i_d}^e(t=0) ,t^{n+1}) \; .
\end{align*}
Moderately large deformations are possible as long as~$\bm{f}_{\mathrm{G}}$ remains invertible, and invalid elements with invalid mapping or Jacobian will otherwise occur in the discrete setting. The numerical examples shown in this work use a high-order, isoparametric mapping with~$k_{\mathrm{m}}=k_u$.

In DG methods, integrals have to be computed over the interface~$f_{e^-/e^+}=\partial \Omega_{e^-} \cap \partial \Omega_{e^+}$ of two adjacent elements~$\Omega_{e^-}$ and~$\Omega_{e^+}$, where the outward pointing normal vectors are~$\bm{n}^{-}$ for~$\Omega_{e^-}$ and~$\bm{n}^{+}$ for~$\Omega_{e^+}$. Furthermore, let~$u^{-}_h$ and~$u^{+}_h$ denote the solution~$u_h$ on~$f_{e^-,e^+}$ evaluated from the interior of element~$e^-$ and element~$e^+$, respectively. By~$\Gamma_{h}^{\mathrm{int}}$ we denote the set of all interior faces. In the following, we make use of the average operator~$\avg{u}=\left(u^- + u^+\right)/2$, the jump operator~$ \jump{u} = u^- \otimes \bm{n}^- + u^+ \otimes \bm{n}^+$, and the oriented jump operator~$ \jumporiented{u} = u^- - u^+ $. Moreover, an element-by-element formulation is used where volume integrals are performed over the current element~$\Omega_e$ and face integrals over the boundary~$\partial \Omega_e$ of element~$e$. By definition, we denote interior information on the current element~$\Omega_e$ by the superscript~$(\cdot)^-$ and exterior information from neighboring elements by the superscript~$(\cdot)^+$. In this context, the normal vector~$\bm{n}$ equals~$\bm{n}^-$, while~$\bm{n}^+=-\bm{n}^-=-\bm{n}$. Finally, we introduce the abbreviations~$\intele{v}{u} = \int_{\Omega_e} v \odot u \; \mathrm{d}\Omega$ and~$\inteleface{v}{u} = \int_{\partial \Omega_e} v \odot u \; \mathrm{d} \Gamma$, where the operator~$\odot$ symbolizes inner products and will become clear from the context. An integral over the computational domain is to be understood as~$\intdomain{v}{u} = \sum_{e=1}^{N_{\mathrm{el}}}\intele{v}{u}$, and similarly for integrals over all interior faces, e.g.,~$\intinteriorfaces{\avg{v}}{u^*} = \sum_{e=1}^{N_{\mathrm{el}}}\intelefaceInterior{\frac{1}{2}v}{u^*}$ if~$u^*$ is single-valued.
\end{notation}

Due to the ansatz~\eqref{eq:finite_element_expansion} with a separation of space and time, applying the temporal discretization to the spatially discretized equations in ALE form becomes trivial in the sense that the structure of the equations is equivalent to the Eulerian case. Consider the time derivative term in equation~\eqref{eq:MomentumEquationALE} multiplied by test functions~$ß\bm{v}_h$, integrated over element~$\Omega_e(t=t_{n+1})$, and to be discretized in time
\begin{align*}
\intelenp{\bm{v}_h}{\left.\frac{\partial \bm{u}_h (t)}{\partial t}\right\vert_{\boldsymbol{\chi}}} &= \intelenp{\bm{v}_h}{\left.\frac{\partial \sum_{i_1,...,i_d=0}^{k_u} l_{i_1...i_d}^{k_u}(\boldsymbol{\xi})\bm{u}_{i_1...i_d}^e(t)}{\partial t}\right\vert_{\boldsymbol{\xi}}}\\
 &= \intelenp{\bm{v}_h}{\sum_{i_1,...,i_d=0}^{k_u} l_{i_1...i_d}^{k_u}(\boldsymbol{\xi}) \frac{\partial \bm{u}_{i_1...i_d}^e(t)}{\partial t}}\\
& \approx \intelenp{\bm{v}_h}{\sum_{i_1,...,i_d=0}^{k_u} l_{i_1...i_d}^{k_u}(\boldsymbol{\xi}) \frac{\gamma_0^n \bm{u}_{i_1...i_d}^{e,n+1}- \sum_{i=0}^{J-1} \alpha_i^n\bm{u}_{i_1...i_d}^{e,n-i} }{\Delta t_n}}\\
 & = \intelenp{\bm{v}_h}{\left.\frac{\gamma_0^n \bm{u}_{h}^{n+1}- \sum_{i=0}^{J-1} \alpha_i^n\bm{u}_{h}^{n-i} }{\Delta t_n}\right\vert_{\boldsymbol{\xi}}} \; .
\end{align*}
The BDF rule introduced in the third row of the above equation approximates the acceleration at time~$t_{n+1}$ consistently with the integral over~$\Omega_e$ taken at the same instant of time. In the following, we skip the label~$\left.\right\vert_{\boldsymbol{\xi}}$ for simplicity, as it is clear from the above derivation that the BDF rule is simply applied to the global solution vector containing the unknown degrees of freedom and that all terms of the BDF sum use the same mass matrix at time~$t_{n+1}$. As explained in more detail in Section~\ref{sec:GCL}, using the same mass matrix for all solution vectors is important in order to satisfy the geometric conservation law~\cite{Foerster2006}. The above equation highlights that discretization in space and time commutate, meaning that the last term of the above equation would have also been obtained by discretizing equation~\eqref{eq:TemporalDiscretization_Coupled_Momentum} in space. However, the projection-type solution methods considered in this work are already formulated in a time-discrete manner, since the splitting is performed on the level of differential operators. For this reason, the derivation of DG formulations shown in the following starts from the time-discrete problems stated in Section~\ref{sec:TemporalDiscretization}.
 
Following~\cite{Foerster2006}, the grid velocity is computed in the same way via a BDF time derivative of the nodal grid coordinates~$\bm{x}_{i_1...i_d}^e$ in the time discrete case in order to achieve high-order temporal convergence on moving meshes
\begin{align}
\bm{u}_{\mathrm{G},h}^{n+1} = \left. \frac{\partial \bm{x}_h}{\partial t} (t_{n+1})\right\vert_{\bm{\chi}} \approx   \left.\frac{\gamma_0^n \bm{x}_h^{n+1} - \sum_{i=0}^{J-1} \alpha_i^n\bm{x}_h^{n-i} }{\Delta t_n}\right\vert_{\boldsymbol{\xi}} \; .\label{eq:MeshVelocity}
\end{align}
This procedure is different from~\cite{Hughes1981, Beskok2001} where the grid coordinates are updated by integrating the mesh velocity forward in time. In the following, we summarize the variational formulation of the different ALE Navier--Stokes solvers.

 \begin{table}[!h]
\caption{Weak imposition of boundary conditions: choice of exterior values~$\left(\cdot\right)^+$ on domain boundaries as a function of interior values~$\left(\cdot\right)^-$ and prescribed boundary data for velocity and pressure in order to weakly impose boundary conditions according to the mirror principle. The procedure is equivalent to the Eulerian case~\cite{Fehn17}.}\label{tab:BCsWeak}
\renewcommand{\arraystretch}{1.1}
\begin{center}
\begin{tabular}{lll}
\toprule
 & $\hboundary{D}$ & $\hboundary{N}$\\
\midrule
velocity & $\bm{u}_h^{+} = -\bm{u}_h^{-} + 2 \bm{g}_{u}$ & $\bm{u}_h^{+} = \bm{u}_h^{-}$\\
 & $\Grad{\bm{u}_{h}^{+}}\cdot \bm{n} = \Grad{\bm{u}_{h}^{-}}\cdot\bm{n}$ & $\Grad{\bm{u}_{h}^{+}}\cdot \bm{n} =  -\Grad{\bm{u}_{h}^{-}}\cdot\bm{n}+\frac{2\bm{h}_u}{\nu}$\\
pressure & $p^+_h = p^-_h$ & $p^+_h = - p^-_h + 2 g_p$\\
& $\Grad{p_{h}^{+}}\cdot \bm{n} = -\Grad{p_{h}^{-}}\cdot\bm{n} + 2 h_p$ & $\Grad{p_{h}^{+}}\cdot \bm{n} = \Grad{p_{h}^{-}}\cdot\bm{n}$\\
\bottomrule
\end{tabular}
\end{center}
\renewcommand{\arraystretch}{1}
\end{table}

\subsection{Monolithic solution approach}\label{sec:WeakFormCoupled}
Beginning with the monolithic solution approach, the weak discontinuous Galerkin formulation of the fully discrete problem with implicit formulation of the convective term can be summarized as follows: Find~$\bm{u}^{n+1}_h\in\mathcal{V}^u_h$,~$p^{n+1}_h\in \mathcal{V}^{p}_h$ such that
\begin{align}
\intelenp{\bm{v}_h}{\frac{\gamma_0^n \bm{u}^{n+1}_h-\sum_{i=0}^{J-1}\alpha_i^n\bm{u}^{n-i}_h}{\Delta t_n}}
+ c^{e,n+1}_h\left(\bm{v}_h,\bm{u}^{n+1}_h,\bm{u}_{\mathrm{G},h}^{n+1};\bm{g}_u^{n+1}\right)& \nonumber\\
+ v^{e,n+1}_h\left(\bm{v}_h,\bm{u}^{n+1}_h;\bm{g}_u^{n+1},\bm{h}_u^{n+1}\right)
+ g^{e,n+1}_h\left(\bm{v}_h,p^{n+1}_h;g_p^{n+1}\right)  &\nonumber\\
+ a^{e,n+1}_{\mathrm{D},h}(\bm{v}_h,\bm{u}_h^{n+1})
+ a^{e,n+1}_{\mathrm{C},h}(\bm{v}_h,\bm{u}_h^{n+1};\bm{g}_u^{n+1}) 
- \intelenp{\bm{v}_h}{\bm{f}(t_{n+1})} &= 0 \; , \label{eq:WeakForm_CoupledSolution_Momentum}\\
-d^{e,n+1}_h(q_h,\bm{u}^{n+1}_h;\bm{g}_u^{n+1})& = 0  \; ,\label{eq:WeakForm_CoupledSolution_Continuity}
\end{align}
for all~$(\bm{v}_h, q_h) \in \mathcal{V}^{u}_{h,e} \times \mathcal{V}^{p}_{h,e}$ and for all elements~$e=1,...,N_{\text{el}}$. The time label~$n+1$, e.g. in~$c_h^{e, n+1}$, indicates that the integral is evaluated on the domain~$\Omega_e^{n+1}$. When formulating the convective term explicitly, the convective term in the discretized momentum equation~\eqref{eq:WeakForm_CoupledSolution_Momentum} is replaced by
\begin{align*}
c^{e,n+1}_h\left(\bm{v}_h,\bm{u}^{n+1}_h,\bm{u}_{\mathrm{G},h}^{n+1};\bm{g}_u^{n+1}\right) \rightarrow \sum_{i=0}^{J-1}\beta_i^{n} c^{e,n+1}_h\left(\bm{v}_h,\bm{u}^{n-i}_h,\bm{u}_{\mathrm{G},h}^{n+1};\bm{g}_u^{n+1}\right)\; .
\end{align*}
As proposed in~\cite{Fehn18b}, a computational efficient variant of the monolithic system of equations, equations~\eqref{eq:WeakForm_CoupledSolution_Momentum} and~\eqref{eq:WeakForm_CoupledSolution_Continuity}, is to apply the divergence and continuity penalty terms in a postprocessing step
\begin{align}
\intelenp{\bm{v}_h}{\bm{u}^{n+1}_h}+ a^{e,n+1}_{\mathrm{D},h}\left(\bm{v}_h,\bm{u}^{n+1}_h\right)\Delta t_n + a^{e,n+1}_{\mathrm{C},h}\left(\bm{v}_h,\bm{u}^{n+1}_h;\bm{g}_u^{n+1}\right)\Delta t_n  =& \intelenp{\bm{v}_h}{\hat{\bm{u}}_h} \; ,\label{eq:WeakForm_CoupledSolution_Postprocessing}
\end{align}
where~$\hat{\bm{u}}_h$ is an intermediate velocity obtained as the solution of the coupled system of equations without penalty terms. For the numerical results studied in this work, the formulation shown in equations~\eqref{eq:WeakForm_CoupledSolution_Momentum} and~\eqref{eq:WeakForm_CoupledSolution_Continuity} is used with penalty terms added to the momentum equation.

We next present the discontinuous Galerkin formulation of the individual terms of the incompressible Navier--Stokes equations, see also~\cite{Fehn17,Fehn18a} for more detailed derivations. Boundary conditions are imposed according to the mirror principle as summarized in Table~\ref{tab:BCsWeak}. Central flux functions are used for the velocity divergence term
\begin{align}
\begin{split}
d^e_{h, \mathrm{weak}}\left(q_h,\bm{u}_h;\bm{g}_u\right) 
&= -\intele{\Grad{q_h}}{\bm{u}_h}+\inteleface{q_h}{\avg{\bm{u}_h}\cdot\bm{n}}\\
&=-\intele{\Grad{q_h}}{\bm{u}_h}
+\intelefaceInterior{q_h}{\avg{\bm{u}_h}\cdot\bm{n}}
+\intelefaceNeumann{q_h}{\bm{u}_h\cdot\bm{n}}
+\intelefaceDirichlet{q_h}{\bm{g}_{u}\cdot\bm{n}} \; ,
\end{split}\label{WeakFormVelocityDivergence}
\end{align}
and for the pressure gradient term
\begin{align}
\begin{split}
g^e_{h, \mathrm{weak}}\left(\bm{v}_h,p_h;g_p\right) 
&= -\intele{\Div{\bm{v}_h}}{p_h}+\inteleface{\bm{v}_h}{\avg{p_h}\bm{n}}\\
&= -\intele{\Div{\bm{v}_h}}{p_h}
+\intelefaceInterior{\bm{v}_h}{\avg{p_h}\bm{n}}
+\intelefaceDirichlet{\bm{v}_h}{p_h\bm{n}}
+\intelefaceNeumann{\bm{v}_h}{g_p\bm{n}}\; .
\end{split}\label{WeakFormPressureGradient}
\end{align}
As an alternative to the above weak forms, we consider the so-called strong formulations by performing integration-by-parts once again
\begin{align}
d^e_{h,\mathrm{strong}}\left(q_h,\bm{u}_h\right) 
&= \intele{q_h}{\Div{\bm{u}_h}}-\inteleface{q_h}{\frac{1}{2}\jumporiented{\bm{u}_h}\cdot\bm{n}}\; , \label{eq:StrongFormVelocityDivergence} \\
g^e_{h,\mathrm{strong}}\left(\bm{v}_h,p_h\right) 
&= \intele{\bm{v}_h}{\Grad{p_h}}-\inteleface{\bm{v}_h}{\frac{1}{2}\jumporiented{p_h}\bm{n}}\; . \label{eq:StrongFormPressureGradient}
\end{align}
The weak and strong formulations are equivalent as long as integrals are evaluated exactly, which does not hold in general for the quadrature rules typically used, see Section~\ref{sec:NumericalIntegration}. For this reason, the weak and strong formulations behave differently regarding the fulfillment of the geometric conservation law as discussed in Section~\ref{sec:GCL}.
The discretization of the viscous term is based on the symmetric interior penalty Galerkin~(SIPG) method~\cite{arnold2002unified}
\begin{align}
\begin{split}
v_{h}^{e}(\bm{v}_h,\bm{u}_h;\bm{g}_u,\bm{h}_u) =&
 \intele{\Grad{\bm{v}_h}}{\nu \Grad{\bm{u}_h}}
  - \intelefaceInterior{\Grad{\bm{v}_h}}{\frac{\nu}{2} \jump{\bm{u}_h}}
  - \intelefaceInterior{\bm{v}_h}{\nu \avg{\Grad{\bm{u}_h}}\cdot\bm{n}}\\
  &+ \intelefaceInterior{\bm{v}_h}{\nu\tau \jump{\bm{u}_h}\cdot\bm{n}} \; .
  \end{split} \label{eq:weak_form_viscous_term}
\end{align}
Inserting the boundary conditions acccording to Table~\ref{tab:BCsWeak}, the viscous operator~$v_{h}^{e} = v_{h,\mathrm{hom}}^{e} + v_{h,\mathrm{inhom}}^{e}$ can be split into homogeneous contributions
\begin{align*}
\begin{split}
v_{h,\mathrm{hom}}^{e}(\bm{v}_h,\bm{u}_h) =&
  \intele{\Grad{\bm{v}_h}}{\nu \Grad{\bm{u}_h}}
  - \intelefaceInterior{\Grad{\bm{v}_h}}{\frac{\nu}{2} \jump{\bm{u}_h}}
  - \intelefaceDirichlet{\Grad{\bm{v}_h}}{\nu \bm{u}_h \otimes\bm{n}}\\
 & - \intelefaceInterior{\bm{v}_h}{\nu \avg{\Grad{\bm{u}_h}}\cdot\bm{n}} 
 - \intelefaceDirichlet{\bm{v}_h}{\nu \Grad{\bm{u}_h}\cdot\bm{n}} \\
 & + \intelefaceInterior{\bm{v}_h}{\nu\tau \jump{\bm{u}_h}\cdot\bm{n}} + \intelefaceDirichlet{\bm{v}_h}{2\nu\tau \bm{u}_h} \; ,
 \end{split}
\end{align*}
and inhomogeneous contributions
\begin{align*}
v_{h,\mathrm{inhom}}^{e}(\bm{v}_h;\bm{g}_u,\bm{h}_u) =\intelefaceDirichlet{\Grad{\bm{v}_h}}{\nu\, \bm{g}_{u} \otimes\bm{n}}
- \intelefaceNeumann{\bm{v}_h}{\bm{h}_u}
- \intelefaceDirichlet{\bm{v}_h}{2\nu\tau\bm{g}_{u}}\; .
\end{align*}
The SIPG penalty parameter~$\tau$ depends on the polynomial degree~$k$ and a characteristic element length~$h$. It has to be large enough to ensure coercivity of the bilinear form. For quadrilateral/hexahedral elements used in the present work, bounds for the penalty parameter have been derived in~\cite{Hillewaert13}. Following this work, the penalty parameter~$\tau_e$ of element~$e$ is calculated as
\begin{align}
\tau_e = (k+1)^2 \frac{A\left(\partial \Omega_e \setminus \Gamma_h\right)/2 + A\left(\partial \Omega_e \cap \Gamma_h\right)}{V\left(\Omega_e\right)}\; ,\label{TauIP_Element}
\end{align}
where~$V\left(\Omega_e\right) = \int_{\Omega_e}\mathrm{d}\Omega$ and~$A(f) = \int_{f\subset\partial\Omega_e}\mathrm{d}\Gamma$ are the element volume and surface area, respectively. The maximum value from both sides is chosen on interior faces,~$\tau = \max\left(\tau_{e^-},\tau_{e^+}\right)$ if face~$f \subseteq \partial \Omega_e \setminus \Gamma_h$, while~$\tau = \tau_e$ is used on boundary faces~$f \subseteq \partial \Omega_e \cap \Gamma_h$.

The convective term is particularly relevant in the ALE context and is written in non-conservative form with grid velocity~$\bm{u}_{\mathrm{G},h}$. We perform integration by parts twice (strong formulation), since we observed sub-optimal rates of convergence for the weak formulation in case of even polynomial degrees. Then, an upwind flux is used as numerical flux function to obtain
\begin{align}
\begin{split}
c_{h}^e\left(\bm{v}_h, \bm{u}_h, \bm{u}_{\mathrm{G},h};\bm{g}_{u} \right) =  
& \intele{\bm{v}_h}{\left(\Grad{\bm{u}_h}\right) \cdot \left(\bm{u}_h - \bm{u}_{\mathrm{G},h}\right)} - \inteleface{\bm{v}_h}{\left(\left(\avg{\bm{u}_h}-\bm{u}_{\mathrm{G},h}\right)\cdot\bm{n}\right) \bm{u}_h}\\ 
& + \Big(\bm{v}_h,\underbrace{\left(\left(\avg{\bm{u}_h}-\bm{u}_{\mathrm{G},h}\right)\cdot\bm{n}\right) \avg{\bm{u}_h} + \frac{1}{2} \vert \left(\avg{\bm{u}_h}-\bm{u}_{\mathrm{G},h}\right)\cdot\bm{n} \vert \jumporiented{\bm{u}_h}}_{\text{upwind flux}} \Big)_{\partial \Omega_e} \; .
\end{split}\label{eq:weak_form_convective_term}
\end{align}
To keep the formulation compact, we do not explicitly highlight the dependency of the variational form~$c_{h}^e$ on the boundary condition~$\bm{g}_{u}$ as it is clear that the boundary condition enters the formulation through the choice of exterior values on domain boundaries according to Table~\ref{tab:BCsWeak}, i.e.,~$\bm{u}_h^+ = - \bm{u}_h^- + 2 \bm{g}_{u}$ on~$\hboundary{D}$. Since the convective term is nonlinear and the residual forms the right-hand side of the linear solver in case of a Newton--Krylov approach, there is no need to split the convective operator into homogeneous and inhomogeneous contributions.
 
Finally, the divergence penalty term~$a^e_{\mathrm{D},h}$ and continuity penalty term~$a^e_{\mathrm{C},h}$ have to be defined. These terms can be interpreted as a weak enforcement of~$H(\mathrm{div})$-conformity (normal continuous velocity) along with suitable function spaces for velocity and pressure that lead to an exactly (pointwise) divergence-free velocity (such as Raviart--Thomas)~\cite{Fehn18a, Fehn19Hdiv, Akbas2018}. 
These penalty terms are mandatory to obtain a robust discretization for under-resolved problems such as turbulent flows when using standard~$L^2$-conforming spaces. The penalty terms are defined as~\cite{Fehn18a}
\begin{alignat*}{2}
a^e_{\mathrm{D},h}(\bm{v}_h,\bm{u}_h) &= \intele{\Div{\bm{v}_h}}{\tau_{\mathrm{D}}\Div{\bm{u}_h}} ,\ &&\tau_{\mathrm{D},e}=\zeta_{\mathrm{D}} \; \overline{\Vert\bm{u}^{n+1,\mathrm{ex}}_h \Vert} \; \frac{h_e}{k_u + 1} \; ,\\
a^e_{\mathrm{C},h}(\bm{v}_h,\bm{u}_h;\bm{g}_u) &= \inteleface{\bm{v}_h\cdot \bm{n}}{\tau_{\mathrm{C}} \jumporiented{\bm{u}_h}\cdot \bm{n}} , \ &&\tau_{\mathrm{C},e}=\zeta_{\mathrm{C}}\;\overline{\Vert \bm{u}^{n+1,\mathrm{ex}}_h \Vert} \; ,
\end{alignat*}
where~$\bm{u}^{n+1,\mathrm{ex}}_h=\sum_{i=0}^{J-1} \beta_i^n \bm{u}^{n-i}_h$ is an extrapolation of the velocity field of order~$J$,~$\overline{(\cdot)}$ an elementwise volume-averaged quantity, and~$h_e=V_e^{1/3}$ a characteristic element length with~$V_e$ the volume of the element. A minor modification compared to~\cite{Fehn18a} is that the continuity penalty term is applied not only on interior faces, but also on boundary faces with exterior values according to Table~\ref{tab:BCsWeak}. While mainly numerical examples with periodic boundary conditions have been studied~\cite{Fehn18a} where this does not make a difference, we observed that it is advantageous to apply this penalty term on all faces for general boundary conditions. The continuity penalty parameter is~$\tau_{\mathrm{C}}=\avg{\tau_{\mathrm{C},e}}$ on interior faces and~$\tau_{\mathrm{C}}=\tau_{\mathrm{C},e}$ on boundary faces. As shown above for other operators, the continuity penalty term can then be split into homogeneous and inhomogeneous contributions,~$a^e_{\mathrm{C},h}(\bm{v}_h,\bm{u}_h;\bm{g}_u) = a^e_{\mathrm{C},h,\mathrm{hom}}(\bm{v}_h,\bm{u}_h) + a^e_{\mathrm{C},h,\mathrm{inhom}}(\bm{v}_h;\bm{g}_u)$.

\begin{remark}
Consistency of the above variational formulation, equations~\eqref{eq:WeakForm_CoupledSolution_Momentum} and~\eqref{eq:WeakForm_CoupledSolution_Continuity}, immediately follows from the fact that the weak form is derived using integration by parts, using consistent numerical flux functions, and consistent boundary conditions according to Table~\ref{tab:BCsWeak}. The additional penalty terms are consistent as well, since these terms contain the divergence of the velocity or the jump of the velocity over interior faces.
\end{remark}

\subsection{High-order dual splitting scheme}\label{sec:WeakFormDualSplitting}
For the dual splitting projection scheme, the variational formulation can be summarized as follows: Find~$\hat{\bm{u}}_h,\hat{\hat{\bm{u}}}_h, \hat{\hat{\hat{\bm{u}}}}_h  ,\bm{u}_h^{n+1}\in\mathcal{V}^u_h$ and~$p_h^{n+1}\in\mathcal{V}^p_h$ such that for all~$\bm{v}_h \in \mathcal{V}^{u}_{h,e}$,~$q_h \in \mathcal{V}^{p}_{h,e}$ and for all elements~$e=1,...,N_{\text{el}}$
\begin{align}
\begin{aligned}
\intelenp{\bm{v}_h}{\frac{\gamma_0^n \hat{\bm{u}}_h-\sum_{i=0}^{J-1}\alpha_i^n\bm{u}^{n-i}_h}{\Delta t_n}}
= - \sum_{i=0}^{J-1}\beta_i^{n} c^{e,n+1}_h\left(\bm{v}_h,\bm{u}^{n-i}_h,\bm{u}_{\mathrm{G},h}^{n+1};\bm{g}_u^{n+1}\right)
+ \intelenp{\bm{v}_h}{\bm{f}(t_{n+1})} ,
\end{aligned}\label{eq:DualSplitting_ConvectiveStep_WeakForm}\\
\begin{aligned}
l_{h,\text{hom}}^{e,n+1}\left(q_h,p_h^{n+1}\right) 
= - \frac{\gamma_0^n}{\Delta t_n} d_{h}^{e,n+1}\left(q_h,\hat{\bm{u}}_h;\bm{g}_{\hat{u}}^{n+1}\right)
 - l_{h,\text{inhom}}^{e,n+1}\left(q_h;g_{p}^{n+1},h_{p}^{n+1}\right) ,
\end{aligned}\label{eq:DualSplitting_Pressure_WeakForm}\\
\begin{aligned}
\intelenp{\bm{v}_h}{\hat{\hat{\bm{u}}}_h} 
= \intelenp{\bm{v}_h}{\hat{\bm{u}}_h}-\frac{\Delta t_n}{\gamma_0^n} g_h^{e,n+1}\left(\bm{v}_h,p_h^{n+1};g_p^{n+1}\right) ,
\end{aligned}\label{eq:DualSplitting_Projection_WeakForm}\\
\begin{aligned}
\intelenp{\bm{v}_h}{\frac{\gamma_0^n}{\Delta t_n} \hat{\hat{\hat{\bm{u}}}}_h}
+ v^{e,n+1}_{h,\text{hom}}\left(\bm{v}_h,\hat{\hat{\hat{\bm{u}}}}_h\right)
= \intelenp{\bm{v}_h}{\frac{\gamma_0^n}{\Delta t_n}\hat{\hat{\bm{u}}}_h}
- v_{h,\text{inhom}}^{e,n+1}(\bm{v}_h;\bm{g}_u^{n+1},\bm{h}_u^{n+1}) ,
\end{aligned}\label{eq:DualSplitting_ViscousStep_WeakForm} \\
\begin{aligned}
\intelenp{\bm{v}_h}{\bm{u}_h^{n+1}}+ a^{e,n+1}_{\mathrm{D},h}\left(\bm{v}_h,\bm{u}_h^{n+1}\right)\Delta t_n
+ a^{e,n+1}_{\mathrm{C},h,\mathrm{hom}}\left(\bm{v}_h,\bm{u}_h^{n+1}\right)\Delta t_n  =\\ \intelenp{\bm{v}_h}{\hat{\hat{\hat{\bm{u}}}}_h} - a^{e,n+1}_{\mathrm{C},h,\mathrm{inhom}}\left(\bm{v}_h;\bm{g}_u^{n+1}\right)\Delta t_n  . 
\end{aligned}\label{eq:DualSplitting_Penalty_WeakForm}
\end{align}
The Laplace operator~$l_{h}^{e}$ appearing in the pressure Poisson equation is discretized using the SIPG method
\begin{align}
\begin{split}
l_h^e\left(q_h,p_h;g_{p}^{n+1},h_{p}^{n+1}\right) =& \intele{\Grad{q_h}}{\Grad{p_h}}
-\inteleface{\Grad{q_h}}{\frac{1}{2}\jump{p_h}}
- \inteleface{q_h}{\avg{\Grad{p_h}}\cdot\bm{n}}\\
&+ \inteleface{q_h}{\tau\jump{p_h}\cdot\bm{n}}	\; ,
\end{split}\label{WeakFormulationLaplace}
\end{align}
and is again split into homogeneous contributions
\begin{align}
\begin{split}
l_{h,\mathrm{hom}}^e\left(q_h,p_h\right) =& 
\intele{\Grad{q_h}}{\Grad{p_h}}
-\intelefaceInterior{\Grad{q_h}}{\frac{1}{2}\jump{p_h}} - \intelefaceNeumann{\Grad{q_h}}{p_h \bm{n}}\\
&- \intelefaceInterior{q_h}{\avg{\Grad{p_h}}\cdot\bm{n}} - \intelefaceNeumann{q_h}{\Grad{p_h}\cdot\bm{n}} \\
&+ \intelefaceInterior{q_h}{\tau\jump{p_h}\cdot\bm{n}} + \intelefaceNeumann{q_h}{2\tau p_h}	\; .
\end{split}\label{WeakFormulationLaplace_Homogeneous}
\end{align}
and inhomogeneous contributions
\begin{align}
\begin{split}
l_{h,\mathrm{inhom}}^e\left(q_h;g_{p}^{n+1},h_{p}^{n+1}\right) =& 
 \intelefaceNeumann{\Grad{q_h}}{g_p \bm{n}} - \intelefaceDirichlet{q_h}{h_p} 
- \intelefaceNeumann{q_h}{2\tau g_p}	\; .
\end{split}\label{WeakFormulationLaplace_Inhomogeneous}
\end{align}

In our previous work~\cite{Fehn18b}, the penalty terms have been applied in the projection equation~\eqref{eq:DualSplitting_Projection_WeakForm}. However, since we evaluate the continuity penalty operator also on boundary faces in the present work, the penalty terms are evaluated in a postprocessing step, equation~\eqref{eq:DualSplitting_Penalty_WeakForm}. Adding the penalty terms to the projection equation~\eqref{eq:DualSplitting_Projection_WeakForm} would prevent to achieve high-order temporal accuracy since prescribing the boundary condition~$\bm{g}_u$ for the intermediate velocity~$\hat{\hat{\bm{u}}}_h$ would be inconsistent.

Some comments are in order regarding the evaluation of the boundary conditions~$\bm{g}_{\hat{u}}(t_{n+1})$ and~$h_{p}(t_{n+1})$ on the right-hand side of equation~\eqref{eq:DualSplitting_Pressure_WeakForm}. The convective and viscous terms have to be evaluated on~$\partial \Omega_e$ using the finite element expansion of the velocity solution~$\bm{u}_h$ on element~$e$. In the discrete case, the convective term is simply calculated as~$\left(\Grad{\bm{u}_h}\right) \cdot (\bm{u}_h-\bm{u}_{\mathrm{G},h})$ by taking the derivative of the shape functions.
Since the viscous term involves second derivatives, we calculate it in a two-step process, computing the vorticity~$\bm{\omega}_h \in \mathcal{V}^u_h$ in a first step by a local~$L^2$-projection, see~\cite{Krank17, Fehn17}
\begin{align*}
\intele{\bm{v}_h}{\bm{\omega}_h} = \intele{\bm{v}_h}{\nabla \times \bm{u}_h}\; .
\end{align*}
In the second step, the contribution of the viscous term to the pressure Neumann boundary condition is obtained by calculating the curl of the vorticity~$\bm{\omega}_h$ according to equation~\eqref{eq:DualSplitting_PressureBC_GammaD}.

\subsection{Pressure-correction scheme}\label{sec:PressureCorrectionWeakForm}
Finally, the variational formulation is stated for the class of pressure-correction methods: Find~$\hat{\bm{u}}_h,\hat{\hat{\bm{u}}}_h,\bm{u}_h^{n+1}\in\mathcal{V}^u_h$ and~$\phi_h^{n+1},p_h^{n+1}\in\mathcal{V}^p_h$  such that for all~$\bm{v}_h \in \mathcal{V}^{u}_{h,e}$,~$q_h \in \mathcal{V}^{p}_{h,e}$ and for all elements~$e=1,...,N_{\text{el}}$
\begin{align}
\begin{aligned}
\intelenp{\bm{v}_h}{\frac{\gamma_0^n \hat{\bm{u}}_h-\sum_{i=0}^{J-1}\alpha_i^n\bm{u}^{n-i}_h}{\Delta t_n}}
+ c^{e,n+1}_h\left(\bm{v}_h,\hat{\bm{u}}_h,\bm{u}_{\mathrm{G},h}^{n+1};\bm{g}_u^{n+1}\right) 
+ v^{e,n+1}_h\left(\bm{v}_h,\hat{\bm{u}}_h;\bm{g}_u^{n+1},\bm{h}_u^{n+1}\right) \\
= - \sum_{i=0}^{J_p-1} \beta_i^n g^e_h\left(\bm{v}_h,p^{n-i}_h; g_p^{n-i} \right)
+ \intelenp{\bm{v}_h}{\bm{f}(t_{n+1})}  , 
\end{aligned}\label{eq:PressureCorrection_MomentumImplicit_Nonlinear_WeakForm}\\
\begin{aligned}
l_{h,\text{hom}}^{e,n+1}\left(q_h,\phi_h^{n+1}\right) =
- \frac{\gamma_0^n}{\Delta t_n} d_{h}^{e,n+1}\left(q_h,\hat{\bm{u}}_h;\bm{g}_u^{n+1}\right)
- l_{h,\text{inhom}}^{e,n+1}\left(q_h;g_{\phi}^{n+1},h_{\phi}^{n+1}\right) ,
\end{aligned}\label{eq:PressureCorrection_PressureStep_WeakForm}\\
\begin{aligned}
\intelenp{q_h}{p_h^{n+1}} = \intelenp{q_h}{\phi_h^{n+1} + \sum_{i=0}^{J_p-1}\left(\beta_i p_h^{n-i}\right)}
- \zeta \nu\; d_{h}^{e,n+1}\left(q_h,\hat{\bm{u}}_h;\bm{g}_u^{n+1}\right) ,
\end{aligned}\label{eq:PressureCorrection_PressureUpdate_WeakForm}\\
\begin{aligned}
\intelenp{\bm{v}_h}{\hat{\hat{\bm{u}}}_h} 
= \intelenp{\bm{v}_h}{\hat{\bm{u}}_h} -\frac{\Delta t_n}{\gamma_0^n}g_h^{e,n+1}\left(\bm{v}_h,\phi_h^{n+1};g_{\phi}^{n+1}\right) ,
\end{aligned}\label{eq:PressureCorrection_Projection_WeakForm}\\
\begin{aligned}
\intelenp{\bm{v}_h}{\bm{u}_h^{n+1}} + a^{e,n+1}_{\mathrm{D},h}(\bm{v}_h,\bm{u}^{n+1}_h)\Delta t_n + a^{e,n+1}_{\mathrm{C},h,\mathrm{hom}}(\bm{v}_h,\bm{u}^{n+1}_h)\Delta t_n  
=\\
 \intelenp{\bm{v}_h}{\hat{\hat{\bm{u}}}_h} - a^{e,n+1}_{\mathrm{C},h,\mathrm{inhom}}(\bm{v}_h,\bm{g}^{n+1}_u)\Delta t_n  .
\end{aligned}\label{eq:PressureCorrection_Penalty_WeakForm}
\end{align}
Similar to the monolithic solver, we consider an alternative formulation that formulates the convective term explicitly, replacing
\begin{align*}
c^{e,n+1}_h\left(\bm{v}_h,\hat{\bm{u}}_h,\bm{u}_{\mathrm{G},h}^{n+1};\bm{g}_u^{n+1}\right) \rightarrow \sum_{i=0}^{J-1}\beta_i^{n} c^{e,n+1}_h\left(\bm{v}_h,\bm{u}^{n-i}_h,\bm{u}_{\mathrm{G},h}^{n+1};\bm{g}_u^{n+1}\right) \; .
\end{align*}

\subsection{Geometric conservation law}\label{sec:GCL}
It can be proven that the fully discrete ALE-DG methods derived above satisfy the geometric conservation law, i.e., they are able to preserve a constant flow field~\cite{Thomas1979}. In other words, the constant solution~$\bm{u}(\bm{x},t)=\bm{u}_0$,~$p(\bm{x},t)=p_0$ is a solution of the fully discrete formulations for~$\bm{f}=\bm{0}$.

\begin{theorem}
Assume the solution at time instant~$t_n$ is given as~$\bm{u}^{n}_h = \bm{u}_0$,~$p^{n}_h=p_0$ where~${u}_{i,0} = C_i$,~$i=1,...,d$, and~$p_0=C$ (and similarly for previous time instants in case of high-order schemes), and further assume~$\bm{f}=\bm{0}$ and exact numerical integration of the velocity divergence term and pressure gradient term. Then, the fully discrete ALE-DG incompressible Navier--Stokes solvers introduced in Sections~\ref{sec:WeakFormCoupled},~\ref{sec:WeakFormDualSplitting}, and~\ref{sec:PressureCorrectionWeakForm} preserve a constant solution and yield~$\bm{u}^{n+1}_h = \bm{u}_0$,~$p^{n+1}_h=p_0$ at time~$t_{n+1}=t_n + \Delta t_n$, independently of the order of the time integration and extrapolation schemes, and for arbitrary mesh velocities.
\end{theorem}

\begin{proof}
For the monolithic solver, equations~\eqref{eq:WeakForm_CoupledSolution_Momentum} and~\eqref{eq:WeakForm_CoupledSolution_Continuity}, it is obvious that the time derivative term becomes zero since it holds~$\gamma_0^n = \sum_{i=0}^{J-1} \alpha_i^n$. As explained in~\cite{Foerster2006}, this originates from the fact that we discretize the differential form of the ALE equations with time derivative applied to the velocity only, as opposed to formulations that apply the time derivative to an integral quantity. Hence, to complete the proof one needs to show that the weak forms of all individual terms evaluate to zero. This is trivial for the divergence and continuity penalty terms, since these evaluate the divergence inside the element or the jump over interior faces, both vanishing for a constant solution. The volume term of the convective operator contains the gradient, and the face terms vanish due to the conservativity and consistency of the numerical flux. Note that it is enough to consider interior faces in this context, as one realizes that boundary faces behave as interior faces when evaluating the boundary values according to Table~\ref{tab:BCsWeak} for a constant solution. The SIPG discretization of the viscous term also vanishes, as each form either contains gradients or jumps of the solution. The strong formulation of the velocity divergence term, equation~\eqref{eq:StrongFormVelocityDivergence}, and pressure gradient term, equation~\eqref{eq:StrongFormPressureGradient}, vanish for constant solutions, since the volume integrals contain derivatives of the solution and the face integrals contain jumps. The strong formulation therefore always satisfies the (discrete) geometric conservation law. Since strong and weak formulation are only equivalent under the assumption of exact integration, these terms do not vanish exactly for the weak formulation of the velocity--pressure coupling terms, equation~\eqref{WeakFormVelocityDivergence} and~\eqref{WeakFormPressureGradient}, when applying standard Gaussian quadrature rules and considering arbitrarily deformed elements. Since the divergence and continuity penalty terms vanish for a constant solution, the GCL is also fulfilled when applying these terms in a postprocessing step, equation~\eqref{eq:WeakForm_CoupledSolution_Postprocessing}.

Regarding the dual splitting scheme, it follows from the above argumentation that the first sub-step in equation~\eqref{eq:DualSplitting_ConvectiveStep_WeakForm} yields~$\hat{\bm{u}}_h=\bm{u}_0$. Particular attention has to be paid to the boundary conditions~$h_p$ in equation~\eqref{eq:DualSplitting_PressureBC_GammaD} and~$\bm{g}_{\hat{u}}$ in equation~\eqref{eq:DualSplitting_DBC_IntermediateVelocity}. The pressure Neumann boundary condition vanishes for~$\bm{f}=\bm{0}$, as it contains derivatives in either space or time, and we obtain~$\bm{g}_{\hat{u}}=\bm{g}_u$ in the case of a constant solution. Hence, the divergence operator on the right-hand side of the pressure Poisson equation~\eqref{eq:DualSplitting_Pressure_WeakForm} vanishes, and the pressure Poisson equation is satisfied for a constant solution~$p_0$. It immediately follows from the argumentation for the monolithic solver that the remaining sub-steps, equations~\eqref{eq:DualSplitting_Projection_WeakForm},~\eqref{eq:DualSplitting_ViscousStep_WeakForm}, and~\eqref{eq:DualSplitting_Penalty_WeakForm}, yield~$\hat{\hat{\bm{u}}}_h=\bm{u}_0$,~$\hat{\hat{\hat{\bm{u}}}}_h=\bm{u}_0$, and~$\bm{u}_h^{n+1}=\bm{u}_0$, thus preserving a constant flow state.

Turning to the pressure-correction scheme, we first notice that the momentum equation~\eqref{eq:PressureCorrection_MomentumImplicit_Nonlinear_WeakForm} results in~$\hat{\bm{u}}_h=\bm{u}_0$ for reasons explained above. The pressure boundary condition~$g_{\phi}$ becomes zero, as~$\sum_{i=0}^{J_p-1} \beta_i^n = 1$, so that~$\phi_h^{n+1}=0$ is a solution of the pressure Poisson equation~\eqref{eq:PressureCorrection_PressureStep_WeakForm}. The pressure update equation~\eqref{eq:PressureCorrection_PressureUpdate_WeakForm} results in~$p_h^{n+1}=p_0$, since the divergence term on the right-hand side becomes zero. The same holds for the pressure gradient term in equation~\eqref{eq:PressureCorrection_Projection_WeakForm}, since its arguments are~$\phi_h^{n+1}=0$,~$g_{\phi}^{n+1}=0$, so that the constant solution~$\hat{\hat{\bm{u}}}_h=\bm{u}_0$ and~$\bm{u}_h^{n+1}=\bm{u}_0$ is recovered in the last sub-steps. \qed
\end{proof}

\begin{remark}
It is worth noting that no assumption has been made regarding the mesh velocity or how it is computed numerically to show compliance with the GCL. This is also a consequence of the fact that we discretize -- as suggested in~\cite{Foerster2006} -- the differential form of the ALE equations in convective formulation, equation~\eqref{eq:MomentumEquationALE}. In Section~\ref{sec:NumericalResults}, we demonstrate numerically that the geometric conservation is fulfilled exactly for the strong formulations, and that it is not fulfilled exactly down to rounding errors for the weak formulation of the velocity--pressure coupling terms in general. However, since the temporal discretization and spatial discretization are designed to satisfy the GCL, we expect that no relevant difference is observed for practical problems (where the solution is non-constant) due to this variational crime, since those problems always suffer from non-exact integration, and variational crimes have to be accepted in several respects. For example, fulfilling discrete energy stability exactly with respect to the velocity--pressure coupling terms, see Section~\ref{sec:EnergyStability}, requires that one term is formulated in weak form, and the other one in strong form, so that the formulation becomes symmetric independently of integration errors. We also emphasize that previous works have shown that fulfilling the GCL is neither a necessary nor a sufficient condition for the time integrator to preserve its high-order accuracy on moving meshes~\cite{Etienne2009,Geuzaine2003}. We will therefore carefully investigate the temporal convergence behavior of the present ALE schemes and demonstrate that the high-order accuracy of the Navier--Stokes solvers on fixed meshes is preserved on moving meshes when using definition~\eqref{eq:MeshVelocity} to calculate the mesh velocity.
\end{remark}

\subsection{Energy stability}\label{sec:EnergyStability}
As detailed in the introduction, energy stability is a crucial ingredient to obtain a flow solver that is robust in the under-resolved regime. In this section, we analyze the energy balance of the present ALE-DG methods. For this analysis, we assume vanishing body forces,~$\bm{f}=\bm{0}$, and consider the inviscid limit,~$\nu=0$. This is reasonable since the viscous term has a dissipative character both physically and numerically, where the numerically dissipative character immediately follows from the fact that the SIPG discretization of the viscous term, equation~\eqref{eq:weak_form_viscous_term}, is positive definite and symmetric. Therefore, the critical case is to investigate whether a numerical method is energy stable in the absence of viscous dissipation. For simplicitly, we also assume periodic boundaries (including periodicity of the grid velocity). Under these assumptions and for smooth solutions, the ALE incompressible Navier--Stokes equations~\eqref{eq:MomentumEquationALE} and~\eqref{eq:ContinuityEquationALE} are energy-conserving in the following sense
\begin{align}
\left.\frac{\partial }{\partial t}\int_{\Omega_0} \frac{1}{2}\bm{u} \cdot \bm{u} \det\bm{J} \; \mathrm{d}\Omega \right\vert_{\boldsymbol{\chi}} = \int_{\Omega(t)} \left. \frac{\partial \frac{1}{2}\bm{u} \cdot \bm{u}}{\partial t}\right\vert_{\boldsymbol{\chi}} \mathrm{d}\Omega + \int_{\Omega(t)} \frac{1}{2}\left( \bm{u} \cdot \bm{u} \right) \Div{\bm{u}_{\mathrm{G}}} \; \mathrm{d}\Omega = 0 \; , \label{eq:EnergyConservation}
\end{align}
where~$\bm{J} =\partial \bm{x}/\partial \boldsymbol{\chi}$ is the Jacobian. A derivation of this relation is shown in~\ref{sec:EnergyConservation}. A spatially discretized ALE incompressible Navier--Stokes solver 
\begin{align}
\intele{\bm{v}_h}{\left.\frac{\partial \bm{u}_h}{\partial t}\right\vert_{\boldsymbol{\chi}}}
+ c^{e}_h\left(\bm{v}_h,\bm{u}_h,\bm{u}_{\mathrm{G},h}\right)
+ g^{e}_h\left(\bm{v}_h,p_h\right)
+ a^{e}_{\mathrm{D},h}(\bm{v}_h,\bm{u}_h)
+ a^{e}_{\mathrm{C},h}(\bm{v}_h,\bm{u}_h) &= \bm{0} \; , \label{eq:WeakForm_Momentum_EnergyBalance} \\
-d^{e}_h(q_h,\bm{u}_h)& = 0  \; , \label{eq:WeakForm_Continuity_EnergyBalance}
\end{align}
is called energy-stable if it fulfills the following discrete analogy
\begin{align}
\int_{\Omega_h(t)} \left. \frac{\partial \frac{1}{2}\bm{u}_h \cdot \bm{u}_h}{\partial t}\right\vert_{\boldsymbol{\chi}} \mathrm{d}\Omega + \sum_{e=1}^{N_{\mathrm{el}}}\int_{\Omega_e(t)} \frac{1}{2}\left( \bm{u}_h \cdot \bm{u}_h \right) \Div{\bm{u}_{\mathrm{G},h}} \mathrm{d}\Omega \leq 0 \; . \label{eq:DiscreteEnergyStability}
\end{align}
We begin with reformulating the left term in equation~\eqref{eq:DiscreteEnergyStability} so that the semi-discrete momentum equation~\eqref{eq:WeakForm_Momentum_EnergyBalance} can be inserted
\begin{align}
\int_{\Omega_h(t)} \left. \frac{\partial \frac{1}{2}\bm{u}_h \cdot \bm{u}_h}{\partial t}\right\vert_{\boldsymbol{\chi}} \mathrm{d}\Omega &= \int_{\Omega_h(t)} \left. \bm{u}_h \cdot \frac{\partial \bm{u}_h}{\partial t}\right\vert_{\boldsymbol{\chi}} \mathrm{d}\Omega = \sum_{e=1}^{N_{\mathrm{el}}} \intele{\bm{u}_h}{\left.\frac{\partial \bm{u}_h}{\partial t}\right\vert_{\boldsymbol{\chi}}}  \nonumber \\
&= - \sum_{e=1}^{N_{\mathrm{el}}} \left(
c^{e}_h\left(\bm{u}_h,\bm{u}_h,\bm{u}_{\mathrm{G},h}\right)
+ g^{e}_h\left(\bm{u}_h,p_h\right)
+ a^{e}_{\mathrm{D},h}(\bm{u}_h,\bm{u}_h)
+ a^{e}_{\mathrm{C},h}(\bm{u}_h,\bm{u}_h)\right) \; .\label{eq:EnergyStability1}
\end{align}
Regarding the pressure gradient term, we first observe that the following relation holds
\begin{align}
\sum_{e=1}^{N_{\mathrm{el}}} g^{e}_{h,\mathrm{strong}}\left(\bm{u}_h,p_h\right)
= 
 \intdomain{\bm{u}_h}{\Grad{p_h}} - \intinteriorfaces{\avg{\bm{u}_h}}{\jump{p_h}} =  - \sum_{e=1}^{N_{\mathrm{el}}} d^{e}_{h,\mathrm{weak}}\left(p_h,\bm{u}_h\right) = 0 \; ,
\label{eq:EnergyStabilityPressureGradientTerm}
\end{align}
and similarly between~$g^{e}_{h,\mathrm{weak}}$ and~$d^{e}_{h,\mathrm{strong}}$. For these combinations of the velocity--pressure coupling terms, the formulation is symmetric independently of integration errors, and the pressure gradient term will not contribute to the energy evolution due to the discrete continuity equation~\eqref{eq:WeakForm_Continuity_EnergyBalance}.

The divergence and continuity penalty terms have a dissipative character as these are positive semi-definite by definition. Hence, it remains to consider the convective term. 
Inserting equation~\eqref{eq:weak_form_convective_term} and reformulating yields
\begin{align}
\begin{split}
\sum_{e=1}^{N_{\mathrm{el}}} c^{e}_h\left(\bm{u}_h,\bm{u}_h,\bm{u}_{\mathrm{G},h}\right) =  &+ \sum_{e=1}^{N_{\mathrm{el}}} \left(
\intele{\bm{u}_h}{\left(\Grad{\bm{u}_h}\right) \cdot \left(\bm{u}_h - \bm{u}_{\mathrm{G},h}\right)} 
- \inteleface{\bm{u}_h}{\left(\left(\avg{\bm{u}_h}-\bm{u}_{\mathrm{G},h}\right)\cdot\bm{n}\right) \bm{u}_h}\right)\\
&+ \sum_{e=1}^{N_{\mathrm{el}}}\inteleface{\bm{u}_h}{\left(\left(\avg{\bm{u}_h}-\bm{u}_{\mathrm{G},h}\right)\cdot\bm{n}\right)\avg{\bm{u}_h}
+ \frac{1}{2} \vert \left(\avg{\bm{u}_h}-\bm{u}_{\mathrm{G},h}\right)\cdot\bm{n} \vert \jumporiented{\bm{u}_h}} \\
= & + \sum_{e=1}^{N_{\mathrm{el}}} \left(
\intele{\bm{u}_h}{\left(\Grad{\bm{u}_h}\right) \cdot \left(\bm{u}_h - \bm{u}_{\mathrm{G},h}\right)} 
- \inteleface{\bm{u}_h}{\left(\left(\avg{\bm{u}_h}-\bm{u}_{\mathrm{G},h}\right)\cdot\bm{n}\right) \frac{1}{2}\jumporiented{\bm{u}_h}}\right)\\
&+ \intinteriorfaces{\jumporiented{\bm{u}_h}}{\frac{1}{2} \vert \left(\avg{\bm{u}_h}-\bm{u}_{\mathrm{G},h}\right)\cdot\bm{n} \vert \jumporiented{\bm{u}_h}} \; .
\end{split}
\label{eq:ConvectiveTerm1}
\end{align}
In~\ref{sec:EnergyBalanceConvectiveTerm}, we show that the first row on the right-hand side of the above equation can be reformulated as follows by algebraic manipulations
\begin{align}
\begin{split}
\sum_{e=1}^{N_{\mathrm{el}}} \left(
\intele{\bm{u}_h}{\left(\Grad{\bm{u}_h}\right) \cdot \left(\bm{u}_h - \bm{u}_{\mathrm{G},h}\right)} 
- \inteleface{\bm{u}_h}{\left(\left(\avg{\bm{u}_h}-\bm{u}_{\mathrm{G},h}\right)\cdot\bm{n}\right) \frac{1}{2}\jumporiented{\bm{u}_h}}\right)=\\
=-\frac{1}{2}\intdomain{\Div{\left(\bm{u}_h-\bm{u}_{\mathrm{G},h}\right)}}{\bm{u}_h\cdot\bm{u}_h}
+\frac{1}{2}\intinteriorfaces{\jumporiented{\bm{u}_h}\cdot \bm{n}}{\avg{\bm{u}_h\cdot \bm{u}_h}} \; ,
\end{split}\label{eq:ConvectiveTermReformulation}
\end{align}
which is easier to interpret in terms of energy stability since this formulation contains the divergence of the velocity, i.e., a residual of the incompressible Navier--Stokes equations, and the divergence of the grid velocity. It is worth emphasizing that the face integrals related to the moving mesh dropped out completely, since the ALE term is a linear transport term. The volume integral of the ALE transport term does not drop out since the grid velocity is not divergence-free. However, one can see that this is exactly the second term in~\eqref{eq:DiscreteEnergyStability}. Inserting equations~\eqref{eq:EnergyStabilityPressureGradientTerm},~\eqref{eq:ConvectiveTerm1}, and~\eqref{eq:ConvectiveTermReformulation} into equation~\eqref{eq:EnergyStability1} yields the following result
\begin{align}
\begin{split}
\int_{\Omega_h(t)} \left. \frac{\partial \frac{1}{2}\bm{u}_h \cdot \bm{u}_h}{\partial t}\right\vert_{\boldsymbol{\chi}} \mathrm{d}\Omega & + \sum_{e=1}^{N_{\mathrm{el}}}\int_{\Omega_e(t)} \frac{1}{2}\left( \bm{u}_h \cdot \bm{u}_h \right) \Div{\bm{u}_{\mathrm{G},h}} \mathrm{d}\Omega = \\
= 
& + \frac{1}{2}\intdomain{\Div{\bm{u}_h}}{\bm{u}_h\cdot\bm{u}_h} - a_{\mathrm{D},h}(\bm{u}_h,\bm{u}_h)\\
& - \frac{1}{2}\intinteriorfaces{\jumporiented{\bm{u}_h}\cdot \bm{n}}{\avg{\bm{u}_h\cdot \bm{u}_h}} - a_{\mathrm{C},h}(\bm{u}_h,\bm{u}_h)\\
& - \intinteriorfaces{\jumporiented{\bm{u}_h}}{\frac{1}{2} \vert \left(\avg{\bm{u}_h}-\bm{u}_{\mathrm{G},h}\right)\cdot\bm{n} \vert \jumporiented{\bm{u}_h}} 
\; .
\end{split}
\end{align}
It is interesting to realize that the result derived in~\cite{Fehn18a} for the conservative formulation of the convective term in the Eulerian case is very similar to the convective formulation in ALE form considered here, where terms with exactly the same structure occur. In particular, one observes that the ALE formulation does not introduce new terms on the right-hand side as compared to the Eulerian case (a consequence of the fact that the additional ALE term is a linear transport term). One might therefore argue that energy stability for the Eulerian case translates into energy stability for the ALE case with moving meshes. The third row on the right-hand side is the upwind stabilization term of the convective term and always exhibits a dissipative behavior. However, it is well-known that this term is not able to render the nonlinear convective term energy-stable. As argued in~\cite{Fehn18a,Fehn19Hdiv}, the consistent divergence and continuity penalty terms are positive semi-definite and control the non-vanishing divergence and non-vanishing jumps of the velocity in normal direction. By the use of consistent penalty terms, energy stability of the DG discretization is enforced weakly; there is currently no proof that the kinetic energy is strictly non-increasing at all times. The sign-indefinite terms of the above equation are exactly zero for an~$H(\mathrm{div})$-conforming (normal continuous) velocity space together with a pressure space that ensures that the velocity is pointwise divergence-free. For these reasons, the stabilized approach with divergence and continuity penalty terms can also be denoted as~$H(\mathrm{div})$-stabilization.

\subsection{Numerical integration and implementation}\label{sec:NumericalIntegration}
Integrals in the variational form are evaluated numerically by means of Gaussian quadrature, where we choose the number of one-dimensional quadrature points to ensure exact integration on affine element geometries with constant Jacobian. The velocity mass matrix term, the body force term, the viscous term, the velocity divergence term, the pressure gradient term, and the two penalty terms are integrated with~$n_{\rm{q}} = k_u+1$ quadrature points. To avoid aliasing effects, we use~$n_{\rm{q}}= \lfloor \frac{3k_u}{2} \rfloor +1$ quadrature points for the convective term containing quadratic nonlinearities. The Laplace operator in the pressure Poisson equation and the pressure mass matrix operator are integrated with~$n_{\rm{q}}=k_p+1$ quadrature points. The present ALE-DG methods are implemented in~\texttt{C++} using the~\texttt{deal.II} finite element library~\cite{dealII90}, and especially the matrix-free evaluation techniques developed in~\cite{Kronbichler2019fast} for the evaluation of volume and face integrals of discretized DG operators. State-of-the-art iterative solvers are used to solve the (non-)linear systems of equations of the fully discrete problem.

\subsection{CFL condition}
For the formulations with explicit treatment of the convective term, the time step size is restricted according to the CFL condition. Since the transport velocity is~$\bm{u}_h-\bm{u}_{\mathrm{G},h}$ in the ALE case, the CFL condition used on static meshes has to be adjusted accordingly using the relative velocity between fluid and grid motion. Here, we distinguish between two types of CFL condition. The global CFL condition~\cite{Shahbazi07,Fehn18a} applied to the ALE case
\begin{align}
\Delta t = \frac{\mathrm{Cr}}{k_u^{1.5}}\frac{h_{\mathrm{min}}}{\Vert \bm{u}_h-\bm{u}_{\mathrm{G},h} \Vert_{\mathrm{max}}} \; ,\label{eq:global_CFL_Condition}
\end{align}
with global estimates of the minimum element length~$h_{\mathrm{min}}$ and maximum velocity~$\Vert \bm{u}_h-\bm{u}_{\mathrm{G},h} \Vert_{\mathrm{max}}$ is used for time stepping with constant~$\Delta t$. In the above equation,~$\mathrm{Cr}$ denotes the Courant number and the term~$k_u^{1.5}$ was found to describe well the relation between critical time step size and polynomial degree for the present DG discretization~\cite{Fehn18a}. Since the minimum element length and the maximum velocity are difficult to estimate a priori and since the maximum velocity does not necessarily occur in the smallest element, a local CFL condition is used in case of adaptive time-stepping~\cite{KarniadakisSherwin2005}
\begin{align}
\Delta t = \min_{e=1,...,N_{\mathrm{el}}} \Delta t_e, \, \Delta t_e = \min_{q=1,...,N_{q,e}} \frac{\mathrm{Cr}}{k_u^{1.5}}\left.\frac{h}{\Vert \bm{u}_h-\bm{u}_{\mathrm{G},h} \Vert }\right\vert_{q,e}\; ,\label{eq:local_CFL_Condition}
\end{align}
with the local velocity-to-mesh-size ratio~$\left.\frac{\Vert \bm{u}_h-\bm{u}_{\mathrm{G},h} \Vert}{h}\right\vert_{q,e} = \Vert \bm{J}^{-\mathsf{T}} (\bm{u}_h-\bm{u}_{\mathrm{G},h}) \Vert_{q,e}$ evaluated at quadrature point~$q$ of element~$e$. This CFL condition ensures that the time step size does not exceed the critical one in any element in any quadrature point.

\section{Numerical results}\label{sec:NumericalResults}
The aim of this section is to display the numerical discretization properties of the proposed ALE-DG incompressible Navier--Stokes solvers. We select a set of academic test cases that address different aspects of ALE solvers, with the goal to obtain a picture as complete as possible. In detail, we study the geometric conservation property by the example of the free stream preservation test in Section~\ref{sec:ResultsGCL}. The convergence behavior in terms of temporal as spatial convergence rates are investigated in Section~\ref{sec:ResultsConvergenceInTimeAndSpace} by the example of a two-dimensional vortex problem with moving Dirichlet and Neumann boundaries. In this section, we also test the robustness of the different incompressible Navier--Stokes solvers in the limit of large grid deformations. Finally, the robustness of the proposed discretization methods for under-resolved turbulent flows is studied in Section~\ref{sec:ResultsUnderresolvedTurbulentFlows} by considering the three-dimensional Taylor--Green vortex problem for viscous flows at~$\mathrm{Re}=1600$, and also in the very challenging inviscid limit.

Convergence rates and relative~$L^2$-errors for problems with known analytical solution are computed as defined in~\cite{Fehn17}. Solver tolerances are selected to not spoil accuracy, e.g., by choosing relative solver tolerances of~$10^{-6}$ and absolute solver tolerances of~$10^{-12}$. In case of pure Dirichlet boundary conditions, the pressure level is undefined and can be fixed, e.g., by setting the mean value of the pressure DoF vector to zero.

\subsection{Geometric conservation law -- free stream preservation test}\label{sec:ResultsGCL}
We study the free stream preservation test to investigate whether the ALE formulations derived above fulfill the geometric conservation law. Satisfying the geometric conservation law means that a constant flow field is not disturbed by a moving mesh, i.e., the solver is able to preserve the free stream flow conditions exactly and independently of the mesh deformation. The analytical solution of the free stream preservation test is  therefore the constant flow state
\begin{align*}
\bm{u}(\bm{x}, t) = \left(1, \hdots, 1\right)^{\mathsf{T}} , \; p(\bm{x},t) = 1 \; ,
\end{align*}
where we prescribe pure Dirichlet boundary conditions~$\bm{u}(\bm{x}, t) = \bm{g}_u(\bm{x}, t)$ on~$\hboundary{D}(t)=\hboundary{}(t)$. The computational domain at initial time is~$\Omega_0 = \Omega(t=0) = \left[-L/2, L/2\right]^2$ with~$L=1$. The simulation is run over a time interval of~$0 \leq t \leq T=10$. The following analytical mesh movement with sine functions in both time and space is prescribed in two space dimensions
\begin{align}
\bm{x}(\bm{\chi}, t) = \bm{\chi} + A \sin\left({2\pi}\frac{t}{T_{\mathrm{G}}}\right) \begin{pmatrix}
\sin\left({2\pi}\frac{\chi_2 + L/2}{L}\right)\\
\sin\left({2\pi}\frac{\chi_1 + L/2}{L}\right)
\end{pmatrix}\; ,\label{eq:mesh_motion_2d}
\end{align}
and in three space dimensions
\begin{align}
\bm{x}(\bm{\chi}, t) = \bm{\chi} + A \sin\left({2\pi}\frac{t}{T_{\mathrm{G}}}\right) \begin{pmatrix}
\sin\left({2\pi}\frac{\chi_2 + L/2}{L}\right) \sin\left({2\pi}\frac{\chi_3 + L/2}{L}\right)\\
\sin\left({2\pi}\frac{\chi_1 + L/2}{L}\right) \sin\left({2\pi}\frac{\chi_3 + L/2}{L}\right)\\
\sin\left({2\pi}\frac{\chi_1 + L/2}{L}\right) \sin\left({2\pi}\frac{\chi_2 + L/2}{L}\right)
\end{pmatrix}\; ,\label{eq:mesh_motion_3d}
\end{align}
where the amplitude is set to~$A=0.08$ resulting in a strongly deformed mesh. The period length of the grid motion is set to~$T_{\mathrm{G}} = T/10$ and the wavenumber in space is chosen such that the length and height of the domain are exactly one period. In Figure~\ref{fig:mesh_deformed}, the mesh deformation is illustrated for~$d=2$, where the initial, undeformed mesh is a uniform Cartesian grid. The mesh reaches its maximum deformation at times~$t=T_{\mathrm{G}}/4 + i\ T_{\mathrm{G}}/2, i=0,1,2,...$. The viscosity is set to~$\nu=0.025$. Adaptive time-stepping is used where the time step size is adjusted dynamically according to the CFL condition~\eqref{eq:local_CFL_Condition} using~$\mathrm{Cr}=0.25$.
\begin{figure}
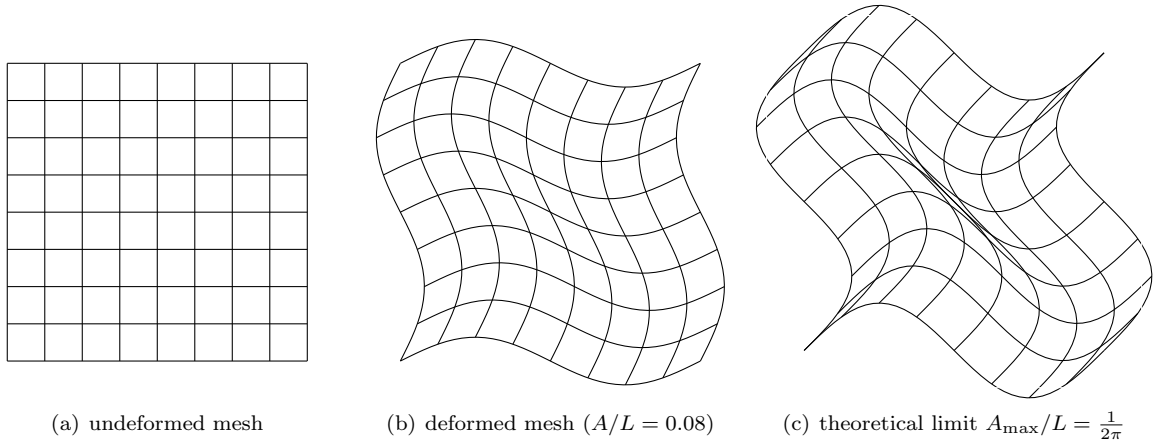

	\centering
     	\subfigure[undeformed mesh]{\plotgridsine{0.0}\label{fig:mesh_undeformed}}
		\subfigure[deformed mesh ($A/L=0.08$)]{\plotgridsine{0.08}\label{fig:mesh_deformed}}
		\subfigure[theoretical limit $A_{\mathrm{max}}/L = \frac{1}{2 \pi}$]{\plotgridsine{0.1591549431}\label{fig:mesh_max_deformed}}
		\caption{Illustration of sine-like mesh motion in two space dimensions according to equation~\eqref{eq:mesh_motion_2d}.}
		\label{fig:undeformed_vs_deformed_mesh}
\end{figure}
In the following, we use a mesh with~$8^d$ elements as in Figure~\ref{fig:undeformed_vs_deformed_mesh} and consider a polynomial degree of~$k=3$. Both absolute and relative solver tolerances are set to a small value of~$10^{-14}$. Table~\ref{tab:geometric_conservation_law} reports relative errors for velocity and pressure for BDF schemes of order 1 to 3 and the three different incompressible Navier--Stokes solvers considered in this work. For the dual splitting scheme,~$J_p=\min(2, J)$ is used but the simulations have also been stable for the choice~$J_p=J$ for the high-order BDF scheme~$J=3$. For the pressure-correction scheme,~$J_p=\min(2, J)-1$ is used for all~$J$ and the rotational formulation with~$\chi=1$. Here, the choice~$J_p=J-1$ lead to instabilities for the high-order scheme~$J=3$ in agreement with theory.

\begin{table}[!h]
\caption{Numerical results for free stream preservation test for both strong and weak formulations of velocity--pressure coupling terms and two- and three-dimensional problems:~$J_p=\min(2, J)$ is used for the dual splitting scheme, and~$J_p=\min(2, J)-1$ for the pressure-correction scheme in incremental formulation. An explicit formulation of the convective term is used for all three solvers.}\label{tab:geometric_conservation_law}
\renewcommand{\arraystretch}{1.1}
\begin{center}
\begin{small}

\subtable[Two-dimensional problem ($d=2$), weak formulations~$d^e_{h, \mathrm{weak}}$ and~$g^e_{h, \mathrm{weak}}$]{
\begin{tabular}{lcccccccc}
\toprule
& \multicolumn{3}{c}{relative~$L^2$-error~$\bm{u}_h$} & &\multicolumn{3}{c}{relative~$L^2$-error~$p_h$}\\
\cline{2-4} \cline{6-8}  & BDF1 & BDF2 & BDF3 &  & BDF1 & BDF2 & BDF3 \\
\midrule
monolithic           & 1.8E--16 & 1.8E--15 & 2.1E--15 & & 1.3E--16 & 6.0E--13 & 2.9E--13\\
dual splitting       & 1.2E--15 & 1.6E--15 & 1.8E--15 & & 2.5E--13 & 7.5E--14 & 6.1E--13\\
pressure-correction  & 1.6E--15 & 2.7E--15 & 3.0E--15 & & 1.8E--13 & 9.5E--14 & 4.6E--13\\
\bottomrule
\end{tabular}}

\subtable[Two-dimensional problem ($d=2$), strong formulations~$d^e_{h, \mathrm{strong}}$ and~$g^e_{h, \mathrm{strong}}$]{
\begin{tabular}{lcccccccc}
\toprule
& \multicolumn{3}{c}{relative~$L^2$-error~$\bm{u}_h$} & &\multicolumn{3}{c}{relative~$L^2$-error~$p_h$}\\
\cline{2-4} \cline{6-8}  & BDF1 & BDF2 & BDF3 &  & BDF1 & BDF2 & BDF3 \\
\midrule
monolithic           & 1.8E--16 & 1.1E--15 & 1.8E--15 & & 1.3E--16 & 2.1E--13 & 2.4E--13\\
dual splitting       & 1.0E--15 & 1.3E--15 & 1.4E--15 & & 1.5E--13 & 3.8E--13 & 6.5E--13\\
pressure-correction  & 1.0E--15 & 1.6E--15 & 1.8E--15 & & 2.3E--13 & 2.3E--13 & 4.7E--13\\
\bottomrule
\end{tabular}}

\subtable[Three-dimensional problem ($d=3$), weak formulations~$d^e_{h, \mathrm{weak}}$ and~$g^e_{h, \mathrm{weak}}$]{
\begin{tabular}{lcccccccc}
\toprule
& \multicolumn{3}{c}{relative~$L^2$-error~$\bm{u}_h$} & &\multicolumn{3}{c}{relative~$L^2$-error~$p_h$}\\
\cline{2-4} \cline{6-8}  & BDF1 & BDF2 & BDF3 &  & BDF1 & BDF2 & BDF3 \\
\midrule
monolithic           & 1.3E--08 & 1.3E--08 & 1.3E--08 & & 8.7E--09 & 8.2E--09 & 8.2E--09\\
dual splitting       & 1.1E--08 & 1.2E--08 & 1.2E--08 & & 6.4E--09 & 6.6E--09 & 6.8E--09\\
pressure-correction  & 1.2E--08 & 1.3E--08 & 1.3E--08 & & 7.0E--09 & 7.1E--08 & 6.1E--08\\
\bottomrule
\end{tabular}}

\subtable[Three-dimensional problem ($d=3$), strong formulations~$d^e_{h, \mathrm{strong}}$ and~$g^e_{h, \mathrm{strong}}$]{
\begin{tabular}{lcccccccc}
\toprule
& \multicolumn{3}{c}{relative~$L^2$-error~$\bm{u}_h$} & &\multicolumn{3}{c}{relative~$L^2$-error~$p_h$}\\
\cline{2-4} \cline{6-8}  & BDF1 & BDF2 & BDF3 &  & BDF1 & BDF2 & BDF3 \\
\midrule
monolithic           & 3.5E--16 & 1.4E--14 & 1.2E--14 & & 1.7E--16 & 8.5E--13 & 8.4E--13\\
dual splitting       & 3.5E--16 & 1.8E--15 & 1.1E--15 & & 1.4E--13 & 7.3E--13 & 1.3E--12\\
pressure-correction  & 1.9E--15 & 1.9E--15 & 2.3E--15 & & 3.3E--13 & 7.8E--13 & 1.1E--12\\
\bottomrule
\end{tabular}}

\end{small}
\end{center}
\renewcommand{\arraystretch}{1}
\end{table}
The results in Table~\ref{tab:geometric_conservation_law} reveal that all schemes fulfill the geometric conservation law for~$d=2$, and in particular also for the weak formulation of the velocity--pressure coupling terms. For~$d=3$, the geometric conservation law is fulfilled exactly only when using the strong formulations~$d^e_{h, \mathrm{strong}}$ and~$g^e_{h, \mathrm{strong}}$ as expected theoretically, see Section~\ref{sec:GCL}. Errors larger than the solver tolerances are observed for the weak formulations for~$d=3$. For all variants studied here, similar results are obtained when using an implicit formulation of the convective term for the monolithic solver and the pressure-correction scheme, where we had to slightly relax the solver tolerances to~$10^{-12}$ to ensure convergence of the Newton solver. Hence, it remains to explain why the weak formulation of velocity--pressure coupling terms fulfills the GCL exactly for~$d=2$. As the critical aspect in this context is the exact evaluation of integrals, see Section~\ref{sec:GCL}, it can be conjectured that integrals of the velocity--pressure coupling terms are indeed evaluated exactly for~$d=2$ for constant solutions on deformed elements. However, this is a special case that only occurs for the free stream preservation test due to a solution of lowest polynomial degree, and this does not hold for general non-constant solutions and arbitrarily deformed elements. Hence, we do not pay further attention to this point. In a similar direction, our results show that the weak formulations are very accurate as well, and will therefore be used for the following examples. The reason for this choice is that the conclusion is even stronger when being able to demonstrate optimal convergence behavior for a formulation that appears to be sub-optimal regarding the free stream preservation test.

\subsection{Temporal and spatial convergence behavior}\label{sec:ResultsConvergenceInTimeAndSpace}

Next, we analyze the convergence behavior of the ALE-DG methods and study whether optimal rates of convergence observed in the Eulerian case carry over to moving meshes. For this purpose, we select the two-dimensional vortex problem from~\cite{Hesthaven07}, which is an analytical solution of the two-dimensional incompressible Navier--Stokes equations in the absence of body forces,~$\bm{f}=\bm{0}$, namely
\begin{align}
\begin{split}
\bm{u}(\bm{x},t) &=  \begin{pmatrix}
-\sin(2\pi x_2)\\
+\sin(2\pi x_1)
\end{pmatrix}
\exp\left(-4\nu\pi^2 t\right)\; ,\\
p(\bm{x},t) &=  -\cos(2\pi x_1)\cos(2\pi x_2)\exp\left(-8\nu \pi^2 t\right)\; ,
\end{split}\label{AnalyticalSolutionVortex}
\end{align}
with viscosity set to~$\nu=0.025$, and simulated over the time interval~$0\leq t\leq T=1$. The computational domain at start time is~$\Omega_0=[-L/2,L/2]^2$ with length~$L=1$, and is deformed according to the two-dimensional mesh movement function~\eqref{eq:mesh_motion_2d} with parameters~$A=0.08$ and~$T_{\mathrm{G}}= 4 T$ (maximum deformation reached at end time~$t=T$) unless specified otherwise. Again, a mesh as depicted in Figure~\ref{fig:mesh_deformed} with refinement level~$l$ is used. Sine-like mesh deformations are commonly used to verify high-order ALE-DG implementations, see for example~\cite{Nguyen2010,Mavriplis2011,Schnucke2018arxiv}. In these works, however, the sine functions are defined in a way that the boundaries are not moving. Instead, we intentionally choose a setup for which the boundaries are moving since our goal is to test all parts of the algorithm relevant for FSI. The verification of boundary conditions is particularly relevant for the splitting-type solvers and some effects might not be visible if the boundaries are fixed. For example, if the boundaries are non-moving, the normal vector in equation~\eqref{eq:DualSplitting_PressureBC_GammaD} would not change over time and the ALE transport term would simply drop out since~$\bm{u}_{\mathrm{G}}=\bm{0}$ on the boundary. According to the setup in~\cite{Hesthaven07}, each of the four sides of the square is split into a Dirichlet boundary and a Neumann boundary according to the inflow and outflow sections, respectively. Note that the chosen mesh deformation is in compliance with these boundary conditions.

\begin{figure}[!ht]
 \centering
	\subfigure[$p(\bm{x},T)$ ($l=1$, $k=3$)]{\includegraphics[width=0.225\textwidth]{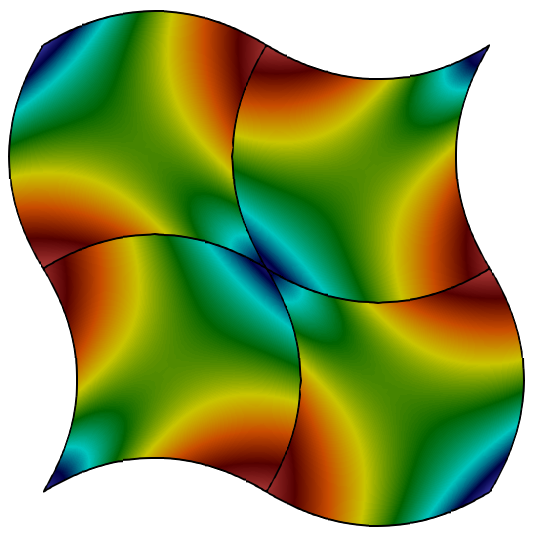}}
	\subfigure[$p(\bm{x},T)$ ($l=2$, $k=3$)]{\includegraphics[width=0.225\textwidth]{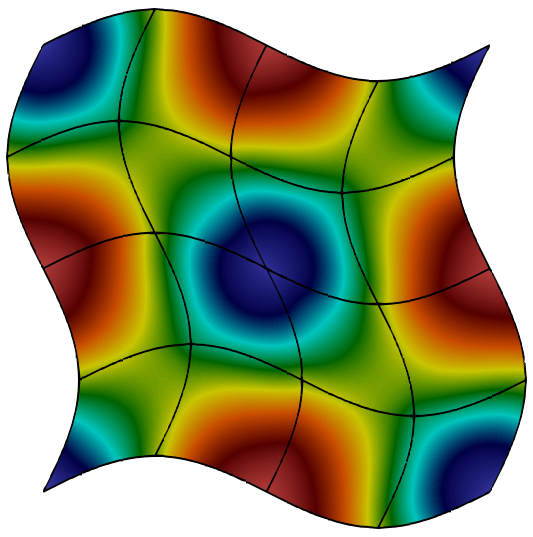}}
	\subfigure[$\Vert \bm{u}(\bm{x},T) \Vert$ ($l=1$, $k=3$)]{\includegraphics[width=0.225\textwidth]{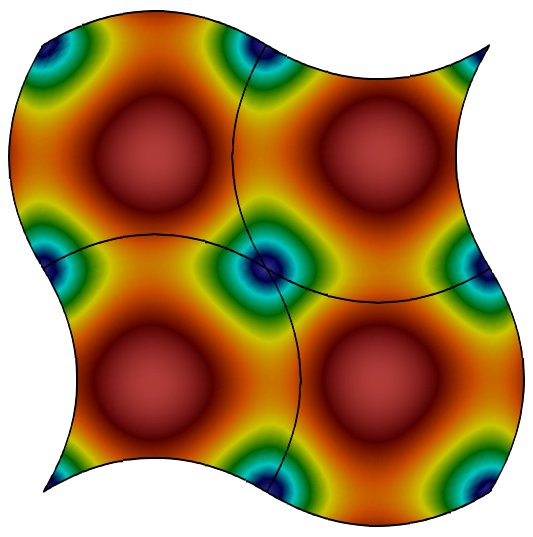}}
	\subfigure[$\Vert \bm{u}(\bm{x},T) \Vert$ ($l=2$, $k=3$)]{\includegraphics[width=0.225\textwidth]{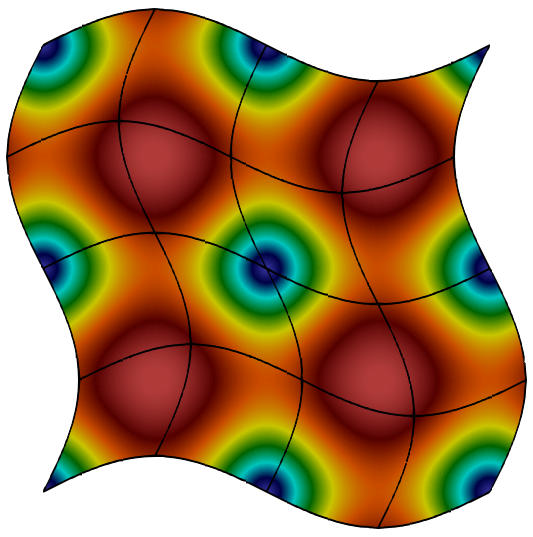}}
\caption{Vortex problem: visualization of solution at final time~$t=T$ for two different mesh resolutions,~$l=1$ and~$l=2$, with polynomial degree~$k=3$ for the velocity and~$2$ for the pressure (red indicates high values and blue low values). The amplitude of the mesh deformation is~$A=0.08$.} 
\label{fig:vortex_visualization}
\end{figure}

Figure~\ref{fig:vortex_visualization} shows a visualization of the solution at the time of maximal mesh deformation~$t=T$ using polynomial shape functions of degree~$k=3$ and considering the two lowest refinement levels of~$l=1,2$ (the grid has to consist of at least~$2^2$ elements due to the type of boundary conditions prescribed with each face of the rectangular domain cut into a Dirichlet part and a Neumann part). While the velocity field is already well resolved on the coarsest mesh with~$l=1$, the pressure field of polynomial degree~$2$ is approximated poorly for refinement level~$l=1$ with distinct discontinuities between the elements. For~$l=2$, the pressure solution appears to be visually converged with only minor differences as compared to the solution on even finer meshes. In the following, we study the convergence quantitatively in terms of errors against the analytical solution as well as convergence rates measured in space and time.

\begin{figure}[t]
 \centering
 \subfigure[constant~$\Delta t$]{
	\includegraphics[width=0.9\textwidth]{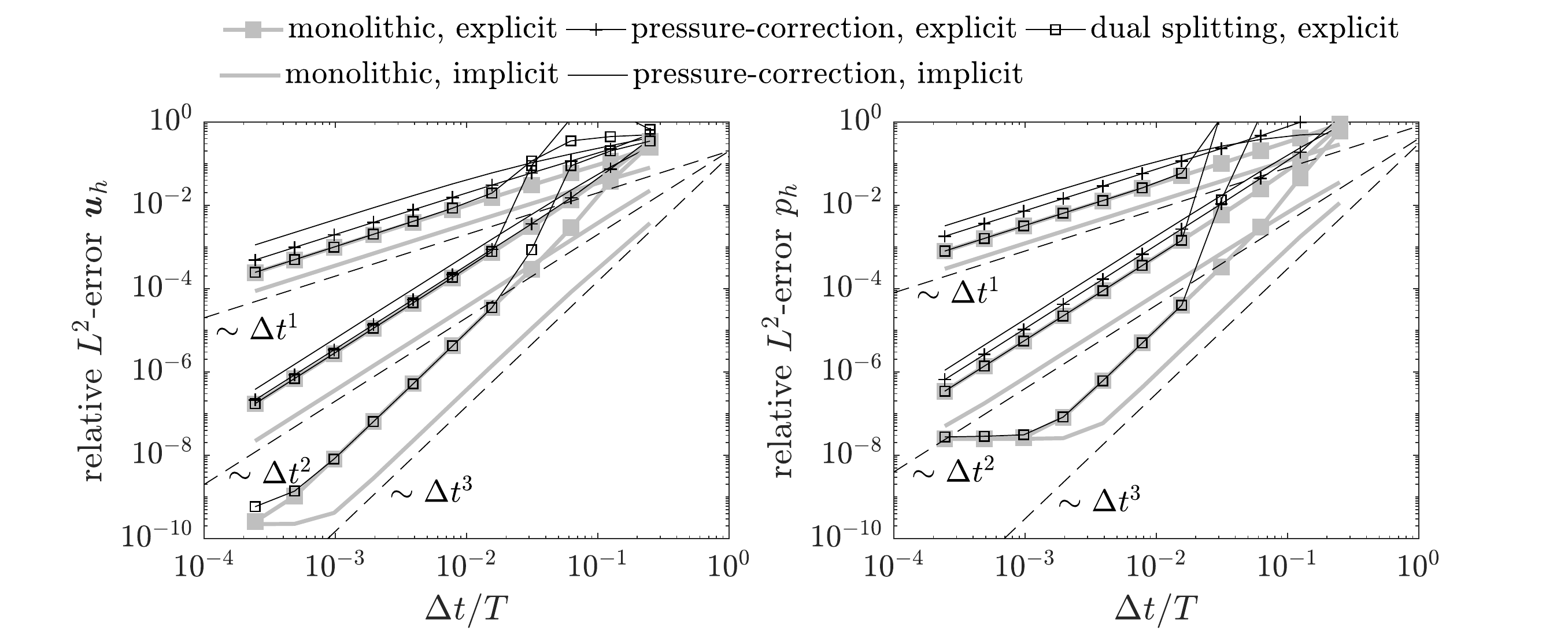}}
 \subfigure[adaptive~$\Delta t_n$]{
	\includegraphics[width=0.9\textwidth]{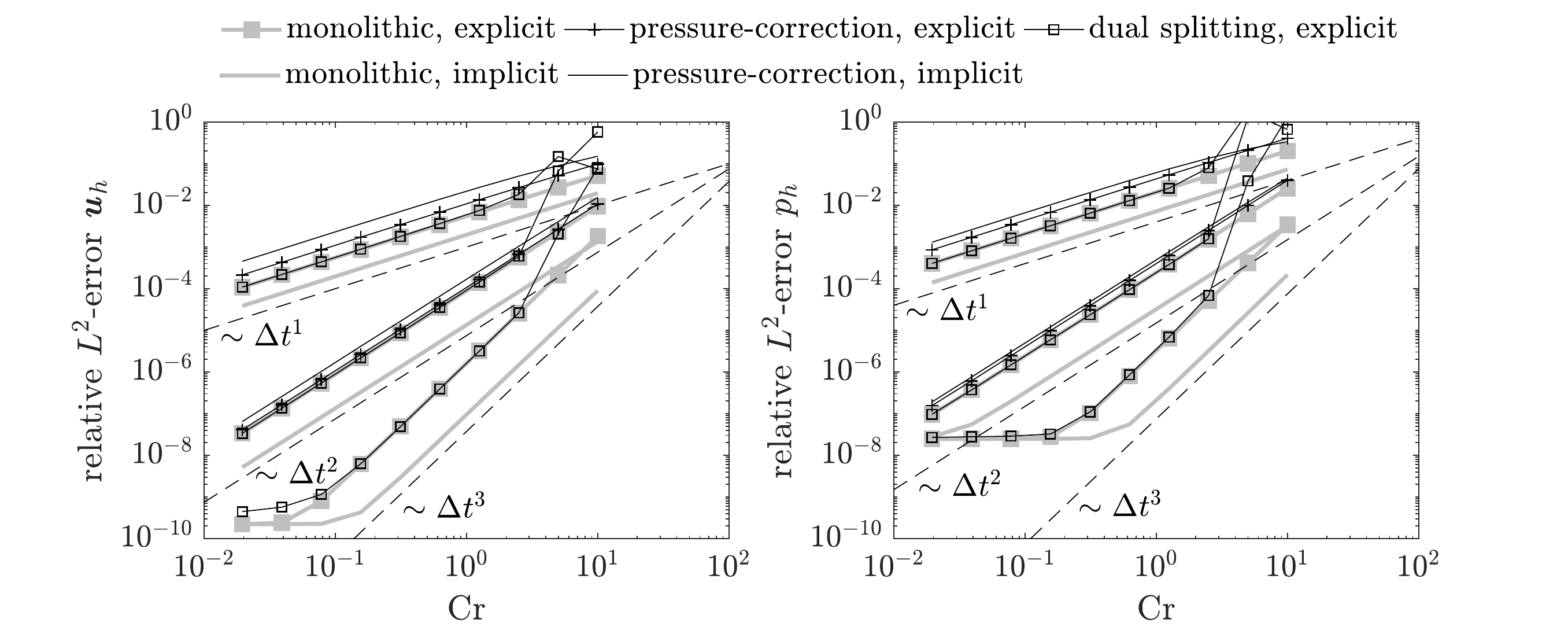}}
\caption{Vortex problem: temporal convergence tests for ALE incompressible Navier--Stokes solvers for BDF schemes of order~$J=1,...,3$ with~$J_p=\min(2, J)$ for the dual splitting scheme and~$J_p=\min(2, J)-1, J\leq 2$ for the pressure-correction scheme.}
\label{fig:vortex_temporal_convergence_ale}
\end{figure}

In a first set of experiments, we test the temporal convergence behavior for both constant and adaptive time step sizes. The chosen spatial resolution is fine,~$l=3$ and~$k=8$, to make sure that errors are dominated by temporal discretization errors. The second aspect why this test case is interesting is the fact that the CFL condition does often not show up for this particular test case for moderate Reynolds numbers, probably due to the fact that the vortex is not moving. This allows to measure temporal convergence rates of an incompressible Navier--Stokes solver with explicit formulation of the convective term, which would be difficult otherwise because in the regular case with~$\mathrm{Cr} < \mathrm{Cr}_{\mathrm{crit}}$ one is often operating in a regime where temporal discretization errors are already negligible as compared to spatial errors for higher-order time integration schemes.

Figure~\ref{fig:vortex_temporal_convergence_ale} shows results of a temporal convergence study for BDF schemes of order~$J=1,2,3$ using constant and adaptive time step sizes. All types of solvers converge with optimal rates of convergence on the moving mesh. Compared to additional simulations performed for a static, Cartesian mesh, the errors are almost the same and only slightly larger. The lowest errors are obtained for the monolithic solver with implicit formulation of the convective term. The dual splitting scheme and the monolithic solver with explicit convective term yield similar errors, and the errors are again slightly larger for both explicit and implicit pressure-correction formulations. For very small time step sizes and the BDF3 scheme, the spatial error becomes dominant at some point. Note that BDF3 schemes are not considered for the pressure-correction scheme since~$J_p=1$ (required for stability) limits convergence rates to second-order in that case. For the dual splitting scheme, large errors occur for Courant numbers~$\mathrm{Cr}>1$. This effect does not show up for the non-moving mesh and we conjecture that this effect originates from the CFL condition. For a regular Navier--Stokes problem with explicit formulation of the convective term, no stability can be expected in this range of Courant numbers. While a sharp CFL bound is not visible for the chosen parameters, all solvers become unstable for~$\mathrm{Cr}> 1$ (for~$J=2$) when using a smaller viscosity of~$\nu=10^{-3}$ (higher Reynolds number), showing the usual sharp CFL bound.

In a second set of experiments, we study the spatial convergence behavior for polynomial degrees~$k=2,3,4,5$ by a mesh refinement study, considering refine levels~$l=1,...,6$. The BDF2 time integration scheme with a small, constant time step of~$\Delta t = 5 \cdot 10^{-5}$ is used to obtain small temporal errors. Figure~\ref{fig:vortex_spatial_convergence} shows results for the three different solver types with an explicit formulation of the convective term. Since this is a spatial convergence test, the results would be indistinguishable when using an implicit formulation of the convective term and are therefore not shown explicitly. For comparison, we also show results for a non-moving, Cartesian mesh for the monolithic solution approach. Overall, all variants converge will optimal rates of convergence for all polynomial degrees until the temporal discretization error is reached. Compared to the static mesh, the errors are slightly larger on the moving mesh, and the gap between moving and static meshes increases for increasing polynomial degree.

\begin{figure}[t]
 \centering
	\includegraphics[width=0.9\textwidth]{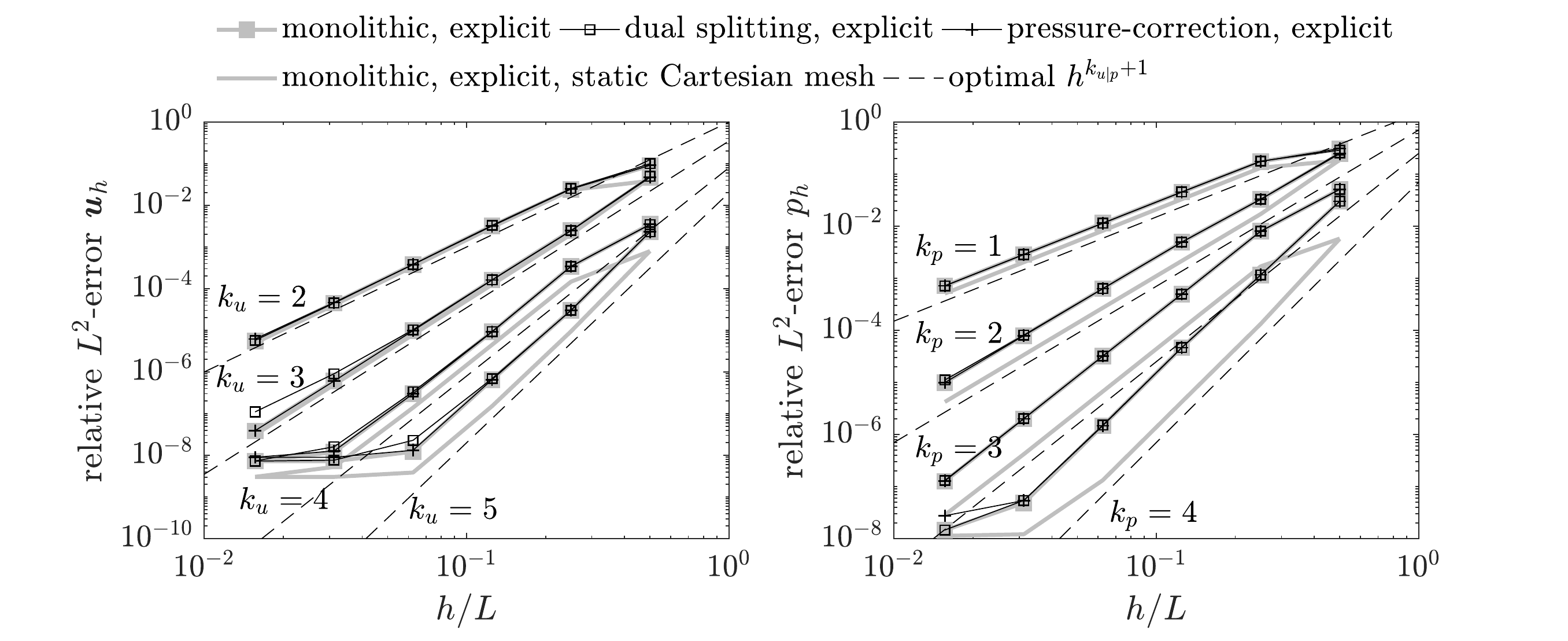}
\caption{Vortex problem: spatial convergence tests for ALE incompressible Navier--Stokes solvers for polynomial degrees~$k=2,3,4,5$ and comparison to static Cartesian mesh.}
\label{fig:vortex_spatial_convergence}
\end{figure}

\begin{figure}[t]
 \centering
	\includegraphics[width=0.9\textwidth]{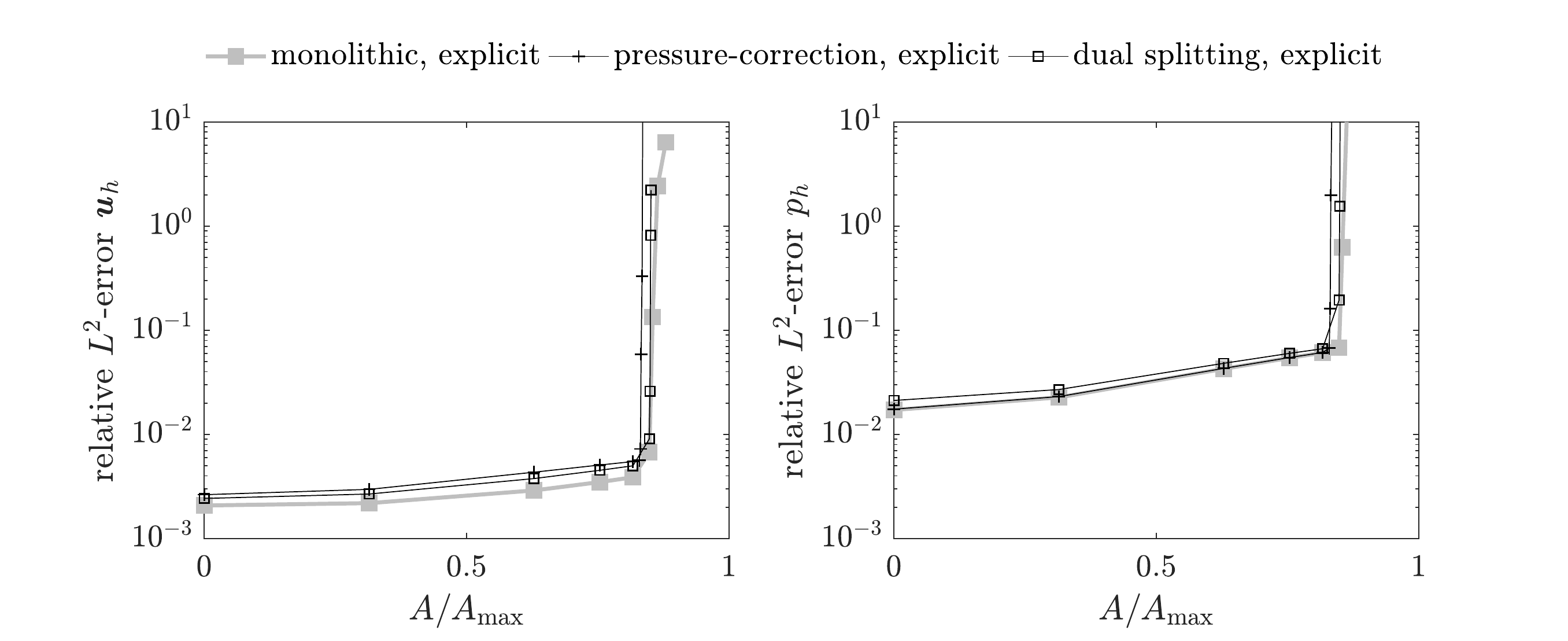}
\caption{Vortex problem: robustness test of ALE incompressible Navier--Stokes solvers on a mesh with~$4^2$ elements and polynomial degree~$k=3$.}
\label{fig:vortex_robustness_test_ale}
\end{figure}

Finally, we test the robustness by increasing the amplitude~$A$ of the mesh movement function to its theoretical limit, i.e., the value at which the aspect ratio tends to infinity or at which inverted elements occur, in order to characterize the point where the proposed ALE-DG methods will break down. This theoretical limit is reached when the lower left corner of the domain becomes an arbitrarily thin needle, see Figure~\ref{fig:mesh_max_deformed}. Mathematically, this limit is reached when the slope of the lower domain boundary reaches a value of~$1$ in the lower left corner at the time of maximum deformation, i.e.,
\begin{align*}
\left.\frac{\partial x_2(\chi_1, \chi_2 = - L/2, t=T_{\mathrm{G}}/4)}{\partial \chi_1}\right\vert_{\chi_1=-L/2} = A \left.\cos\left( 2\pi \frac{\chi_1 + L/2}{L}\right)\right\vert_{\chi_1=-L/2} \frac{2\pi }{L}  \overset{!}{=} 1 \leadsto \frac{A_{\mathrm{max}}}{L} = \frac{1}{2\pi} \; .
\end{align*}
In Figure~\ref{fig:vortex_robustness_test_ale}, we plot the relative errors of velocity and pressure for a coarse mesh with~$4^2$ elements and polynomial degree~$k=3$ as a function of~$A/A_{\mathrm{max}}$. Adaptive time-stepping, equation~\eqref{eq:local_CFL_Condition}, with~$\mathrm{Cr}=0.2$ is used. The error increases moderately for small amplitudes of the mesh deformation, and a rapid increase in errors can be observed around~$85 \%$ of the theoretically maximum amplitude for this problem.

\subsection{Robustness for under-resolved turbulent flows}\label{sec:ResultsUnderresolvedTurbulentFlows}
We study the applicability of the present ALE-DG solvers to transitional and turbulent flows by the example of the three-dimensinal Taylor--Green vortex (TGV) problem~\cite{Taylor1937}. Starting from the initial state
\begin{align*}
u_1(\bm{x},t=0) &= +\sin\left(x_1 \right)\cos\left(x_2\right)\cos\left(x_3\right)\; ,\\
u_2(\bm{x},t=0) &= -\cos\left(x_1 \right)\sin\left(x_2\right)\cos\left(x_3\right)\; ,\\
u_3(\bm{x},t=0) &= 0\; ,\\
p(\bm{x},t=0) &= \frac{1}{16}\left(\cos\left(2x_1\right)+ \cos\left(2x_2\right)\right)\left(\cos\left(2x_3\right)+2 \right)\; ,
\end{align*}
the flow transitions to turbulence in the absence of body forces,~$\bm{f}=\bm{0}$. The Reynolds number is~$\mathrm{Re}=1/\nu$, and we consider both the standard setting~$\mathrm{Re}=1600$ and the inviscid limit~$\mathrm{Re}\rightarrow\infty$ in the present work. The simulated time interval is~$0 \leq t \leq T=20$. At start time, the domain~$\Omega_0=\left[-L/2, L/2\right]^3=\left[-\pi, \pi\right]^3$ is a Cartesian box that deforms over time according to the mesh motion described in equation~\eqref{eq:mesh_motion_3d}, with an amplitude of~$A = \pi/6$ and varying mesh velocities characterized by period times decreassing from~$T_{\mathrm{G}} = 20$ to~$T_{\mathrm{G}} = 1$. Periodic boundaries are used in all coordinate directions. The domain boundaries are moving for the given mesh motion, but the mesh deformation is defined periodically in order to ensure consistency with the periodic boundary conditions. An illustration of the mesh deformation for the above parameters is given in Figure~\ref{fig:tgv_visualization}. We consider meshes that are originally Cartesian with~$N_{\mathrm{el}} = (2^l)^3$ elements, where~$l$ denotes the level of refinement, and denote the number of velocity nodes~$(2^l (k+1))^3$ as effective mesh resolution. A BDF2 time integration scheme along with an explicit treatment of the convective term is used for all solver types, with~$J_p=2$ for the dual splitting scheme and~$J_p=1$ for the pressure-correction scheme unless specified otherwise. Moreover, we use adaptive time-stepping, equation~\eqref{eq:local_CFL_Condition}, with a Courant number of~$\mathrm{Cr}=0.2$.

\begin{figure}[!ht]
 \centering
	\subfigure[$\Vert \bm{u}(\bm{x},t=0) \Vert$]{\includegraphics[width=0.19\textwidth]{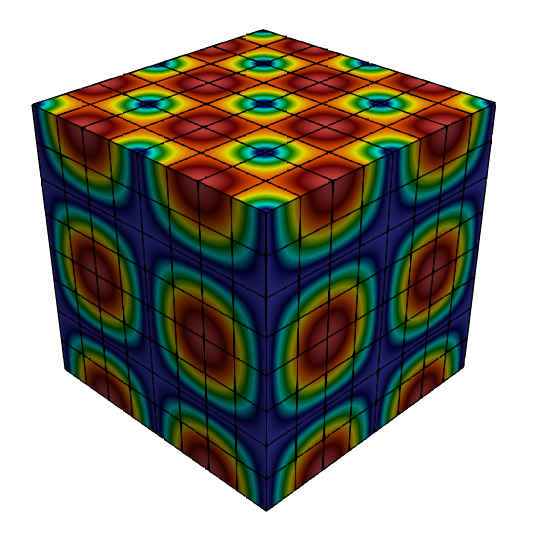}}
	\subfigure[$\Vert \bm{u}(\bm{x},t=5) \Vert$]{\includegraphics[width=0.19\textwidth]{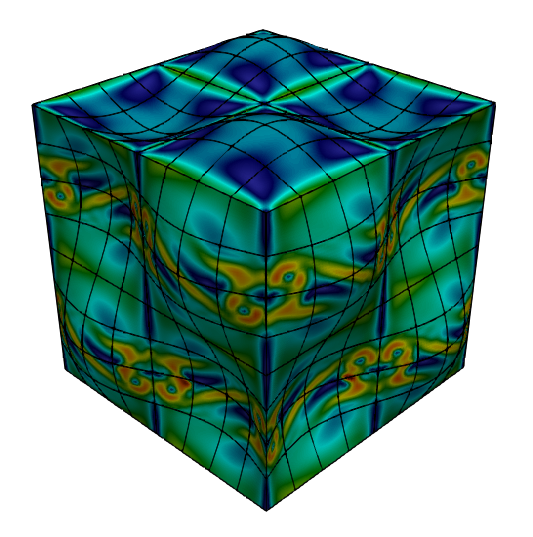}}
	\subfigure[$\Vert \bm{u}(\bm{x},t=10) \Vert$]{\includegraphics[width=0.19\textwidth]{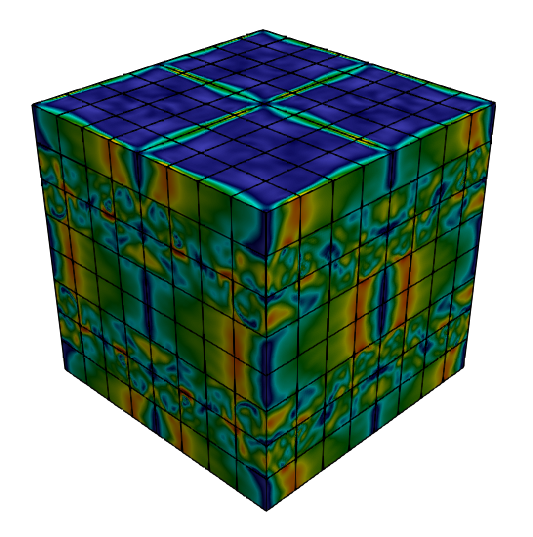}}
	\subfigure[$\Vert \bm{u}(\bm{x},t=15) \Vert$]{\includegraphics[width=0.19\textwidth]{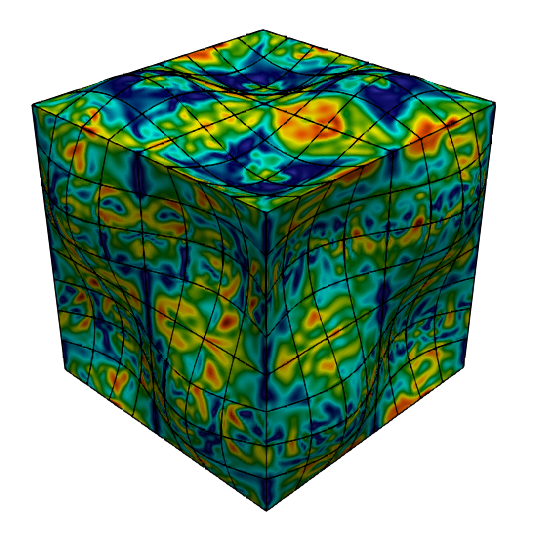}}
	\subfigure[$\Vert \bm{u}(\bm{x},t=20) \Vert$]{\includegraphics[width=0.19\textwidth]{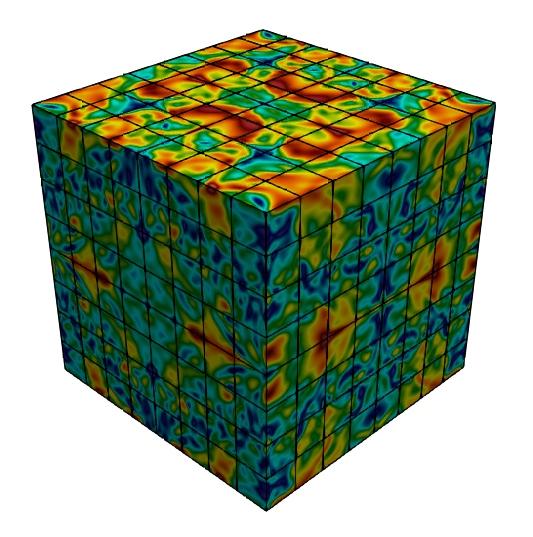}}
\caption{Taylor--Green vortex problem at~$\mathrm{Re}=1600$: visualization of velocity magnitude at different times for a spatial resolution with~$l=3$ and polynomial degree~$k=7$. The parameters of the mesh deformation are~$A = \pi/6$ and~$T_{\mathrm{G}} = 20$. The color map has been rescaled for each time instant, where red indicates high velocity and blue low velocity. The results shown have been simulated with the dual splitting scheme.}
\label{fig:tgv_visualization}
\end{figure}

\begin{figure}[t]
 \centering
 \subfigure[Kinetic energy.]{
	\includegraphics[width=0.9\textwidth]{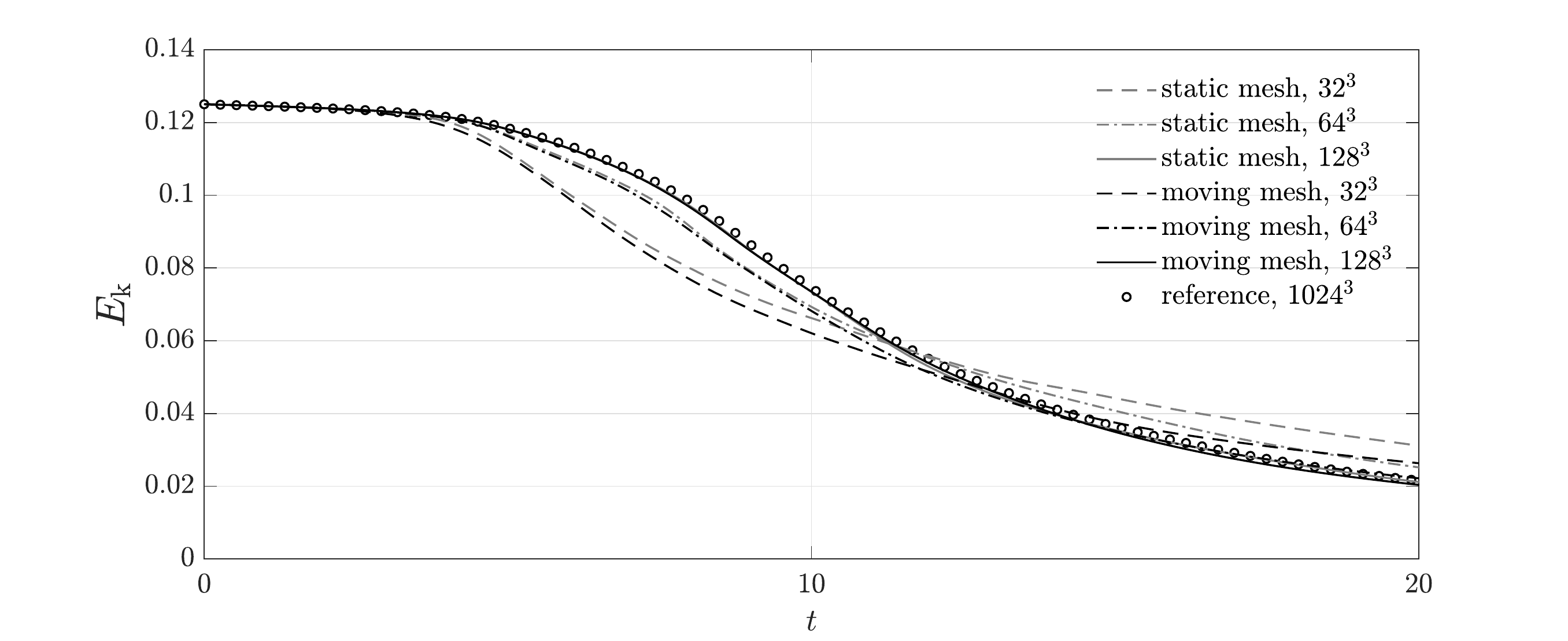}}
 \subfigure[Kinetic energy dissipation rate.]{
	\includegraphics[width=0.9\textwidth]{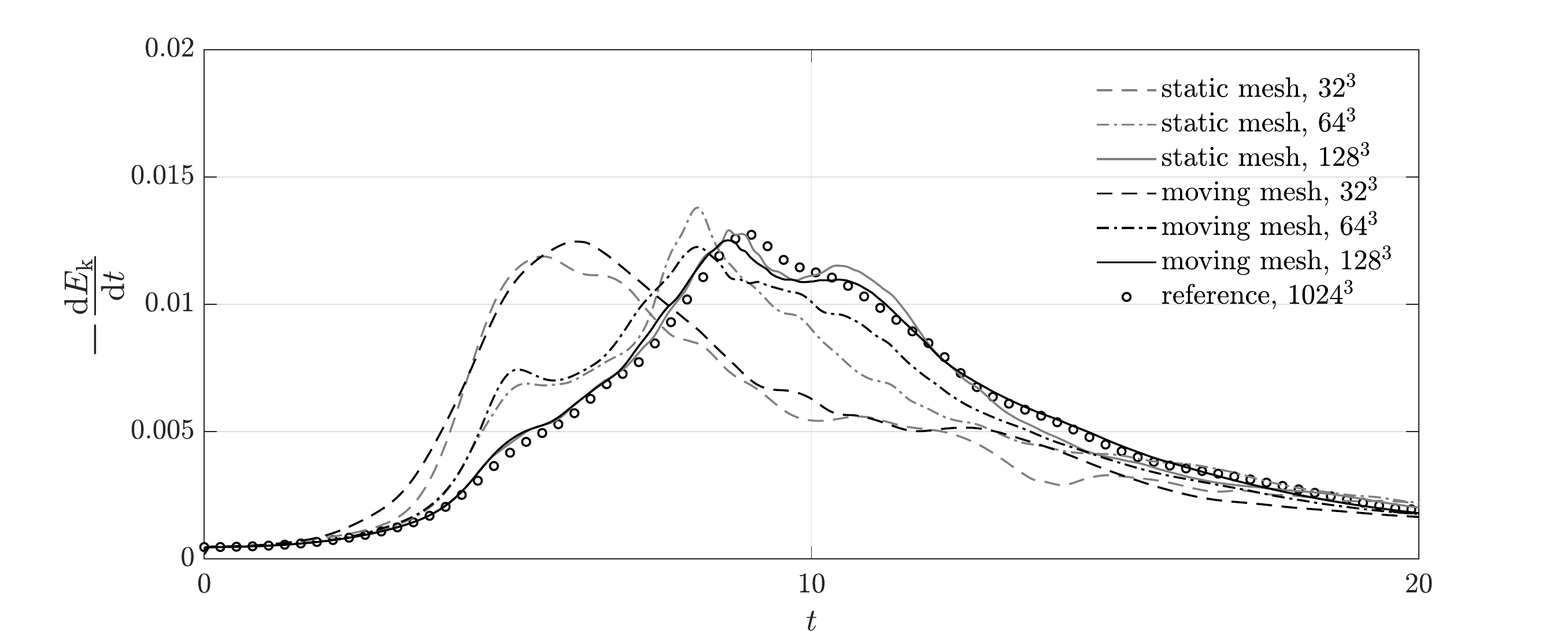}}
\caption{Taylor--Green vortex problem at~$\mathrm{Re}=1600$: comparison between static and moving meshes for polynomial degree~$k=3$ for increasing mesh resolutions of~$32^3$,~$64^3$, and~$128^3$. The results shown have been simulated with the dual splitting scheme.}
\label{fig:tgv_Re1600_convergence}
\end{figure}

\begin{figure}[t]
 \centering
\subfigure[Kinetic energy.]{
	\includegraphics[width=0.9\textwidth]{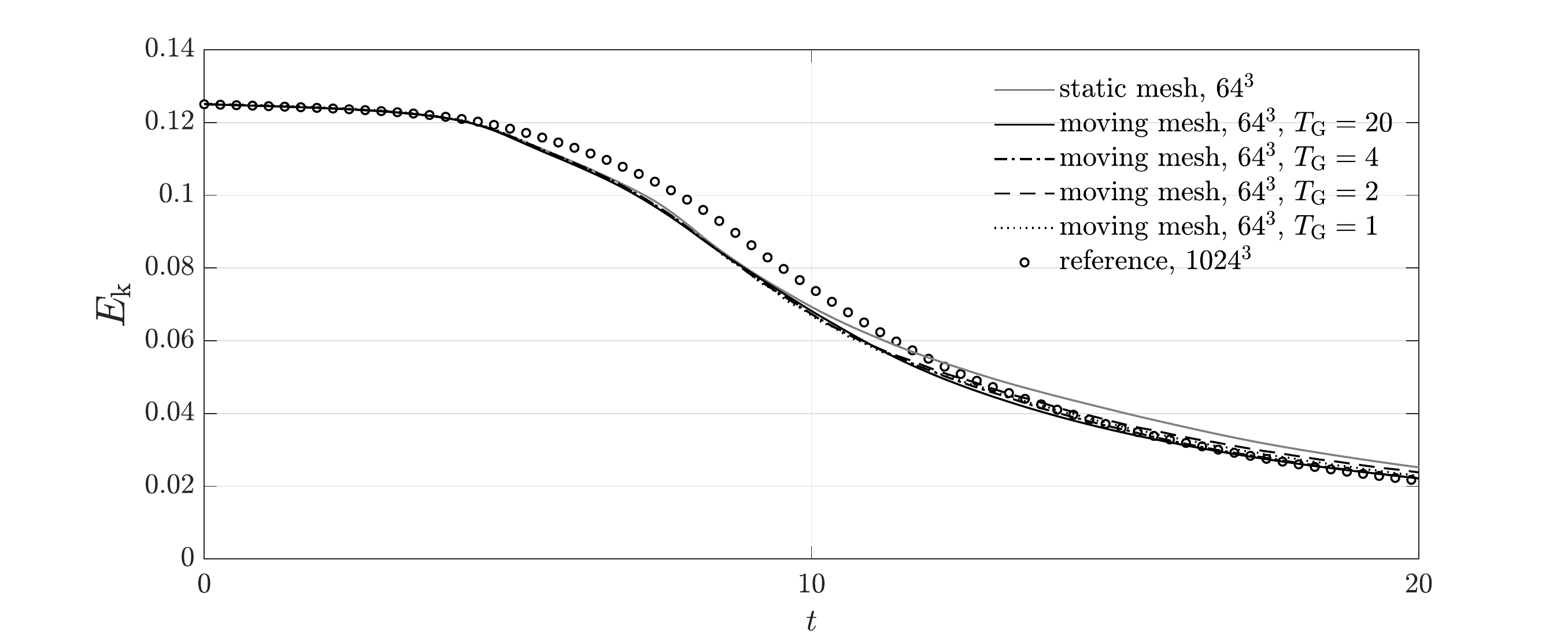}}
 \subfigure[Kinetic energy dissipation rate.]{
	\includegraphics[width=0.9\textwidth]{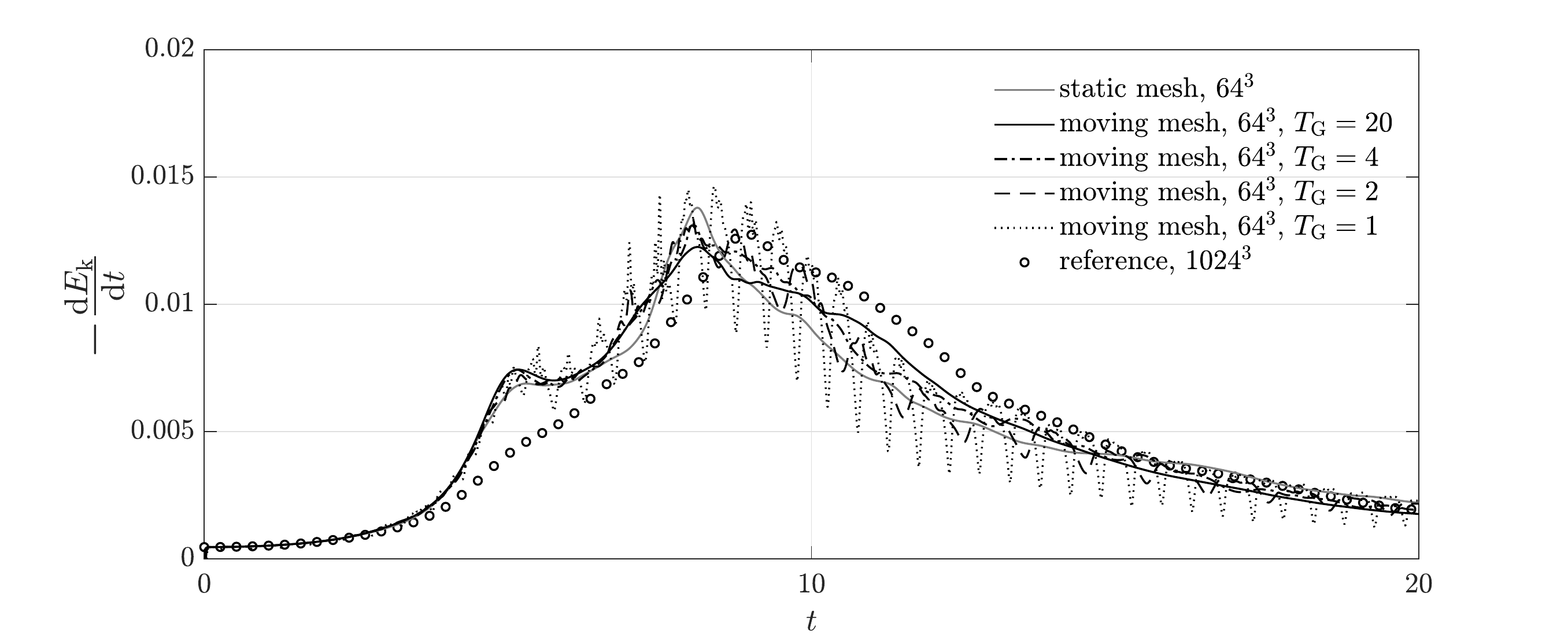}}
\caption{Taylor--Green vortex problem at~$\mathrm{Re}=1600$: comparison between static and moving meshes for polynomial degree~$k=3$ for increasing mesh velocity (decreasing period times of~$T_{\mathrm{G}}=20, 4, 2, 1$). A mesh with polynomial degree~$k=3$ and effective resolution of~$64^3$ is considered. The results shown have been simulated with the dual splitting scheme.}
\label{fig:tgv_Re1600_mesh_velocity}
\end{figure}

In a first set of experiments, the viscous case at~$\mathrm{Re}=1600$ is investigated. In Figure~\ref{fig:tgv_Re1600_convergence}, results for the kinetic energy~$E_k = \int_{\Omega_h} \frac{1}{2} \bm{u}_h \cdot\bm{u}_h \mathrm{d}\Omega / \int_{\Omega_h} \mathrm{d}\Omega$ and the dissipation rate of the kinetic energy obtained on the moving mesh are compared to results on a static Cartesian mesh. A mesh refinement study for degree~$k=3$ is performed using a rather slow mesh motion with a period time of~$T_{\mathrm{G}} = T$. The results converge towards the accurate DNS reference solution under mesh refinement, and the solution quality is comparable for static and moving meshes. In Figure~\ref{fig:tgv_Re1600_mesh_velocity}, we test the robustness w.r.t.~the grid velocity for the~$64^3$ mesh resolution by decreasing the period of the mesh motion down to~$T_{\mathrm{G}} = T/20 = 1$, resulting in a very fast mesh motion. With increasing mesh velocity, an oscillating behavior can be observed in the kinetic energy dissipation rate where the frequency of these oscillations follows the mesh motion. However, the temporal evolution of the kinetic energy is almost indistinguishable for the different mesh velocities. It can be observed that the oscillations in the dissipation rate are very small in the beginning of the simulation where the solution is smooth, while the oscillations grow once the transition to a turbulent state took place. In Section~\ref{sec:EnergyStability}, it was noted that the ALE transport term does not contribute to the energy evolution apart from the upwind stabilization term. However, this only holds under the assumption of exact numerical integration, which is not fulfilled on generally deformed geometries. A possible explanation for the results in Figure~\ref{fig:tgv_Re1600_mesh_velocity} could therefore be that integration errors (which are larger for non-smooth solutions or under-resolved scenarios) are amplified if the period of the mesh motion~$T_{\mathrm{G}}$ tends to zero and the mesh velocity tends to infinity. This is supported by the observation that the oscillications are larger on coarser meshes. Let us mention that this experiment is performed here to test the robustness of the solver and that such a scenario (increasing the mesh velocity for a fixed fluid velocity) is not representative of a fluid--structure interaction problem for which the fluid has to follow the motion of the fluid--structure interface due to no-slip conditions.

In a second set of experiments, the inviscid Taylor--Green vortex problem is studied, which is considered one of the most challenging benchmark examples to test the robustness of a flow solver for turbulent flows due to the absence of viscous dissipation, see for example~\cite{Schnucke2018arxiv}. Although the test case is academic, it can be expected that if a numerical method is robust for the inviscid Taylor--Green problem, it can also be successfully applied to practical, engineering problems. Table~\ref{tab:Inviscid_TGV_Robustness} shows results of robustness tests for the inviscid limit, where polynomial degrees of~$k=3,5,7,11,15$ are considered for  several spatial refinement levels~$l$. The highest mesh velocity studied above (corresponding to~$T_{\mathrm{G}} = 1$) is used for the inviscid TGV simulations. Robustness is particularly critical for coarse spatial resolutions where the flow is severly under-resolved, see for example~\cite{Fehn18a} for similar considerations on static meshes and a comparison of the present stabilized DG approach to a non-stabilized one, which is why we start our investigations with a mesh consisting of only one element,~$l=0$. With the standard penalty factor of~$\zeta = 1$ for the divergence and continuity penalty terms, robustness is achieved for all spatial resolutions and for all solver types (monolithic, dual splitting, pressure-correction). There is one exception: For the pressure-correction scheme and the parameters described above, using the incremental formulation lead to instabilities for the coarsest possible spatial resolution,~$l=0$ and~$k=3$, while the simulation was again stable for the non-incremental formulation~($J_p=0$). Overall, these results are encouraging in the sense that the stabilized DG approach proposed in~\cite{Fehn18a} for static meshes is well designed, meaning that robustness carries over to moving meshes without having to adjust discretization parameters. At the same time, we emphasize again that it is unclear to which extent energy stability can be guaranteed theoretically for the present stabilized DG approach.

\begin{table}[!ht]
\caption{Inviscid 3D Taylor--Green vortex problem: robustness of proposed ALE-DG methods is tested for refinement levels~$l=0,...,5$ and polynomial degrees~$k=3,5,7,11,15$. All simulations completed successfully as indicated by the symbol~\Checkmark. The sign '$-$' indicates that the specific spatial resolution is not considered since the effective resolution is limited to coarse discretizations for this robustness test.}
\label{tab:Inviscid_TGV_Robustness}
\renewcommand{\arraystretch}{1.1}
\begin{center}
\begin{tabular}{cccccc}
\toprule
& \multicolumn{5}{c}{Polynomial degree $k$}\\
\cline{2-6}
$l$ & $k=3$ & $k=5$ & $k=7$ & $k=11$ & $k=15$\\
\midrule
0        & \Checkmark   & \Checkmark & \Checkmark & \Checkmark & \Checkmark \\
1        & \Checkmark   & \Checkmark & \Checkmark & \Checkmark & \Checkmark \\
2        & \Checkmark   & \Checkmark & \Checkmark & \Checkmark & \Checkmark \\
3        & \Checkmark   & \Checkmark & \Checkmark & \Checkmark & $-$        \\
4        & \Checkmark   & \Checkmark & \Checkmark & $-$        & $-$        \\
5        & \Checkmark   & $-$        & $-$        & $-$        & $-$        \\
\bottomrule
\end{tabular}
\end{center}
\renewcommand{\arraystretch}{1}
\end{table}

\section{Conclusion and outlook}\label{sec:Conclusion}
We presented ALE-DG methods for the incompressible Navier--Stokes equations that are up to third-order accurate in time and arbitrarily high-order accurate in space for sufficiently smooth problems. Moving mesh formulations are derived for both monolithic and splitting approaches based on a method-of-lines approach, considering both implicit and explicit formulations of the convective term. The time integration framework relies on BDF and extrapolation schemes and extends naturally to adaptive time-stepping. Stable and high-order accurate boundary conditions are derived for the splitting-type approaches. The ALE methods are designed to automatically fulfill the geometric conservation law. An important aspect is that the proposed methods are simple to implement since the equations are solved on the deformed geometry, i.e., the generic finite element software takes care of the mapping and the geometry terms, and only one instance of the mesh is stored at a time. Fast matrix-free evaluation techniques are applied for all parts of the Navier--Stokes solvers as in the Eulerian case. A key feature is the use of consistent divergence and continuity penalty terms to stabilize the method for under-resolved turbulent flows. Numerical results demonstrate optimality in terms of convergence rates and the geometric conservation law, and robustness and accuracy of the proposed methods has been investigated for under-resolved turbulent flows. In the future, application to fluid--structure interaction problems is planned.

\appendix

\section{Adaptive time-stepping}\label{sec:AdaptiveTimeStepping}

For variable time step sizes, the time integration constants can be derived from Lagrange interpolation polynomials. To obtain the BDF coefficients~$\gamma_0^n$,~$\alpha_i^n$, the derivative of the Lagrange interpolation polynomials with support points at~$t_{n+1}, t_n,...,t_{n-J+1}$ is evaluated at time~$t_{n+1}$. To obtain the extrapolation coefficients~$\beta_i^n$, the Lagrange interpolation polynomials with support points at~$ t_n,...,t_{n-J+1}$ are evaluated at time~$t_{n+1}$. Table~\ref{tab:CoefficientsBDFTimeIntegrationAdaptive} summarizes the time integration constants for variable time step sizes using the notation introduced in Section~\ref{sec:TemporalDiscretization}.

\begin{table}[!ht]
\caption{Coefficients of BDF time integration scheme and extrapolation scheme for adaptive time step sizes.}\label{tab:CoefficientsBDFTimeIntegrationAdaptive}
\begin{center}
\begin{tabular}{cccc}
\toprule
  & BDF1 & BDF2 & BDF3\\ 
\midrule
$\gamma_0^n$ & 1 
             & $\frac{2 \Delta t_n + \Delta t_{n-1}}{\Delta t_n + \Delta t_{n-1}}$ 
	         & $1 + \frac{\Delta t_n}{\Delta t_n + \Delta t_{n-1}} + \frac{\Delta t_n}{\Delta t_n + \Delta t_{n-1} + \Delta t_{n-2}}$ \\
$\alpha_0^n$ & 1 
             & $\frac{\Delta t_n + \Delta t_{n-1}}{\Delta t_{n-1}}$ 
             & $\frac{\left(\Delta t_n + \Delta t_{n-1} \right)\left(\Delta t_n + \Delta t_{n-1} + \Delta t_{n-2} \right)}{\Delta t_{n-1} \left(\Delta t_{n-1} + \Delta t_{n-2}\right)} $\\
$\alpha_1^n$ & - 
             & $-\frac{\Delta t^2_n}{\left(\Delta t_n + \Delta t_{n-1}\right)\Delta t_{n-1}}$ 
             & $- \frac{\Delta t^2_n \left(\Delta t_n + \Delta t_{n-1} + \Delta t_{n-2} \right)}{\left(\Delta t_n + \Delta t_{n-1}\right) \Delta t_{n-1} \Delta t_{n-2}}$\\
$\alpha_2^n$ & - 
             & - 
             & $\frac{\Delta t^2_n \left(\Delta t_n + \Delta t_{n-1}\right)}{\left(\Delta t_n + \Delta t_{n-1} + \Delta t_{n-2}\right)\left(\Delta t_{n-1} + \Delta t_{n-2}\right) \Delta t_{n-2}} $\\
\hline
$\beta_0^n$  & 1 
             & $\frac{\Delta t_n + \Delta t_{n-1}}{\Delta t_{n-1}}$ 
             & $\frac{\left(\Delta t_n + \Delta t_{n-1} \right)\left(\Delta t_n + \Delta t_{n-1} + \Delta t_{n-2} \right)}{\Delta t_{n-1} \left(\Delta t_{n-1} + \Delta t_{n-2} \right)} $\\
$\beta_1^n$  & - 
             & $-\frac{\Delta t_n }{\Delta t_{n-1}}$ 
             & $-\frac{\Delta t_n \left(\Delta t_n + \Delta t_{n-1} + \Delta t_{n-2} \right)}{\Delta t_{n-1} \Delta t_{n-2}}$ \\
$\beta_2^n$  & - 
             & - 
             & $\frac{\Delta t_n \left(\Delta t_n + \Delta t_{n-1} \right)}{\left(\Delta t_{n-1} + \Delta t_{n-2}\right) \Delta t_{n-2}} $\\
\bottomrule
\end{tabular} 
\end{center}
\end{table}

\section{Energy conservation property on moving meshes}\label{sec:EnergyConservation}
In this section, we derive an energy conservation property according to equation~\eqref{eq:EnergyConservation} for the ALE form of the incompressible Navier--Stokes equations under the assumption of~$\nu=0$,~$\bm{f}=\bm{0}$, and periodic boundaries. We begin with
\begin{align}
\begin{split}
\int_{\Omega(t)} \left. \frac{\partial \frac{1}{2}\bm{u} \cdot \bm{u}}{\partial t}\right\vert_{\boldsymbol{\chi}} \mathrm{d}\Omega &= \int_{\Omega(t)} \bm{u} \cdot  \left. \frac{\partial \bm{u}}{\partial t}\right\vert_{\boldsymbol{\chi}} \mathrm{d}\Omega\\
& = -\int_{\Omega(t)} \bm{u} \cdot  \left( ((\bm{u}-\bm{u}_{\mathrm{G}})\cdot \nabla )\bm{u} + \nabla p \right) \mathrm{d}\Omega \, ,
\end{split}\label{eq:EnergyConservation_Begin}
\end{align}
where the momentum equation in ALE form, equation~\eqref{eq:MomentumEquationALE}, has been inserted in the second step. Next, the convective term is reformulated by making use of the identity~$\bm{u} \cdot \Grad{\bm{u}} \cdot \bm{w} = -\frac{1}{2}\left(\bm{u}\cdot \bm{u}\right) \Div{\bm{w}} + \frac{1}{2}\Div{\left(\bm{w} \left(\bm{u} \cdot \bm{u}\right)\right)}$ with~$\bm{w} = \bm{u} - \bm{u}_{\mathrm{G}}$, so that we obtain after applying Gauss' divergence theorem for the second term
\begin{align}
\begin{split}
 \int_{\Omega(t)} \bm{u} \cdot \nabla \bm{u} \cdot \left(\bm{u} - \bm{u}_{\mathrm{G}}\right)\; \mathrm{d}\Omega = 
 &- \frac{1}{2} \int_{\Omega(t)} \left( \bm{u} \cdot \bm{u} \right) \Div{\left(\bm{u} - \bm{u}_{\mathrm{G}}\right)} \; \mathrm{d}\Omega
 +\frac{1}{2} \int_{\partial \Omega(t)} \left( \bm{u} \cdot \bm{u} \right) \left(\bm{u} - \bm{u}_{\mathrm{G}}\right)\cdot\bm{n} \; \mathrm{d}\Gamma \\
 = & + \frac{1}{2} \int_{\Omega(t)} \left( \bm{u} \cdot \bm{u} \right) \Div{\bm{u}_{\mathrm{G}}} \; \mathrm{d}\Omega\, .
\end{split}\label{eq:EnergyConservation_ConvectiveTerm}
\end{align}
In the second step, we exploited that the surface integral vanishes due to the assumption of periodicity, and that~$\Div{\bm{u}}=0$ holds (continuity equation). 
For the pressure gradient term, integration by parts is performed to obtain
\begin{align}
 \int_{\Omega(t)} \bm{u} \Grad{p} \; \mathrm{d}\Omega = -  \int_{\Omega(t)} \underbrace{\Div{\bm{u}}}_{=0} \; p \; \mathrm{d}\Omega  + \int_{\partial \Omega(t)} p \; \bm{u}\cdot\bm{n} \; \mathrm{d}\Gamma = 0 \, ,\label{eq:EnergyConservation_PressureTerm}
\end{align}
where the volume integral vanished due to~$\Div{\bm{u}}=0$ and the surface integral due to periodic boundaries. Inserting equations~\eqref{eq:EnergyConservation_ConvectiveTerm} and~\eqref{eq:EnergyConservation_PressureTerm} into equation~\eqref{eq:EnergyConservation_Begin}, we arrive at the result
\begin{align*}
\int_{\Omega(t)} \left. \frac{\partial \frac{1}{2}\bm{u} \cdot \bm{u}}{\partial t}\right\vert_{\boldsymbol{\chi}} \mathrm{d}\Omega + \int_{\Omega(t)} \frac{1}{2}\left( \bm{u} \cdot \bm{u} \right) \Div{\bm{u}_{\mathrm{G}}} \; \mathrm{d}\Omega = 0 \, .
\end{align*}
To explain why this equation describes conservation of kinetic energy, we transform the integral onto the reference domain~$\Omega_0$, using the Jacobian~$\bm{J} =\partial \bm{x}/\partial \boldsymbol{\chi}$
\begin{align*}
\int_{\Omega_0} \left. \frac{\partial \frac{1}{2}\bm{u} \cdot \bm{u}}{\partial t}\right\vert_{\boldsymbol{\chi}} \det\bm{J} \mathrm{d}\Omega + \int_{\Omega_0} \frac{1}{2}\left( \bm{u} \cdot \bm{u} \right) \Div{\bm{u}_{\mathrm{G}}} \det\bm{J} \; \mathrm{d}\Omega = 0 \, .
\end{align*}
Using the relation
\begin{align*}
\left.\frac{\partial \det \bm{J}}{\partial t}\right\vert_{\boldsymbol{\chi}} = \Div{\bm{u}_{\mathrm{G}}} \det\bm{J} \; ,
\end{align*}
which is frequently used in the context of ALE derivations, see for example~\cite{Foerster2006}, we can further simplify the integral
\begin{align*}
\begin{split}
0 &= \int_{\Omega_0} \left. \frac{\partial \frac{1}{2}\bm{u} \cdot \bm{u}}{\partial t}\right\vert_{\boldsymbol{\chi}} \det\bm{J} + \frac{1}{2}\left( \bm{u} \cdot \bm{u} \right) \left.\frac{\partial \det \bm{J}}{\partial t}\right\vert_{\boldsymbol{\chi}}\; \mathrm{d}\Omega = \int_{\Omega_0} \left. \frac{\partial \frac{1}{2}\bm{u} \cdot \bm{u} \det\bm{J}}{\partial t}\right\vert_{\boldsymbol{\chi}} \; \mathrm{d}\Omega =\\ 
& = \left.\frac{\partial }{\partial t}\int_{\Omega_0} \frac{1}{2}\bm{u} \cdot \bm{u} \det\bm{J} \; \mathrm{d}\Omega \right\vert_{\boldsymbol{\chi}} \, .
\end{split}
\end{align*}
In the last step, the time derivative has been pulled out of the integral since the integral is performed over~$\Omega_0$, which does not depend on~$t$. From the above relation, it is obvious that energy is conserved under the given assumptions.

\section{Energy balance of convective term}\label{sec:EnergyBalanceConvectiveTerm}
In this section, we show that equation~\eqref{eq:ConvectiveTermReformulation} holds. 

\begin{proof}
The first goal is to achieve that the volume integral of the convective term contains the divergence of the velocity and the grid velocity. Similar to the continuous case in~\ref{sec:EnergyConservation}, we make use of the identity
\begin{align*}
\bm{u}_h \cdot \Grad{\bm{u}_h} \cdot \bm{w}_h = -\frac{1}{2}\Div{\bm{w}_h} \left(\bm{u}_h\cdot \bm{u}_h\right) + \frac{1}{2}\Div{\left(\bm{w}_h\left(  \bm{u}_h \cdot \bm{u}_h\right)\right)} \; ,
\end{align*} 
with~$\bm{w}_h = \bm{u}_h - \bm{u}_{\mathrm{G},h}$. Note that we cannot exploit that the divergence of the discrete velocity~$\bm{u}_h$ is zero in the discrete case. Inserting the above identity into the left-hand side of equation~\eqref{eq:ConvectiveTermReformulation} and applying Gauss' divergence theorem yields
\begin{align}
\begin{split}
\sum_{e=1}^{N_{\mathrm{el}}} &\left(
\intele{\bm{u}_h}{\left(\Grad{\bm{u}_h}\right) \cdot \left(\bm{u}_h - \bm{u}_{\mathrm{G},h}\right)} 
- \inteleface{\bm{u}_h}{\left(\left(\avg{\bm{u}_h}-\bm{u}_{\mathrm{G},h}\right)\cdot\bm{n}\right) \frac{1}{2}\jumporiented{\bm{u}_h}}\right)\\
=
& - \sum_{e=1}^{N_{\mathrm{el}}} \left(\frac{1}{2} \intele{\Div{\left(\bm{u}_h - \bm{u}_{\mathrm{G},h}\right)}}{\bm{u}_h\cdot\bm{u}_h}\right)\\
& + \sum_{e=1}^{N_{\mathrm{el}}} \left(
    \frac{1}{2}\inteleface{\bm{u}_h\cdot\bm{u}_h}{\left(\bm{u}_h-\bm{u}_{\mathrm{G},h}\right)\cdot \bm{n}}
  - \inteleface{\bm{u}_h}{\left(\left(\avg{\bm{u}_h}-\bm{u}_{\mathrm{G},h}\right)\cdot\bm{n}\right) \frac{1}{2}\jumporiented{\bm{u}_h}}\right)\\
= 
& - \frac{1}{2}\intdomain{\Div{\left(\bm{u}_h-\bm{u}_{\mathrm{G},h}\right)}}{\bm{u}_h\cdot\bm{u}_h}\\
& \underbrace{- \frac{1}{2}\intinteriorfaces{\jumporiented{\bm{u}_h\cdot\bm{u}_h}}{\bm{u}_{\mathrm{G},h}\cdot \bm{n}}
  + \intinteriorfaces{\jumporiented{\bm{u}_h}\cdot\avg{\bm{u}_h}}{\bm{u}_{\mathrm{G},h}\cdot \bm{n}}}_{=0 \; (\text{linear transport})}\\
& + \frac{1}{2}\intinteriorfaces{\bm{u}_h^-\cdot \bm{n}^-}{\bm{u}_h^-\cdot\bm{u}_h^-}
  + \frac{1}{2}\intinteriorfaces{\bm{u}_h^+\cdot \bm{n}^+}{\bm{u}_h^+\cdot\bm{u}_h^+}
  - \intinteriorfaces{\jumporiented{\bm{u}_h}\cdot\avg{\bm{u}_h}}{\avg{\bm{u}_h}\cdot \bm{n}}\; ,
\end{split}\label{eq:ReformulationConvectiveTerm_Step1}
\end{align}
where the summation over all elements has been performed in the second step. It can easily be verified that the surface integrals related to the grid velocity add up to zero, as expected for a linear transport term. The third row on the right-hand side of the above equation is more difficult to summarize, since none of the two slots is the same for both sides of the face (originating from the nonlinearity of the transport term), so we need to insert a relation like~$\bm{u}_h^+ = \bm{u}_h^- - \jumporiented{\bm{u}_h}$ in order to be able to combine these terms
\begin{align}
\begin{split}
& + \frac{1}{2}\intinteriorfaces{\bm{u}_h^-\cdot \bm{n}^-}{\bm{u}_h^-\cdot\bm{u}_h^-}
  + \frac{1}{2}\intinteriorfaces{\bm{u}_h^+\cdot \bm{n}^+}{\bm{u}_h^+\cdot\bm{u}_h^+}
  - \intinteriorfaces{\jumporiented{\bm{u}_h}\cdot\avg{\bm{u}_h}}{\avg{\bm{u}_h}\cdot \bm{n}}\\
=
& \underbrace{+ \frac{1}{2}\intinteriorfaces{\bm{u}_h^- \cdot \bm{u}_h^-}{\bm{u}_h^- \cdot \bm{n}^-}
- \frac{1}{2}\intinteriorfaces{\bm{u}_h^+\cdot \bm{u}_h^-}{\bm{u}_h^-\cdot \bm{n}^-}}_{\text{combine}}
+ \frac{1}{2}\intinteriorfaces{\bm{u}_h^+\cdot \jumporiented{\bm{u}_h}}{\bm{u}_h^-\cdot \bm{n}^-}\\
& \underbrace{+ \frac{1}{2}\intinteriorfaces{\bm{u}_h^+\cdot \bm{u}_h^-}{\jumporiented{\bm{u}_h}\cdot \bm{n}^-}
- \frac{1}{2}\intinteriorfaces{\bm{u}_h^+\cdot \jumporiented{\bm{u}_h}}{\jumporiented{\bm{u}_h}\cdot \bm{n}^-}}_{\text{combine}}
- \intinteriorfaces{\jumporiented{\bm{u}_h}\cdot\avg{\bm{u}_h}}{\avg{\bm{u}_h}\cdot \bm{n}}\\
=
& \underbrace{+ \frac{1}{2}\intinteriorfaces{\jumporiented{\bm{u}_h}\cdot \bm{u}_h^-}{\bm{u}_h^-\cdot \bm{n}^-}
+ \frac{1}{2}\intinteriorfaces{\bm{u}_h^+\cdot \jumporiented{\bm{u}_h}}{\bm{u}_h^-\cdot \bm{n}^-}}_{\text{combine}}
+ \frac{1}{2}\intinteriorfaces{\bm{u}_h^+\cdot\bm{u}_h^+}{\jumporiented{\bm{u}_h}\cdot \bm{n}^-}\\
&- \intinteriorfaces{\jumporiented{\bm{u}_h}\cdot\avg{\bm{u}_h}}{\avg{\bm{u}_h}\cdot \bm{n}}\\
=
& \underbrace{+ \intinteriorfaces{\jumporiented{\bm{u}_h}\cdot \frac{\bm{u}_h^- + \bm{u}_h^+}{2}}{\bm{u}_h^-\cdot \bm{n}^-}
- \intinteriorfaces{\jumporiented{\bm{u}_h}\cdot\avg{\bm{u}_h}}{\avg{\bm{u}_h}\cdot \bm{n}}}_{\text{combine}}
+ \frac{1}{2}\intinteriorfaces{\bm{u}_h^+\cdot \bm{u}_h^+}{\jumporiented{\bm{u}_h}\cdot \bm{n}^-}\\
=
& \underbrace{+\intinteriorfaces{\jumporiented{\bm{u}_h}\cdot \avg{\bm{u}_h}}{\left(\bm{u}_h^- - \avg{\bm{u}_h}\right)\cdot \bm{n}^-}
+ \frac{1}{2}\intinteriorfaces{\bm{u}_h^+ \cdot\bm{u}_h^+}{\jumporiented{\bm{u}_h}\cdot \bm{n}^-}}_{\text{combine}}\\
= &
+ \frac{1}{2}\intinteriorfaces{\jumporiented{\bm{u}_h}\cdot \avg{\bm{u}_h} + \bm{u}_h^+ \cdot\bm{u}_h^+}{\jumporiented{\bm{u}_h}\cdot \bm{n}^-} \\
= & + \frac{1}{2}\intinteriorfaces{\avg{\bm{u}_h\cdot \bm{u}_h}}{\jumporiented{\bm{u}_h}\cdot \bm{n}^-}\; ,
\end{split}
\label{eq:ReformulationConvectiveTerm_Step2}
\end{align}
where one can easily verify that~$\jumporiented{\bm{u}_h}\cdot \avg{\bm{u}_h} + \bm{u}_h^+ \cdot\bm{u}_h^+ = \avg{\bm{u}_h\cdot \bm{u}_h}$. Inserting equation~\eqref{eq:ReformulationConvectiveTerm_Step2} into equation~\eqref{eq:ReformulationConvectiveTerm_Step1} completes the proof.\qed
\end{proof}

\section*{Acknowledgments}
The research presented in this paper was partly funded by the German Research Foundation (DFG) under the project ``High-order discontinuous Galerkin for the exa-scale'' (ExaDG) within the priority program ``Software for Exascale Computing'' (SPPEXA), grant agreement no. KR4661/2-1 and WA1521/18-1.


\bibliography{paper}

\begin{thebibliography}{10}
\expandafter\ifx\csname url\endcsname\relax
  \def\url#1{\texttt{#1}}\fi
\expandafter\ifx\csname urlprefix\endcsname\relax\def\urlprefix{URL }\fi
\expandafter\ifx\csname href\endcsname\relax
  \def\href#1#2{#2} \def\path#1{#1}\fi

\bibitem{Hirt1974}
C.~Hirt, A.~Amsden, J.~Cook, An arbitrary {Lagrangian}-{Eulerian} computing
  method for all flow speeds, Journal of Computational Physics 14~(3) (1974)
  227 -- 253.
\newblock \href {http://dx.doi.org/10.1016/0021-9991(74)90051-5}
  {\path{doi:10.1016/0021-9991(74)90051-5}}.

\bibitem{Donea1982}
J.~Donea, S.~Giuliani, J.~Halleux, An arbitrary {L}agrangian--{E}ulerian finite
  element method for transient dynamic fluid-structure interactions, Computer
  Methods in Applied Mechanics and Engineering 33~(1) (1982) 689 -- 723.
\newblock \href {http://dx.doi.org/10.1016/0045-7825(82)90128-1}
  {\path{doi:10.1016/0045-7825(82)90128-1}}.

\bibitem{Donea2017}
J.~Donea, A.~Huerta, J.-P. Ponthot, A.~Rodríguez-Ferran, Arbitrary
  {L}agrangian--{E}ulerian {M}ethods, American Cancer Society, 2017, pp. 1--23.
\newblock \href {http://dx.doi.org/10.1002/9781119176817.ecm2009}
  {\path{doi:10.1002/9781119176817.ecm2009}}.

\bibitem{Donea1977}
J.~Donea, P.~Fasoli-Stella, S.~Giuliani, Lagrangian and {Eulerian} finite
  element techniques for transient fluid-structure interaction problems, Therm
  and Fluid/Struct Dyn Anal.

\bibitem{Hughes1981}
T.~J. Hughes, W.~K. Liu, T.~K. Zimmermann, {Lagrangian}-{Eulerian} finite
  element formulation for incompressible viscous flows, Computer Methods in
  Applied Mechanics and Engineering 29~(3) (1981) 329 -- 349.
\newblock \href {http://dx.doi.org/10.1016/0045-7825(81)90049-9}
  {\path{doi:10.1016/0045-7825(81)90049-9}}.

\bibitem{Beskok2001}
A.~Beskok, T.~C. Warburton, An unstructured hp finite-element scheme for fluid
  flow and heat transfer in moving domains, Journal of Computational Physics
  174~(2) (2001) 492 -- 509.
\newblock \href {http://dx.doi.org/10.1006/jcph.2001.6885}
  {\path{doi:10.1006/jcph.2001.6885}}.

\bibitem{Lesoinne1996}
M.~Lesoinne, C.~Farhat, Geometric conservation laws for flow problems with
  moving boundaries and deformable meshes, and their impact on aeroelastic
  computations, Computer Methods in Applied Mechanics and Engineering 134~(1)
  (1996) 71 -- 90.
\newblock \href {http://dx.doi.org/10.1016/0045-7825(96)01028-6}
  {\path{doi:10.1016/0045-7825(96)01028-6}}.

\bibitem{Lomtev1999}
I.~Lomtev, R.~Kirby, G.~Karniadakis, A discontinuous {Galerkin} {ALE} method
  for compressible viscous flows in moving domains, Journal of Computational
  Physics 155~(1) (1999) 128 -- 159.
\newblock \href {http://dx.doi.org/10.1006/jcph.1999.6331}
  {\path{doi:10.1006/jcph.1999.6331}}.

\bibitem{Nguyen2010}
V.-T. Nguyen, An arbitrary {L}agrangian--{E}ulerian discontinuous {G}alerkin
  method for simulations of flows over variable geometries, Journal of Fluids
  and Structures 26~(2) (2010) 312 -- 329.
\newblock \href {http://dx.doi.org/10.1016/j.jfluidstructs.2009.11.002}
  {\path{doi:10.1016/j.jfluidstructs.2009.11.002}}.

\bibitem{Mavriplis2011}
D.~J. Mavriplis, C.~R. Nastase, On the geometric conservation law for
  high-order discontinuous {G}alerkin discretizations on dynamically deforming
  meshes, Journal of Computational Physics 230~(11) (2011) 4285 -- 4300,
  special issue High Order Methods for CFD Problems.
\newblock \href {http://dx.doi.org/10.1016/j.jcp.2011.01.022}
  {\path{doi:10.1016/j.jcp.2011.01.022}}.

\bibitem{Persson2009}
P.-O. Persson, J.~Bonet, J.~Peraire, Discontinuous {Galerkin} solution of the
  {Navier}-{Stokes} equations on deformable domains, Computer Methods in
  Applied Mechanics and Engineering 198 (2009) 1585--1595.
\newblock \href {http://dx.doi.org/10.1016/j.cma.2009.01.012}
  {\path{doi:10.1016/j.cma.2009.01.012}}.

\bibitem{Schnucke2018arxiv}
G.~Schn{\"u}cke, N.~Krais, T.~Bolemann, G.~J. Gassner, Entropy stable
  discontinuous {G}alerkin schemes on moving meshes with summation-by-parts
  property for hyperbolic conservation laws, arXiv preprint arXiv:1812.09093.

\bibitem{Thomas1979}
P.~D. Thomas, C.~K. Lombard, Geometric conservation law and its application to
  flow computations on moving grids, AIAA Journal 17~(10) (1979) 1030--1037.
\newblock \href {http://dx.doi.org/10.2514/3.61273}
  {\path{doi:10.2514/3.61273}}.

\bibitem{Farhat2004}
C.~Farhat, P.~Geuzaine, Design and analysis of robust {ALE} time-integrators
  for the solution of unsteady flow problems on moving grids, Computer Methods
  in Applied Mechanics and Engineering 193~(39) (2004) 4073 -- 4095, the
  Arbitrary Lagrangian-Eulerian Formulation.
\newblock \href {http://dx.doi.org/10.1016/j.cma.2003.09.027}
  {\path{doi:10.1016/j.cma.2003.09.027}}.

\bibitem{Foerster2006}
C.~F\"orster, W.~A. Wall, E.~Ramm, On the geometric conservation law in
  transient flow calculations on deforming domains, International Journal for
  Numerical Methods in Fluids 50~(12) (2006) 1369--1379.
\newblock \href {http://dx.doi.org/10.1002/fld.1093}
  {\path{doi:10.1002/fld.1093}}.

\bibitem{Etienne2009}
S.~\'Etienne, A.~Garon, D.~Pelletier, Perspective on the geometric conservation
  law and finite element methods for {ALE} simulations of incompressible flow,
  Journal of Computational Physics 228~(7) (2009) 2313 -- 2333.
\newblock \href {http://dx.doi.org/10.1016/j.jcp.2008.11.032}
  {\path{doi:10.1016/j.jcp.2008.11.032}}.

\bibitem{Cockburn2005}
B.~Cockburn, G.~Kanschat, D.~Sch{\"o}tzau, A locally conservative {L}{D}{G}
  method for the incompressible {N}avier--{S}tokes equations, Mathematics of
  Computation 74~(251) (2005) 1067--1095.

\bibitem{Cockburn2007}
B.~Cockburn, G.~Kanschat, D.~Sch{\"o}tzau, A note on discontinuous {G}alerkin
  divergence-free solutions of the {N}avier--{S}tokes equations, Journal of
  Scientific Computing 31~(1) (2007) 61--73.
\newblock \href {http://dx.doi.org/10.1007/s10915-006-9107-7}
  {\path{doi:10.1007/s10915-006-9107-7}}.

\bibitem{Cockburn2009}
B.~Cockburn, G.~Kanschat, D.~Sch{\"o}tzau, An equal-order {D}{G} method for the
  incompressible {N}avier--{S}tokes equations, Journal of Scientific Computing
  40~(1-3) (2009) 188--210.

\bibitem{Bassi2006}
F.~Bassi, A.~Crivellini, D.~D. Pietro, S.~Rebay, An artificial compressibility
  flux for the discontinuous {G}alerkin solution of the incompressible
  {N}avier--{S}tokes equations, Journal of Computational Physics 218~(2) (2006)
  794 -- 815.
\newblock \href {http://dx.doi.org/10.1016/j.jcp.2006.03.006}
  {\path{doi:10.1016/j.jcp.2006.03.006}}.

\bibitem{Shahbazi07}
K.~Shahbazi, P.~F. Fischer, C.~R. Ethier, A high-order discontinuous {G}alerkin
  method for the unsteady incompressible {N}avier--{S}tokes equations, Journal
  of Computational Physics 222~(1) (2007) 391 -- 407.
\newblock \href {http://dx.doi.org/10.1016/j.jcp.2006.07.029}
  {\path{doi:10.1016/j.jcp.2006.07.029}}.

\bibitem{Hesthaven07}
J.~S. Hesthaven, T.~Warburton, Nodal discontinuous {G}alerkin methods:
  {A}lgorithms, analysis, and applications, Springer, 2007.
\newblock \href {http://dx.doi.org/10.1007/978-0-387-72067-8}
  {\path{doi:10.1007/978-0-387-72067-8}}.

\bibitem{Botti11}
L.~Botti, D.~A.~D. Pietro, A pressure-correction scheme for
  convection-dominated incompressible flows with discontinuous velocity and
  continuous pressure, Journal of Computational Physics 230~(3) (2011) 572 --
  585.
\newblock \href {http://dx.doi.org/10.1016/j.jcp.2010.10.004}
  {\path{doi:10.1016/j.jcp.2010.10.004}}.

\bibitem{Ferrer11}
E.~Ferrer, R.~H.~J. Willden, A high order discontinuous {G}alerkin finite
  element solver for the incompressible {N}avier--{S}tokes equations, Computers
  \& Fluids 46~(1) (2011) 224 -- 230.
\newblock \href {http://dx.doi.org/10.1016/j.compfluid.2010.10.018}
  {\path{doi:10.1016/j.compfluid.2010.10.018}}.

\bibitem{Klein13}
B.~Klein, F.~Kummer, M.~Oberlack, A {SIMPLE} based discontinuous {G}alerkin
  solver for steady incompressible flows, Journal of Computational Physics 237
  (2013) 235 -- 250.
\newblock \href {http://dx.doi.org/10.1016/j.jcp.2012.11.051}
  {\path{doi:10.1016/j.jcp.2012.11.051}}.

\bibitem{Steinmoeller13}
D.~T. Steinmoeller, M.~Stastna, K.~G. Lamb, A short note on the discontinuous
  {G}alerkin discretization of the pressure projection operator in
  incompressible flow, Journal of Computational Physics 251 (2013) 480 -- 486.
\newblock \href {http://dx.doi.org/10.1016/j.jcp.2013.05.036}
  {\path{doi:10.1016/j.jcp.2013.05.036}}.

\bibitem{Joshi16}
S.~M. Joshi, P.~J. Diamessis, D.~T. Steinmoeller, M.~Stastna, G.~N. Thomsen, A
  post-processing technique for stabilizing the discontinuous pressure
  projection operator in marginally-resolved incompressible inviscid flow,
  Computers \& Fluids 139 (2016) 120 -- 129, 13th USNCCM International
  Symposium of High-Order Methods for Computational Fluid Dynamics - A special
  issue dedicated to the 60th birthday of Professor David Kopriva.
\newblock \href {http://dx.doi.org/10.1016/j.compfluid.2016.04.021}
  {\path{doi:10.1016/j.compfluid.2016.04.021}}.

\bibitem{Krank17}
B.~Krank, N.~Fehn, W.~A. Wall, M.~Kronbichler, A high-order semi-explicit
  discontinuous {G}alerkin solver for 3{D} incompressible flow with application
  to {DNS} and {LES} of turbulent channel flow, Journal of Computational
  Physics 348 (2017) 634--659.
\newblock \href {http://dx.doi.org/10.1016/j.jcp.2017.07.039}
  {\path{doi:10.1016/j.jcp.2017.07.039}}.

\bibitem{Fehn18a}
N.~Fehn, W.~A. Wall, M.~Kronbichler, Robust and efficient discontinuous
  {G}alerkin methods for under-resolved turbulent incompressible flows, Journal
  of Computational Physics 372 (2018) 667--693.
\newblock \href {http://dx.doi.org/10.1016/j.jcp.2018.06.037}
  {\path{doi:10.1016/j.jcp.2018.06.037}}.

\bibitem{Lehrenfeld2016}
C.~Lehrenfeld, J.~Sch{\"o}berl, High order exactly divergence-free hybrid
  discontinuous {G}alerkin methods for unsteady incompressible flows, Computer
  Methods in Applied Mechanics and Engineering 307 (2016) 339 -- 361.
\newblock \href {http://dx.doi.org/10.1016/j.cma.2016.04.025}
  {\path{doi:10.1016/j.cma.2016.04.025}}.

\bibitem{Piatkowski2018}
M.~Piatkowski, S.~Müthing, P.~Bastian, A stable and high-order accurate
  discontinuous {G}alerkin based splitting method for the incompressible
  {N}avie--{S}tokes equations, Journal of Computational Physics 356 (2018) 220
  -- 239.
\newblock \href {http://dx.doi.org/10.1016/j.jcp.2017.11.035}
  {\path{doi:10.1016/j.jcp.2017.11.035}}.

\bibitem{Schroeder2018}
P.~W. Schroeder, G.~Lube, Divergence-free {H}(div)-{FEM} for time-dependent
  incompressible flows with applications to high {R}eynolds number vortex
  dynamics, Journal of Scientific Computing 75~(2) (2018) 830--858.
\newblock \href {http://dx.doi.org/10.1007/s10915-017-0561-1}
  {\path{doi:10.1007/s10915-017-0561-1}}.

\bibitem{Rhebergen2018}
S.~Rhebergen, G.~N. Wells, A hybridizable discontinuous {G}alerkin method for
  the {N}avier--{S}tokes equations with pointwise divergence-free velocity
  field, Journal of Scientific Computing 76~(3) (2018) 1484--1501.
\newblock \href {http://dx.doi.org/10.1007/s10915-018-0671-4}
  {\path{doi:10.1007/s10915-018-0671-4}}.

\bibitem{Schroeder2017}
P.~W. Schroeder, G.~Lube, Stabilised d{G}-{FEM} for incompressible natural
  convection flows with boundary and moving interior layers on non-adapted
  meshes, Journal of Computational Physics 335 (2017) 760 -- 779.
\newblock \href {http://dx.doi.org/10.1016/j.jcp.2017.01.055}
  {\path{doi:10.1016/j.jcp.2017.01.055}}.

\bibitem{Akbas2018}
M.~Akbas, A.~Linke, L.~G. Rebholz, P.~W. Schroeder, The analogue of grad--div
  stabilization in {DG} methods for incompressible flows: {L}imiting behavior
  and extension to tensor-product meshes, Computer Methods in Applied Mechanics
  and Engineering 341 (2018) 917 -- 938.
\newblock \href {http://dx.doi.org/10.1016/j.cma.2018.07.019}
  {\path{doi:10.1016/j.cma.2018.07.019}}.

\bibitem{Fehn19Hdiv}
N.~Fehn, M.~Kronbichler, C.~Lehrenfeld, G.~Lube, P.~W. Schroeder, High-order
  {DG} solvers for under-resolved turbulent incompressible flows: A comparison
  of $ {L}^{2}$ and $ {H} $(div) methods, International Journal for Numerical
  Methods in Fluids in press (2019) n/a.
\newblock \href {http://dx.doi.org/10.1002/fld.4763}
  {\path{doi:10.1002/fld.4763}}.

\bibitem{Ferrer12b}
E.~Ferrer, R.~H.~J. Willden, A high order discontinuous {G}alerkin –
  {F}ourier incompressible {3D} {N}avier–-{S}tokes solver with rotating
  sliding meshes, Journal of Computational Physics 231~(21) (2012) 7037 --
  7056.
\newblock \href {http://dx.doi.org/10.1016/j.jcp.2012.04.039}
  {\path{doi:10.1016/j.jcp.2012.04.039}}.

\bibitem{Wang2018}
Y.~Wang, A.~Quaini, S.~{\v{C}}ani{\'{c}}, A higher-order discontinuous
  {G}alerkin/arbitrary {L}agrangian {E}ulerian partitioned approach to solving
  fluid--structure interaction problems with incompressible, viscous fluids and
  elastic structures, Journal of Scientific Computing 76~(1) (2018) 481--520.
\newblock \href {http://dx.doi.org/10.1007/s10915-017-0629-y}
  {\path{doi:10.1007/s10915-017-0629-y}}.

\bibitem{Orszag1986}
S.~A. Orszag, M.~Israeli, M.~O. Deville, Boundary conditions for incompressible
  flows, J. Sci. Comput. 1~(1) (1986) 75--111.
\newblock \href {http://dx.doi.org/10.1007/BF01061454}
  {\path{doi:10.1007/BF01061454}}.

\bibitem{Karniadakis1991}
G.~E. Karniadakis, M.~Israeli, S.~A. Orszag, High-order splitting methods for
  the incompressible {N}avier--{S}tokes equations, Journal of Computational
  Physics 97~(2) (1991) 414 -- 443.
\newblock \href {http://dx.doi.org/10.1016/0021-9991(91)90007-8}
  {\path{doi:10.1016/0021-9991(91)90007-8}}.

\bibitem{Fehn17}
N.~Fehn, W.~A. Wall, M.~Kronbichler, On the stability of projection methods for
  the incompressible {N}avier--{S}tokes equations based on high-order
  discontinuous {G}alerkin discretizations, Journal of Computational Physics
  351 (2017) 392--421.
\newblock \href {http://dx.doi.org/10.1016/j.jcp.2017.09.031}
  {\path{doi:10.1016/j.jcp.2017.09.031}}.

\bibitem{Xu2019}
L.~Xu, X.~Xu, X.~Ren, Y.~Guo, Y.~Feng, X.~Yang, Stability evaluation of
  high-order splitting method for incompressible flow based on discontinuous
  velocity and continuous pressure, Advances in Mechanical Engineering 11~(10)
  (2019) 1687814019855586.
\newblock \href {http://dx.doi.org/10.1177/1687814019855586}
  {\path{doi:10.1177/1687814019855586}}.

\bibitem{Rhebergen2012}
S.~Rhebergen, B.~Cockburn, A space--time hybridizable discontinuous {G}alerkin
  method for incompressible flows on deforming domains, Journal of
  Computational Physics 231~(11) (2012) 4185 -- 4204.
\newblock \href {http://dx.doi.org/10.1016/j.jcp.2012.02.011}
  {\path{doi:10.1016/j.jcp.2012.02.011}}.

\bibitem{Horvath2019}
T.~L. Horv\'ath, S.~Rhebergen, A locally conservative and energy-stable
  finite-element method for the {N}avier--{S}tokes problem on time-dependent
  domains, International Journal for Numerical Methods in Fluids 89~(12) (2019)
  519--532.
\newblock \href {http://dx.doi.org/10.1002/fld.4707}
  {\path{doi:10.1002/fld.4707}}.

\bibitem{Nektar++}
C.~Cantwell, D.~Moxey, A.~Comerford, A.~Bolis, G.~Rocco, G.~Mengaldo, D.~D.
  Grazia, S.~Yakovlev, J.-E. Lombard, D.~Ekelschot, B.~Jordi, H.~Xu,
  Y.~Mohamied, C.~Eskilsson, B.~Nelson, P.~Vos, C.~Biotto, R.~Kirby,
  S.~Sherwin, Nektar++: An open-source spectral/hp element framework, Computer
  Physics Communications 192 (2015) 205 -- 219.
\newblock \href {http://dx.doi.org/10.1016/j.cpc.2015.02.008}
  {\path{doi:10.1016/j.cpc.2015.02.008}}.

\bibitem{Nek5000}
P.~Fischer, S.~Kerkemeier, A.~Peplinski, D.~Shaver, A.~Tomboulides, M.~Min,
  A.~Obabko, E.~Merzari, {NEK5000} {W}eb page, {https://nek5000.mcs.anl.gov}
  (2020).

\bibitem{Gassner2013}
G.~J. Gassner, A.~D. Beck, On the accuracy of high-order discretizations for
  underresolved turbulence simulations, Theoretical and Computational Fluid
  Dynamics 27~(3) (2013) 221--237.
\newblock \href {http://dx.doi.org/10.1007/s00162-011-0253-7}
  {\path{doi:10.1007/s00162-011-0253-7}}.

\bibitem{Moura2017}
R.~Moura, G.~Mengaldo, J.~Peir{\'o}, S.~Sherwin, On the eddy-resolving
  capability of high-order discontinuous {G}alerkin approaches to implicit
  {LES} / under-resolved {DNS} of {E}uler turbulence, Journal of Computational
  Physics 330 (2017) 615 -- 623.
\newblock \href {http://dx.doi.org/10.1016/j.jcp.2016.10.056}
  {\path{doi:10.1016/j.jcp.2016.10.056}}.

\bibitem{Fehn18b}
N.~Fehn, W.~A. Wall, M.~Kronbichler, Efficiency of high-performance
  discontinuous {G}alerkin spectral element methods for under-resolved
  turbulent incompressible flows, International Journal for Numerical Methods
  in Fluids 88~(1) (2018) 32--54.
\newblock \href {http://dx.doi.org/10.1002/fld.4511}
  {\path{doi:10.1002/fld.4511}}.

\bibitem{Kronbichler2018}
M.~Kronbichler, W.~A. Wall, A {P}erformance {C}omparison of {C}ontinuous and
  {D}iscontinuous {G}alerkin {M}ethods with {F}ast {M}ultigrid {S}olvers, SIAM
  Journal on Scientific Computing 40~(5) (2018) A3423--A3448.
\newblock \href {http://dx.doi.org/10.1137/16M110455X}
  {\path{doi:10.1137/16M110455X}}.

\bibitem{Kronbichler2019fast}
M.~Kronbichler, K.~Kormann, Fast matrix-free evaluation of discontinuous
  {G}alerkin finite element operators, ACM Trans. Math. Softw. 45~(3) (2019)
  29:1--29:40.
\newblock \href {http://dx.doi.org/10.1145/3325864}
  {\path{doi:10.1145/3325864}}.

\bibitem{Guermond06}
J.~L. Guermond, P.~Minev, J.~Shen, An overview of projection methods for
  incompressible flows, Computer Methods in Applied Mechanics and Engineering
  195~(44–47) (2006) 6011 -- 6045.
\newblock \href {http://dx.doi.org/10.1016/j.cma.2005.10.010}
  {\path{doi:10.1016/j.cma.2005.10.010}}.

\bibitem{Karniadakis13}
G.~E. Karniadakis, S.~J. Sherwin, Spectral/hp element methods for computational
  fluid dynamics, Oxford University Press, 2013.
\newblock \href {http://dx.doi.org/10.1093/acprof:oso/9780198528692.001.0001}
  {\path{doi:10.1093/acprof:oso/9780198528692.001.0001}}.

\bibitem{Chorin68}
A.~J. Chorin, Numerical solution of the {N}avier--{S}tokes equations, Math.
  Comp. 22~(104) (1968) 745--762.
\newblock \href {http://dx.doi.org/10.1090/S0025-5718-1968-0242392-2}
  {\path{doi:10.1090/S0025-5718-1968-0242392-2}}.

\bibitem{hirt1972}
C.~Hirt, J.~Cook, Calculating three-dimensional flows around structures and
  over rough terrain, Journal of Computational Physics 10~(2) (1972) 324--340.

\bibitem{Goda1979}
K.~Goda, A multistep technique with implicit difference schemes for calculating
  two-or three-dimensional cavity flows, Journal of Computational Physics
  30~(1) (1979) 76--95.

\bibitem{VanKan1986}
J.~Van~Kan, A second-order accurate pressure-correction scheme for viscous
  incompressible flow, SIAM Journal on Scientific and Statistical Computing
  7~(3) (1986) 870--891.

\bibitem{Timmermans1996}
L.~Timmermans, P.~Minev, F.~Van De~Vosse, An approximate projection scheme for
  incompressible flow using spectral elements, International Journal for
  Numerical Methods in Fluids 22~(7) (1996) 673--688.

\bibitem{Guermond2004}
J.~Guermond, J.~Shen, On the error estimates for the rotational
  pressure-correction projection methods, Mathematics of Computation 73~(248)
  (2004) 1719--1737.

\bibitem{Wang2008}
D.~Wang, S.~J. Ruuth, Variable step-size implicit-explicit linear multistep
  methods for time-dependent partial differential equations, Journal of
  Computational Mathematics (2008) 838--855.

\bibitem{EscobarVargas14}
J.~A. Escobar-Vargas, P.~J. Diamessis, T.~Sakai, A spectral quadrilateral
  multidomain penalty method model for high {R}eynolds number incompressible
  stratified flows, International Journal for Numerical Methods in Fluids
  75~(6) (2014) 403--425.
\newblock \href {http://dx.doi.org/10.1002/fld.3899}
  {\path{doi:10.1002/fld.3899}}.

\bibitem{Gao2017}
P.~Gao, J.~Ouyang, P.~Dai, W.~Zhou, A coupled continuous and discontinuous
  finite element method for the incompressible flows, International Journal for
  Numerical Methods in Fluids 84~(8) (2017) 477--493.
\newblock \href {http://dx.doi.org/10.1002/fld.4358}
  {\path{doi:10.1002/fld.4358}}.

\bibitem{Ferrer14}
E.~Ferrer, D.~Moxey, R.~H.~J. Willden, S.~J. Sherwin, Stability of projection
  methods for incompressible flows using high order pressure-velocity pairs of
  same degree: Continuous and discontinuous {G}alerkin formulations, Commun.
  Comput. Phys. 16 (2014) 817--840.
\newblock \href {http://dx.doi.org/10.4208/cicp.290114.170414a}
  {\path{doi:10.4208/cicp.290114.170414a}}.

\bibitem{Leriche2006}
E.~Leriche, E.~Perchat, G.~Labrosse, M.~O. Deville, Numerical evaluation of the
  accuracy and stability properties of high-order direct {S}tokes solvers with
  or without temporal splitting, J. Sci. Comput. 26~(1) (2006) 25--43.
\newblock \href {http://dx.doi.org/10.1007/s10915-004-4798-0}
  {\path{doi:10.1007/s10915-004-4798-0}}.

\bibitem{arnold2002unified}
D.~N. Arnold, F.~Brezzi, B.~Cockburn, L.~D. Marini, Unified analysis of
  discontinuous {G}alerkin methods for elliptic problems, SIAM Journal on
  Numerical Analysis 39~(5) (2002) 1749--1779.

\bibitem{Hillewaert13}
K.~Hillewaert, Development of the discontinuous {G}alerkin method for
  high-resolution, large scale {CFD} and acoustics in industrial geometries,
  Ph.D. thesis, Univ. de Louvain (2013).

\bibitem{Geuzaine2003}
P.~Geuzaine, C.~Grandmont, C.~Farhat, Design and analysis of {ALE} schemes with
  provable second-order time-accuracy for inviscid and viscous flow
  simulations, Journal of Computational Physics 191~(1) (2003) 206 -- 227.
\newblock \href {http://dx.doi.org/10.1016/S0021-9991(03)00311-5}
  {\path{doi:10.1016/S0021-9991(03)00311-5}}.

\bibitem{dealII90}
G.~Alzetta, D.~Arndt, W.~Bangerth, V.~Boddu, B.~Brands, D.~Davydov,
  R.~Gassmoeller, T.~Heister, L.~Heltai, K.~Kormann, M.~Kronbichler, M.~Maier,
  J.-P. Pelteret, B.~Turcksin, D.~Wells, \href{www.dealii.org}{The
  \texttt{deal.II} library, version 9.0}, J. Numer. Math. 26~(4) (2018)
  173--184.
\newblock \href {http://dx.doi.org/10.1515/jnma-2018-0054}
  {\path{doi:10.1515/jnma-2018-0054}}.
\newline\urlprefix\url{www.dealii.org}

\bibitem{KarniadakisSherwin2005}
G.~E. Karniadakis, S.~J. Sherwin, Spectral/hp element methods for computational
  fluid dynamics, 2nd Edition, Oxford University Press, 2005.

\bibitem{Taylor1937}
G.~Taylor, A.~Green, Mechanism of the production of small eddies from large
  ones, Proceedings of the Royal Society of London. Series A, Mathematical and
  Physical Sciences 158~(895) (1937) 499--521.

\end{thebibliography}

\end{document}